\documentclass[10pt,journal,compsoc]{IEEEtran}

\usepackage{color, soul}

\usepackage[utf8]{inputenc}
\usepackage[T1]{fontenc}
\usepackage{textcomp}     
\usepackage{mathptmx}  
\usepackage{fancyhdr}
\usepackage{times}
\usepackage{listings}
\usepackage{amssymb}
\usepackage{amsfonts}
\usepackage{amsmath}
\usepackage{array}
\usepackage{url}
\usepackage{pgfplots}

\usepackage{flushend}

\usepackage{epstopdf}
\usepackage{tikz}

\usepackage{multirow}
\usepackage{color, colortbl}

\usepackage{booktabs}
\usepackage{graphicx}
\usepackage{subcaption}
\usepackage{verbatim}
\usepackage[font=small,labelfont=bf]{caption}
\usepackage{ifsym}
\usepackage[algo2e,vlined,linesnumbered,ruled,noend]{algorithm2e}   
\usepackage{relsize}
\usetikzlibrary{shapes,arrows}
\usetikzlibrary{patterns}
\newcommand{\removelatexerror}{\let\@latex@error\@gobble}

\usepackage{pifont}
\newcommand{\cmark}{\ding{51}}%
\newcommand{\xmark}{\ding{55}}%

\usetikzlibrary{shapes,positioning,fit,calc}
\usetikzlibrary{snakes}
\usetikzlibrary{trees}
\usetikzlibrary{arrows}

\tikzset{
  treenode/.style = {align=center, inner sep=0pt, text centered,
  font=\sffamily},
  arn_n/.style = {treenode, circle,black, draw=black, text width=1.3em},%
  level distance = 0.5cm
}

\def\bsq#1{\textquotesingle #1\textquotesingle}

\DeclareCaptionType{copyrightbox}

\newcommand{\kat}[1]{ {\mbox{\tiny #1}} }

\newcommand{\LJoin}{\tiny\textifsym{d|><|}}

\newcolumntype{M}[1]{>{\centering\arraybackslash}m{#1}}
\newcommand\VRule[1][\arrayrulewidth]{\vrule width #1}
\newcolumntype{C}[1]{>{\centering\let\newline\\\arraybackslash\hspace{0pt}}m{#1}}

\definecolor{c1}{RGB}{75,115,75}
\definecolor{c2}{RGB}{60,116,131}
\definecolor{c3}{RGB}{157,50,50}
\definecolor{c4}{RGB}{0,177,179}
\definecolor{c5}{RGB}{212,175,55}
\definecolor{c6}{RGB}{218, 139, 139}
\definecolor{grey}{RGB}{169,169,169}

\newtheorem{definition}{Definition}
\newtheorem{theorem}{Theorem}
\newtheorem{proof}{Proof}

\newtheorem{example}{Example}
\newtheorem{proposition}{Proposition}
\newtheorem{corollary}{Corollary}

\usepackage{todonotes}

\usepackage[inline]{enumitem}

 \newcommand{\squishlist}{
 \begin{list}{$\bullet$}
  { \setlength{\itemsep}{0pt}
     \setlength{\parsep}{1pt}
     \setlength{\topsep}{1pt}
     \setlength{\partopsep}{0pt}
     \setlength{\leftmargin}{1.5em}
     \setlength{\labelwidth}{1em}
     \setlength{\labelsep}{0.5em} } }
 \newcommand{\squishend}{\end{list}}

\title{Lineage-Aware Temporal Windows: Supporting Set Operations in Temporal-Probabilistic Databases}

\selectfont\ttfamily\upshape

\author{\IEEEauthorblockN{Katerina Papaioannou\IEEEauthorrefmark{1},
Martin Theobald\IEEEauthorrefmark{2},
Michael H. B{\"o}hlen\IEEEauthorrefmark{1}} \\
\vspace*{0.15cm}
\IEEEauthorblockA{\IEEEauthorrefmark{1} \em Department of Computer Science, 
University of Zurich \\ $\mathtt{\{papaioannou,boehlen\}@ifi.uzh.ch}$}\\
\vspace*{0.15cm}
\IEEEauthorblockA{\IEEEauthorrefmark{2} \em Computer Science \& Communications Research Unit, University of Luxembourg $\mathtt{martin.theobald@uni.lu}$ }}

\begin{document}

\maketitle

\begin{abstract}


In temporal-probabilistic (TP) databases, the combination of the temporal and the probabilistic dimension adds significant overhead to the computation of set operations. Although set queries are guaranteed to yield linearly sized output relations, existing solutions exhibit quadratic runtime complexity. They suffer from redundant interval comparisons and additional joins for the formation of lineage expressions. In this paper, we formally define the semantics of set operations in TP databases and study their properties. For their efficient computation, we introduce the {\em lineage-aware temporal window}, a mechanism that directly binds intervals with lineage expressions.  We suggest the {\em lineage-aware window advancer} (LAWA) for producing the windows of two TP relations in linearithmic time, and we implement all TP set operations based on LAWA. By exploiting the flexibility of lineage-aware temporal windows, we perform direct filtering of irrelevant intervals and finalization of output lineage expressions and thus guarantee that no additional computational cost or buffer space is needed. A series of experiments over both synthetic and real-world datasets show that (a) our approach has predictable performance, depending only on the input size and not on the number of time intervals per fact or their overlap, and that (b) it outperforms state-of-the-art approaches in both temporal and probabilistic databases.

\end{abstract}

\section{Introduction}
\label{sec:intro}

The need to manage large, temporal-probabilistic (TP) datasets appears
in a wide range of applications, such as temporal predictions (e.g.,
weather) as well as in sensor (e.g., RFID) and other forms of scientific data, which are inherently temporal and frequently contain
erroneous measurements. The combination of the
temporal and the probabilistic dimension in a relational database setting requires that
the result of the relational algebraic operators complies with the
semantics of each dimension. To this end, probabilistic databases rely on the {\em
possible-worlds semantics} to define for which instances of the
probabilistic database an answer tuple is valid.  Conversely, temporal 
databases use the {\em sequenced semantics} to define at which time
points (i.e., snapshots of the temporal database) an answer tuple is
valid. The possible-worlds and the sequenced semantics very nicely
complement each other, since they both employ the notion of {\em data
lineage} to guarantee a closed and complete representation model for
temporal, uncertain data.

In this paper, we introduce a {\em sequenced} TP data model and,
under this model, we define and implement the three principle TP set operations, {\em intersection} ($\cap^\kat{Tp}$), {\em union} ($\cup^\kat{Tp}$) and {\em difference} ($-^\kat{Tp}$)\footnote{Note that, although in a relational setting {\em intersection} is a dependent operation which can be expressed in terms of two {\em difference} operations, we show that considering {\em intersection} as a separate operator has significant performance advantages in a TP setting.}. In the
following example, we illustrate the usefulness of TP set operators in an application involving temporal-probabilistic predictions.

\begin{figure*}[!ht]\centering
  \setlength{\tabcolsep}{5pt}
  \def\arraystretch{1.2}
  \begin{subfigure}[b]{0.75\linewidth} \centering
    \scalebox{0.95}{
      \begin{tabular}{c|c|c|c}
        \multicolumn{4}{l}{$\mathbf{a}$ \textbf{(productsBought)}} \\
        \hline
        \textit{Product} & $\lambda$ & $T$ & $p$ \\ \hline
        \bsq{milk} & {$a_1$} & {[2,10)} & {0.3} \\
        \bsq{chips}  & {$a_2$} & {[4,7)}  & {0.8} \\
        \bsq{dates}  & {$a_3$} & {[1,3)} & {0.6} \\ \hline

      \end{tabular}}
%
    \qquad
    \scalebox{0.95}{
      \begin{tabular}{c|c|c|c}
        \multicolumn{4}{l}{$\mathbf{b}$ \textbf{(productsOrdered)}}\\
        \hline
        \textit{Product} & $\lambda$ & $T$ & $p$ \\ \hline
        \bsq{milk} & $b_1$ & {[5,9)} & {0.6} \\
        \bsq{chips} & $b_2$  & {[3,6)}  & {0.9} \\ \hline
      \end{tabular}}
%
    \qquad
    \scalebox{0.95}{
      \begin{tabular}{c|c|c|c}
        \multicolumn{4}{l}{$\mathbf{c}$ \textbf{(productsInStock)}} \\
        \hline
        \textit{Product} & $\lambda$ & $T$ & $p$ \\ \hline
        \bsq{milk} & {$c_1$} & {[1,4)} & {0.6} \\
        \bsq{milk} & {$c_2$} & {[6,8)} & {0.7} \\
        \bsq{chips} & $c_3$ & [4,5) & {0.7} \\
        \bsq{chips} & $c_4$ & [7,9) & {0.8} \\ \hline
      \end{tabular}}
    \vspace*{0.1cm}
    \caption{Input Relations}
    \label{fig:tpdbRelations}
  \end{subfigure}
\qquad
    \begin{subfigure}[b]{0.2\linewidth}\centering
      \tikzstyle{bag} = [text width=5em, text centered]
      \scalebox{0.8}{
        \begin{tikzpicture}[sloped]
          \tikzstyle{level 1}=[level distance=1cm, sibling distance=1.7cm]
          \tikzstyle{level 2}=[level distance=1cm, sibling distance=1.5cm]
          \tikzstyle{level 3}=[level distance=1cm, sibling distance=1cm]
          \tikzstyle{level 3}=[level distance=1cm, sibling distance=2cm]
          \node[bag] {$\mathbf{-}^\kat{Tp}$}
	    child {
              node[bag] {\large $\mathbf{c}$}
            }
            child {
              node[bag] {\large{$\mathbf{\cup}^\kat{Tp}$}}
                child {
                  node[bag] {\large $\mathbf{a}$}
                }
                child {
                  node[bag] {\large $\mathbf{b}$}
                }
            };
          \end{tikzpicture}}
      \vspace*{3mm}
      \caption{Query}
      \label{fig:tpdbQueryPlan}
    \end{subfigure}

    \vspace*{0.8cm}
    
  \center
    \begin{subfigure}[b]{0.28\linewidth} \centering
      \scalebox{0.95}{
        \begin{tabular}{c|c|c|c}
          \hline 
          \textit{F}  & $\lambda_r$ & $\lambda_s$  & $T$   \\ \hline
          \bsq{milk}  & {$a_1$} & $\mathtt{null}$  & [2,5) \\
          \bsq{milk}  & {$a_1$} & {$b_1$}		   & [5,9) \\
          \bsq{milk}  & {$a_1$} & $\mathtt{null}$  & [9,10) \\
          \bsq{chips} & $\mathtt{null}$ & {$b_2$}  & [3,4) \\
          \bsq{chips} & {$a_2$} & {$b_2$}          & [4,6) \\
          \bsq{chips} & {$a_2$} & $\mathtt{null}$  & [6,7) \\ 
          \bsq{dates} & {$a_3$} & $\mathtt{null}$  & [1,3) \\ [0.1cm] 
          \hline
        \end{tabular}}
      \vspace*{3mm}
      \caption{${\bf W}({\bf a}, {\bf b})$}
      \label{fig:latwAB}
    \end{subfigure}
\qquad
    \begin{subfigure}[b]{0.28\linewidth} \centering
      \scalebox{0.95}{
        \begin{tabular}{c|c|c|c}
          \hline 
          \textit{F}  & $\lambda_r$ & $\lambda_s$  & $T$   \\ \hline
          \bsq{milk}  & {$c_1$} & $\mathtt{null}$  & [1,2) \\
          \bsq{milk}  & {$c_1$} & {$a_1$}		   & [2,4) \\
\rowcolor{c6} \bsq{milk}  & $\mathtt{null}$ & {$a_1$}  & [4,5) \\
\rowcolor{c6} \bsq{milk}  & $\mathtt{null}$ & {$a_1 \lor b_1$} & [5,6) \\          
          \bsq{milk}  & {$c_2$} & {$a_1 \lor b_1$} & [6,8) \\
\rowcolor{c6} \bsq{milk}  & $\mathtt{null}$ & {$a_1 \lor b_1$} & [8,9) \\                    
          \bsq{chips} & {$c_3$} & {$a_2 \lor b_2$} & [4,5) \\
\rowcolor{c6} \bsq{milk}  & $\mathtt{null}$ & {$a_2 \lor b_2$} & [5,6) \\                    
\rowcolor{c6} \bsq{chips} & $\mathtt{null}$ & {$a_2$} & [6,7) \\ 
          \bsq{chips} & {$c_4$} & $\mathtt{null}$  & [7,9) \\ 
          \bsq{dates} & {$a_3$} & $\mathtt{null}$  & [1,3) \\ [0.1cm] 
          \hline
        \end{tabular}}
      \vspace*{3mm}
      \caption{${\bf W}({\bf c}, {\bf a \ \cup^\kat{Tp}\ \bf b})$}
      \label{fig:latwABC}
    \end{subfigure}
\qquad
    \begin{subfigure}[b]{0.33\linewidth} \centering
      \scalebox{0.95}{
        \begin{tabular}{c|c|c|c}
          \hline 
          \textit{Product} & $\lambda$ & $T$ & $p$ \\ \hline
          \bsq{milk} & {$c_1$} & [1,2) & {0.6} \\
          \bsq{milk} & {$c_1 \land \neg a_1$} & [2,4) & {0.42} \\
          \bsq{milk} & {$c_2 \land \neg (a_1 \lor b_1) $} & [6,8) & {0.196} \\
          \bsq{chips} & {$c_3 \land \neg (a_2 \lor b_2) $} & [4,5) & {0.014} \\
          \bsq{chips} & $c_4$ & [7,9) & {0.8} \\ [0.1cm] 
          \hline
        \end{tabular}}
      \vspace*{3mm}
      \caption{Query Result}
      \label{fig:tpdbQueryResult}
    \end{subfigure}
  \caption{The Supermarket Application Scenario}
  \label{fig:tpdb}
\end{figure*}

\begin{example} {
  Consider the supermarket application of Figure~\ref{fig:tpdb}.  The
  supermarket records data related to purchases of clients
  ($\mathbf{a}$), online shopping carts ($\mathbf{b}$), and inventory
  ($\mathbf{c}$).  At each time point (e.g., a day), the supermarket
  aims at predicting the products that clients want to buy or order
  versus those that it has in stock.  For example, the tuple
  \emph{(\bsq{milk}, $a_1$, [2,10), 0.3)} captures that, at each day
  from the $2^{nd}$ to the $10^{th}$ of the month, ``milk is bought"
  with probability {\it 0.3}.  There is a single prediction for each
  fact at each time point and thus, there is no other tuple in
  $\mathbf{a}$ that predicts the probability of buying \bsq{milk} over
  an interval overlapping with $[2,10)$. 
  
  In order to have an overview of its supply and demand, the supermarket
  wants to determine, at each time point, the probability that a
  product is in stock but no client wants to order or buy this
  product.  The corresponding query is
  $Q = \mathbf{c}-^\kat{Tp} (\mathbf{a} \cup^\kat{Tp} \mathbf{b})$,
  i.e., the union of relations $\mathbf{a}$ and $\mathbf{b}$, followed
  by a difference with relation $\mathbf{c}$ (see
  Fig.~\ref{fig:tpdbQueryPlan}). Answer tuple \emph{(\bsq{milk},
    $c_1 \land \neg a_1$, [2,4), 0.42)} (see
  Fig.~\ref{fig:tpdbQueryResult}) expresses that, with probability
  {\it 0.42}, \emph{\bsq{milk}} is in stock but is not ordered or
  bought during interval $[2,4)$. The lineage expression used for the
  computation of the interval and the probability of this tuple is formed
  based on the tuples of the input relations which are valid at each
  time point ($c_1$ and $a_1$) and the semantics of the operation to
  be computed ($\cup^\kat{Tp}$ and $-^\kat{Tp}$).  }
\end{example}

TP set operations are interesting because of the overhead added in
their computation when combining the temporal and probabilistic
dimension. They are however a class of operations that have received
little attention so far: they have not been explicitly defined in
existing TP approaches~\cite{DyllaMT13}, with TP set difference not
being supported at all. Existing temporal techniques suffer from two
main drawbacks. First, approaches used for the computation of temporal
set operations \cite{DignosTODS16, DignosBG12} replicate input tuples
with adjusted intervals before the actual algebraic operations are
applied. They rely on joins with inequality conditions that have
quadratic complexity due to unproductive tuple comparisons.  Second,
stitching lineage expressions to the output tuples in a relational
manner requires additional joins in comparison to the set operations
that are available in current temporal database implementations.
Existing probabilistic approaches \cite{FinkOR11}, on the other hand,
reduce set operations to joins, since their computation not only requires
the comparison of relational attributes among the input tuples, but also
the combination of their lineage expressions. However, the computation
of TP set operations under a sequenced TP data model requires more
sophisticated solutions for the computation of output intervals than
the use of temporal predicates in joins.

In this paper, we introduce the concept of a {\em lineage-aware
temporal window} as a means to combine the computation of the output
intervals and the computation of the input lineage expressions that
will contribute to an output tuple.  The set of all windows of two
TP relations constitutes a common core based on which we can produce
the result of any TP set operation by using appropriate filter and
concatenation functions. Based on this approach, we develop efficient
algorithms for the computation of windows, and we eliminate
redundancies in the steps that existing approaches need to rely on
to identify the input tuples contributing to an output tuple.

\begin{example}
  { In order to compute the query of Fig.~\ref{fig:tpdbQueryPlan},
    we need to first compute the {set of lineage-aware temporal windows}
    ${\bf W}({\bf a}, {\bf b})$ of relations $\mathbf{a}$ and $\mathbf{b}$
    (Fig.~\ref{fig:latwAB}) to compute their union. Each window spans a
    maximal interval over which a set of non-temporal attributes, called
    a ``fact", is included in the same input tuples. The window $w$ = 
    (\bsq{milk}, $a_1$, $b_1$, [5,9)) indicates that at each time point
    in [5,9), the fact \bsq{milk} is included in the tuple of $\mathbf{a}$
    with lineage $a_1$ and in the tuple of $\mathbf{b}$ with lineage $b_1$.
    In the result of a TP union, an output tuple is created when at least
    one of the input tuples is valid, and the windows of $\mathbf{a}$ and
    $\mathbf{b}$ form output tuples by using a disjunction of the input
    lineages. Thus, window $w$ is transformed into output tuple
    (\bsq{milk}, $a_1 \lor b_1$, [5,9), 0.72). For the computation of the
    set difference $\mathbf{c}-^\kat{Tp} (\mathbf{a} \cup^\kat{Tp}
    \mathbf{b})$, the {lineage-aware temporal windows} of relations
    $\mathbf{c}$ and $\mathbf{a} \cup^\kat{Tp} \mathbf{b}$ are computed as
    shown in Fig.~\ref{fig:latwABC}.  The window (\bsq{milk}, $c_2$,
    $a_1\lor b_1$, [6,8)) indicates that at each time point in [6,8), the
    fact \bsq{milk} is included in the tuple with lineage $c_2$ from input
    relation $\mathbf{c}$, while the tuple with lineage $a_1\lor b_1$ is
    included from $\mathbf{a} \cup^\kat{Tp} \mathbf{b}$, respectively. Note
    that the windows of Fig.~\ref{fig:latwABC} which are highlighted in red
    are not included in the final output of the TP set difference, since
    there is no valid tuple in the left input relation.  An output tuple is
    created for each of the remaining {lineage-aware temporal windows} by
    concatenating the lineage expressions $\lambda_r$ and $\lambda_s$ to
    $\lambda_r \land \lnot \lambda_s$.  }
\end{example}

\smallskip
\noindent\textbf{Outline \& Contributions.}
\begin{itemize}
\item We propose a {\em sequenced temporal-probabilistic data
    model} that complies with both the sequenced semantics from
  temporal databases \cite{DignosBG12, BohJen2009} and the
  possible-worlds semantics from probabilistic databases
  \cite{Suciu2009, SuciuKochOlteanuReBook}.
  
\item We formally define the semantics of {\em TP set operations} and study the
  properties of TP set queries under this model.  TP set queries have
  not previously been investigated under a sequenced
  temporal-probabilistic model.
  
\item We introduce the concept of {\em lineage-aware temporal windows}, a
    mechanism that binds an interval with the lineages of the tuples
    that are valid during the interval. We show that each output tuple of a TP
    set operation maps to exactly one window, and we reduce the
    computation of a TP set operation between two TP relations to the
    application of conventional selection and projection operations over their sets
    of {\em lineage-aware temporal windows}.

\item We introduce the {\em lineage-aware window advancer} (LAWA), a
  window-sweeping algorithm that computes all {\em lineage-aware
    temporal windows} of two TP relations and guarantees
  $O(n \, \log n)$ worst-case complexity. Exploiting the flexibility
  of the windows, we are able to finalize lineages and filter out
  irrelevant intervals directly at the time of their creation. No
  additional costs are involved and thus the computation of a TP set
  operation has linearithmic complexity, improving over existing
  implementations with quadratic complexity.
  
\item We experimentally demonstrate that LAWA is the only approach
  that {\em does not deteriorate in performance} as the data history
  grows. In contrast to existing techniques, our solution does not
  depend on the characteristics of the dataset (such as the number of
  intervals per fact, or the overlap among intervals), but only on the
  size of the input relations.
\end{itemize}

This paper is an extension of our ICDE paper~\cite{papaioannou2018} and
it is organized as follows. Section \ref{sec:related-work} provides an
overview of related works on temporal and probabilistic databases with
a focus on set operations. Section \ref{sec:preliminaries} introduces
our TP data model, while Section \ref{sec:querySemantics} defines the
model's query semantics. Section \ref{sec:tpsetOps} defines TP set
operations over duplicate-free input relations.
Section\ref{sec:windows} introduces lineage-aware temporal windows.
Section \ref{sec:sweepingIUE} introduces an algorithm for the
computation of lineage-aware temporal windows, and
Section \ref{sec:basicIUE} includes our implementation of TP set
operations. Section \ref{sec:exper-eval} presents a comprehensive
performance study that compares our implementation of TP set
operations with existing timestamp-adjustment and lineage-computation
approaches. Section \ref{sec:conclusion} concludes the paper.

\section{Related Work}
\label{sec:related-work}

We next review related approaches from both temporal and probabilistic
databases and explain their limitations in terms of supporting TP set
operations. Set difference, for example, has received little attention
in temporal databases and can only be computed using the generic
normalization operator \cite{DignosBG12}.  Under a combined temporal
and probabilistic data model, there is currently no solution that
supports set difference.

\smallskip \noindent\textbf{Temporal Set Operations.}  In temporal
databases, the result of a temporal set operation $op^T$ is defined as
the result of applying $op$ over a sequence of atemporal instances
(the so-called snapshots) of the input relations---a key concept in
temporal databases termed {\em snapshot reducibility}
\cite{AlKatebGCBCP13,lorentzos1997sql,Viqueira2007}.  Maximal
intervals are produced by merging consecutive time points to which the
same input tuples have contributed ({\em change preservation}).
Dign\"os et al.~\cite{DignosBG12, DignosTODS16} use {\em data lineage}
to guarantee change preservation for all relational operations under a
sequenced semantics.  They adapt the {\em Normalization} operator,
introduced by Toman et al.~\cite{toman1998point}, to compute temporal
set queries.  Intuitively, the normalization $N({\bf r}, {\bf s})$ of
a relation ${\bf r}$ based on another relation ${\bf s}$ replicates
the tuples of ${\bf r}$ and assigns new time intervals to them.  The
new intervals are obtained by splitting the original intervals based
on tuples of ${\bf s}$ with which they overlap.  Normalization is a
generic operator that subsequently requires an outer join of {\bf r}
and {\bf s} with quadratic complexity.  Since it is not symmetric, it
has to be computed once for each of the two input relations
\cite{DignosBG12,DignosTODS16} for the computation of temporal
set-operations (cf.\ Fig.~\ref{fig:temporalDBs}).

\begin{figure}[!h]
\tikzset{%
  algebra/.style    = {draw, thick, rectangle, minimum height = 2em,
    minimum width = 2em},
}
\tikzset{%
  lineage/.style    = {draw, thick, rectangle, minimum height = 2em,
    minimum width = 2.3cm},
}
\tikzset{%
  filter/.style    = {draw, thick, rectangle, minimum height = 2em,
    minimum width = 2.75cm},
}
\tikzset{%
  block/.style    = {draw, thick, rectangle, minimum height = 2em,
    minimum width = 2em},
}
\centering
\scalebox{0.65}{
\centering
\begin{tikzpicture}[auto, thick, node distance=2cm, >=triangle 45]
\draw
	node at (0,0)[right=-3mm] (relR) {\large $\mathbf{r}$}
	node [right= 0.5cm of relR] (relR_line1) {}

	node [above right = 0.3cm and 0.8cm of relR_line1, block] (relNR)
		 {\large {\bf N}($\mathbf{r}$, $\mathbf{s}$)}
 	node [above = 0.05cm of relNR.east] (relNR1) {}
	node [below = 0.05cm of relNR.east] (relNR2) {}
 	node [above = 0.05cm of relNR.west] (relNR3) {}
	node [below = 0.05cm of relNR.west] (relNR4) {}
	
	node [below = 0.2cm of relR] (relS) {\large $\mathbf{s}$}
	node [right= 0.25cm of relS] (relS_line1) {}

	node [below right = 0.3cm and 1.05cm of relS_line1, block] (relNS) 
		 {\large {\bf N}($\mathbf{s}$, $\mathbf{r}$)}
 	node [above = 0.05cm of relNS.east] (relNS1) {}
	node [below = 0.05cm of relNS.east] (relNS2) {}
 	node [above = 0.05cm of relNS.west] (relNS3) {}
	node [below = 0.05cm of relNS.west] (relNS4) {}

	node [above right = 0.5cm and 3cm of relNR.north, algebra] (setD)
	{\large \color{c2}{$-$}}
	node [above = 0.05cm of setD.west] (setD1) {}
	node [below left = 0.05cm and 1.2cm of setD.west] (setD11) {}
	node [below = 0.05cm of setD.west] (setD2) {}

	node [below = 1.6cm of setD, algebra] (setI)
	{\large \color{c1}{$\cap$}}
	node [above left = 0.05cm and 0.5cm of setI.west] (setI11) {}
	node [below left = 0.05cm and 0.5cm of setI.west] (setI22) {}
	node [above = 0.05cm of setI.west] (setI1) {}
	node [below = 0.05cm of setI.west] (setI2) {}

	node [below = 1.6cm of setI, algebra] (setU) {\large \color{c3}{$\cup$}}
	node [above = 0.05cm of setU.west] (setU1) {}
	node [left = 0.75cm of setU1] (setU11) {}
	node [below = 0.05cm of setU.west] (setU2) {}
	
	node [right =1.5cm of setU] (resU)
	{\large \color{c3}{$ \mathbf{r} \cup^\kat{T} \mathbf{s}$}}
	node [right =1.5cm of setI] (resI)
	{\large \color{c1}{$ \mathbf{r} \cap^\kat{T} \mathbf{s}$}}
	node [right =1.5cm of setD] (resD)
	{\large \color{c2}{$ \mathbf{r} -^\kat{T} \mathbf{s}$}};

\draw[->] (relR.north) |- (relNR3.center);
\draw[-] (relR.east) -| (relS_line1.center);
\draw[->] (relS_line1.center) |- (relNS3.center);

\draw[->] (relS.south) |- (relNS4.center);
\draw[-] (relS.east) -| (relR_line1.center);
\draw[->] (relR_line1.center) |- (relNR4.center);

\draw[->] (relNR.north) |- (setD1.center);
\draw[-] (relNS1.center) -| (setD11.center);
\draw[->] (setD11.center) |- (setD2.center);

\draw[->] (relNS.south) |- (setU2.center);

\draw[-] (relNS2.center) -| (setI22.south);
\draw[->] (setI22.south) |- (setI2.center);
\draw[-] (relNR1.center) -| (setI11.north);
\draw[->] (setI11.north) |- (setI1.center);

\draw[-] (relNR2.center) -| (setU11.center);
\draw[->] (setU11.center) |- (setU1.center);

\draw[->] (setI) -- (resI);
\draw[->] (setD) -- (resD);
\draw[->] (setU) -- (resU);

\end{tikzpicture}}
\caption{Temporal set operations using Normalize \textbf{N}.}
\label{fig:temporalDBs}
\end{figure}
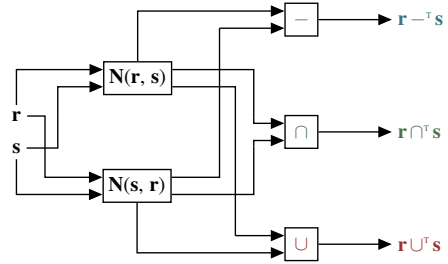

{\em Temporal joins} can be used for the computation of TP set
intersection.  Efficient solutions for temporal joins have been widely
discussed in the literature~\cite{KaufmannTI, DignosBG14,
PlatovICDE16, CafagnaB17}.  Specific solutions either partition the
data~\cite{CafagnaB17} in ways that are not beneficial for our case,
since TP relations are duplicate-free (see
Section~\ref{sec:preliminaries}), or they require fixed-length input
schemas~\cite{PlatovICDE16}. \emph{Timeline Index} (TI) is a data
structure introduced by Kaufmann et al.~\cite{KaufmannTI,
KaufmannVFKF13} to efficiently compute temporal aggregation, join and
time-travel operations. TI of relation $\mathbf{r}$ maps each start or
end point in $\mathbf{r}$ to a list of ids of tuples that start or end
at this time point.  \emph{Timeline Join} (TJ) is applied on the
indexes created for the input relations and implements a combination
of a merge- and a hash-join.  The performance of TJ suffers because
the original tuples need to be fetched both for the application of a
filtering condition and for the creation of the output tuples.

{\em Overlap Interval Partitioning} (OIP) by Dign\"os et al.\
\cite{DignosBG14} is designed to compute a join
$\mathbf{r} \bowtie^T \mathbf{s}$ among tuples with overlapping time
intervals. Initially, OIP splits the time domain into $k$ granules of
equal size.  Adjacent granules are combined to form the partitions of
an input relation $\mathbf{r}$ so that each tuple in $\mathbf{r}$ is
assigned to the smallest partition into which it fits. In order to
compute the overlap join, the overlapping partitions of $\mathbf{r}$
and $\mathbf{s}$ are identified (fast), and then a nested loop is
performed to join the tuples of these partitions (slow).  This
approach finds all pairs of tuples ($r$, $s$), for $r \in \mathbf{r}$
and $s \in \mathbf{s}$, with overlapping time intervals. Although OIP
can be extended to apply additional filtering conditions, e.g.,
equality conditions on the atemporal attributes of the tuples that are
joined, its performance deteriorates when the condition has low
selectivity (see Section \ref{sec:exper-eval}).

{\em Sweeping-based approaches}, finally, have been widely used for
the computation of overlap joins~\cite{PlatovICDE16, Arge1998} in
temporal settings.  A sweepline moves over all start and end points of
tuples, and determines, for each time point, the tuples of both input
relations that are valid.  These approaches cannot directly be applied
for the computation of TP set operations.  First, they generally do
not consider join conditions on the non-temporal attributes. Second,
they support set intersection but cannot produce all output tuples
needed for set difference and union. The creation of output intervals
through the tuples that the sweepline intersects is not sufficient for
these two set operations.

\smallskip\noindent\textbf{Probabilistic Set Operations.} In
probabilistic databases, the result of a probabilistic set operation
$op^p$ is defined as the result of applying $op$ over the set of all
possible instances of the input relations. The Trio system
\cite{SarmaTW08} was among the first to recognize {\em data lineage},
in the form of a Boolean formula, as a means to capture the possible
instances at which an output tuple is valid.  In an effort to provide a
{\em closed and complete} representation model for uncertain
relational data, they introduced {\em Uncertainty and Lineage
  Databases} (ULDBs) \cite{DBLP:journals/vldb/BenjellounSHTW08}.  The
algebraic operators are modified to compute the lineage of the result
tuples in a ULDB, thus capturing all information needed for computing
query answers and their probabilities.  Recently, Fink et
al.~\cite{FinkOR11,FinkO16} reduced the computation of probabilistic
algebraic operations to conventional operations (cf.\
  Fig.~\ref{fig:probabilisticDB}) so that these can be performed
using a DBMS, rather than by an application layer built on top of it.

\begin{figure}[!h]
\tikzset{%
  algebra/.style    = {draw, thick, rectangle, minimum height = 2em,
    minimum width = 2em},
}
\tikzset{%
  lineage/.style    = {draw, thick, rectangle, minimum height = 2em,
    minimum width = 2.3cm},
}
\tikzset{%
  filter/.style    = {draw, thick, rectangle, minimum height = 2em,
    minimum width = 2.75cm},
}
\tikzset{%
  block/.style    = {draw, thick, rectangle, minimum height = 2em,
    minimum width = 2em},
}
\centering
\scalebox{0.65}{
\begin{tikzpicture}[auto, thick, node distance=2cm, >=triangle 45]
\draw
	node at (0,0)[right=-3mm] (relR) {\large $\mathbf{r}$}
 	node [above = 0.05cm of relR.east] (relR1) {}
	node [below = 0.05cm of relR.east] (relR2) {}
	node [right= 0.5cm of relR] (relR_line1) {}
	
	node [below = 0.4cm of relR] (relS) {\large $\mathbf{s}$}
 	node [above = 0.05cm of relS.east] (relS1) {}
	node [below = 0.05cm of relS.east] (relS2) {}
	node [right= 0.25cm of relS] (relS_line1) {}

	node [above right = 0.5cm and 3cm of relR.north, algebra] (setD)
	{\large \color{c2}{$\LJoin$}}
	node [above = 0.05cm of setD.west] (setD1) {}
	node [below left = 0.05cm and 1.2cm of setD.west] (setD11) {}
	node [below = 0.05cm of setD.west] (setD2) {}
	
	node [below = 0.9cm of setD, algebra] (setI)
	{\large \color{c1}{$\bowtie$}}
	node [above left = 0.05cm and 0.5cm of setI.west] (setI11) {}
	node [below left = 0.05cm and 0.5cm of setI.west] (setI22) {}
	node [above = 0.05cm of setI.west] (setI1) {}
	node [below = 0.05cm of setI.west] (setI2) {}

	node [below = 0.9cm of setI, algebra] (setU) {\large \color{c3}{$\cup$}}
	node [above = 0.05cm of setU.west] (setU1) {}
	node [left = 0.75cm of setU1] (setU11) {}
	node [below = 0.05cm of setU.west] (setU2) {}

	node [right = 0.85cm of setI, lineage] (setIL) {\large \color{c1}{and($\lambda_r$, $\lambda_s$)}}
	node [right = 0.85cm of setD, lineage] (setDL) {\large \color{c2}{andNot($\lambda_r$, $\lambda_s$)}}
	node [right = 0.85cm of setU, lineage,white] (setUL) {}
	
	node [right = 0.6cm of setUL, lineage] (dupU) {\large $\vartheta$\color{c3}{or($\lambda$)}}

	node [right =0.8cm of dupU] (resU) {\large \color{c3}{$ \mathbf{r} \cup^\kat{p} \mathbf{s}$}}
	node [right =3.7cm of setIL] (resI) {\large \color{c1}{$ \mathbf{r} \cap^\kat{p} \mathbf{s}$}}
	node [right =3.1cm of setDL] (resD) {\large \color{c2}{$ \mathbf{r} -^\kat{p} \mathbf{s}$}};


%
%
%
\draw[->] (relR.north) |- (setD1.center);
\draw[-] (relS1.center) -| (setD11.center);
\draw[->] (setD11.center) |- (setD2.center);

\draw[->] (relS.south) |- (setU2.center);

\draw[-] (relS2.center) -| (setI22.south);
\draw[->] (setI22.south) |- (setI2.center);
\draw[-] (relR1.center) -| (setI11.north);
\draw[->] (setI11.north) |- (setI1.center);

\draw[-] (relR2.center) -| (setU11.center);
\draw[->] (setU11.center) |- (setU1.center);

\draw[->] (setI)  -- (setIL);
\draw[->] (setD) -- (setDL);
\draw[->] (setU) -- (dupU);

\draw[->] (setIL) -- (resI);
\draw[->] (setDL) -- (resD);
\draw[->] (dupU) -- (resU);

\end{tikzpicture}}
\vspace*{0.1cm}

\caption{Probabilistic set operations. The joins filter out the facts
that are not needed for the result and they add the input lineages in 
the same schema, so that output lineages can be formed using
lineage-concatenating functions.}
\vspace{0.2cm}
\label{fig:probabilisticDB}
\end{figure}
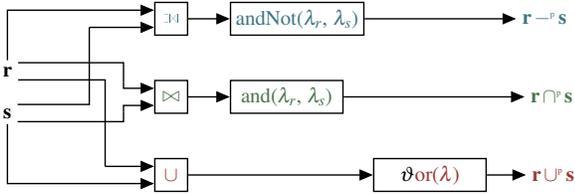

\smallskip \noindent\textbf{Temporal-Probabilistic Set Operations.}
  A temporal-probabilistic model was introduced in the work of
  Dekhtyar et al.~\cite{Dekhtyar2001}.  Each tuple includes a TP part
  consisting of two temporal conditions, corresponding to sets of
  potential starting and ending points, and a pair of probability
  values, corresponding to the minimum and the maximum probability of
  the fact being true.  Conceptually, TP relations are converted into
  {\em annotated relations}, i.e., relations with tuples at a
  time-point granularity, and they are queried using annotated
  operators.  The result is converted back to the initial compact
  representation, using probability combination functions.  The use of
  these functions instead of lineage information has two implications.
  Firstly, change preservation \cite{DignosBG12}, a property of the
  temporal domain is not satisfied, since lineage is not used as a
  criteria to merge the results of consecutive time points into
  maximal intervals.  Secondly, the closure property \cite{
    Imielinski1984,SuciuKochOlteanuReBook} of the probabilistic
  domain is not satisfied, since we lose track of the input tuples
  used for computing the probability of an output tuple, thus making
  the final result non-compositional.

Dylla et al.\ \cite{DyllaMT13} introduced a closed and complete TP
database model, coined TPDB, based on existing temporal and
probabilistic concepts.  Query processing is performed in two steps
(cf.\ Fig.~\ref{fig:tpdbReduction}). The first step, grounding,
evaluates a chosen deduction rule (formulated in Datalog with
additional time variables and temporal predicates) and computes the
lineage expressions of the deduced tuples. The second step,
deduplication, removes the duplicates that could occur in the
grounding step by adjusting their intervals. Although the TPDB data
model is generic, the grounding step cannot cover operations whose
results include subintervals that are only present in one of the two
input relations. As explained in Section~\ref{sec:tpsetOps}, 
sequenced TP set difference is one of these operations and is not
supported by TPDB.

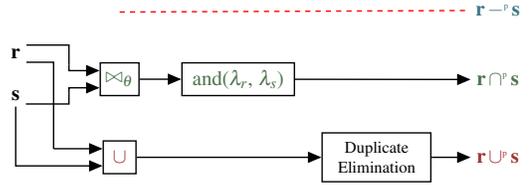
\begin{figure}[!h]
\tikzset{%
  algebra/.style    = {draw, thick, rectangle, minimum height = 2em,
    minimum width = 2em},
}
\tikzset{%
  lineage/.style    = {draw, thick, rectangle, minimum height = 2em,
    minimum width = 2.3cm},
}
\tikzset{%
  filter/.style    = {draw, thick, rectangle, minimum height = 2em,
    minimum width = 2.75cm},
}
\tikzset{%
  block/.style    = {draw, thick, rectangle, minimum height = 2em,
    minimum width = 2em},
}
\centering
\scalebox{0.65}{
\centering
\begin{tikzpicture}[auto, thick, node distance=2cm, >=triangle 45,
  block1/.style = {draw, thick, rectangle, minimum height = 1cm,
    minimum width = 2em}]
\draw
	node at (0,0)[right=-3mm] (relR) {\large $\mathbf{r}$}
 	node [above = 0.05cm of relR.east] (relR1) {}
	node [below = 0.05cm of relR.east] (relR2) {}
	node [right= 0.5cm of relR] (relR_line1) {}
	
	node [below = 0.4cm of relR] (relS) {\large $\mathbf{s}$}
 	node [above = 0.05cm of relS.east] (relS1) {}
	node [below = 0.05cm of relS.east] (relS2) {}
	node [right= 0.25cm of relS] (relS_line1) {}

	node [above right = 0.5cm and 2cm of relR.north] (setD) {}
	node [above = 0.05cm of setD.west] (setD1) {}
	node [below left = 0.05cm and 1.2cm of setD.west] (setD11) {}
	node [below = 0.05cm of setD.west] (setD2) {}
	
	node [below = 0.9cm of setD, algebra] (setI)
	{\large \color{c1}{$\bowtie_{\theta}$}}
	node [above left = 0.05cm and 0.5cm of setI.west] (setI11) {}
	node [below left = 0.05cm and 0.5cm of setI.west] (setI22) {}
	node [above = 0.05cm of setI.west] (setI1) {}
	node [below = 0.05cm of setI.west] (setI2) {}

	node [below = 0.9cm of setI, algebra] (setU) {\large \color{c3}{$\cup$}}
	node [above = 0.05cm of setU.west] (setU1) {}
	node [left = 0.75cm of setU1] (setU11) {}
	node [below = 0.05cm of setU.west] (setU2) {}

	node [right = 0.85cm of setI, lineage] (setIL) {\large \color{c1}{and($\lambda_r$, $\lambda_s$)}}
	node [right = 0.85cm of setU, lineage,white] (setUL) {}
	
	node [right = 0.6cm of setUL, block1, text width=2cm,align=center] (dupU) {Duplicate Elimination}

	node [right =0.8cm of dupU] (resU) {\large \color{c3}{$ \mathbf{r} \cup^\kat{p} \mathbf{s}$}}
	node [above =1.1cm of resU] (resI) {\large \color{c1}{$ \mathbf{r} \cap^\kat{p} \mathbf{s}$}}
	node [above =0.9cm of resI] (resD) {\large \color{c2}{$ \mathbf{r} -^\kat{p} \mathbf{s}$}};

\draw[-, dashed, red] (setD.center) -- (resD.west);

\draw[->] (relS.south) |- (setU2.center);

\draw[-] (relS2.center) -| (setI22.south);
\draw[->] (setI22.south) |- (setI2.center);
\draw[-] (relR1.center) -| (setI11.north);
\draw[->] (setI11.north) |- (setI1.center);

\draw[-] (relR2.center) -| (setU11.center);
\draw[->] (setU11.center) |- (setU1.center);

\draw[->] (setI)  -- (setIL);
\draw[->] (setU) -- (dupU);

\draw[->] (setIL) -- (resI);
\draw[->] (dupU) -- (resU);
\end{tikzpicture}}

\vspace*{0.1cm}
\caption{TP set operations in TPDB. Condition $\theta$ includes
    temporal predicates and duplicate elimination forms output
    intervals.}
\label{fig:tpdbReduction}
\end{figure}

\section{Data Model \& Notation}
\label{sec:preliminaries}

We denote a {\bf temporal-probabilistic schema} by $R^\kat{Tp}$($F$,
$\lambda$, $T$, $p$), where $F$ = ($A_1$, $A_2$, $\ldots$, $A_m$) is
an ordered set of attributes, and each attribute $A_i$ is assigned to
a fixed domain $\Omega_i$.  $\lambda$ is a Boolean formula
corresponding to a lineage expression. $T$ is a {\em temporal
  attribute} with domain $\Omega^T \times \Omega^T$, where $\Omega^T$
is a finite and ordered set of {\em time points}. $p$ is a {\em
  probabilistic attribute} with domain
$\Omega^p = (0,1] \subset {\rm I\!R}$. A {\bf temporal-probabilistic
  relation $\mathbf{r}$} over $R^\kat{Tp}$ is a finite set of
tuples. Each tuple $r \in \mathbf{r}$ is an ordered set of values in
the appropriate domains.  The value of attribute $A_i$ of $r$ is
denoted by $r.A_i$.  The conventional attributes $F$ = ($A_1$, $A_2$,
$\ldots$, $A_m$) of tuple $r$ form a so-called {\em fact}, and we
write $r.F$ to denote the fact $f$ captured by tuple $r$. For example,
the tuple (\bsq{milk}, $a_1$, $[2,10)$, $0.3$) of relation
$\mathbf{a}$ (see Fig.~\ref{fig:tpdbRelations}) includes the fact
$a_1.F$ = (\bsq{milk}), the lineage expression $a_1.\lambda = a_1$,
the time interval $a_1.T = [2,10)$, and the probability value
$a_1.p = 0.3$.  The temporal-probabilistic annotations of the schema
express that
\begin{enumerate*}
\item[(i)] $a_1 = \mathit{true}$ with probability $a_1.p$ for every
  time point in $a_1.T$,
\item[(ii)] $a_1 = \mathit{false}$ with probability $1-a_1.p$ for
  every time point in $a_1.T$, \item[(iii)] and $a_1$ is always
  $\mathit{false}$ outside $a_1.T$.
\end{enumerate*}

By following conventions from \cite{DyllaMT13,DignosTODS16,DignosBG12,
OlteanuHK09}, we assume duplicate-free input and output relations. 
Formally, a temporal-probabilistic relation $\mathbf{r}$ is {\bf
duplicate-free} iff $\forall r, r' \in \mathbf{r} (r \neq r'
\Rightarrow r.F \neq r'.F \vee r.T \cap r'.T = \emptyset))$.  In other
words, the intervals of any two tuples of $\mathbf{r}$ with the same
fact $f$ do not overlap.

A {\bf lineage expression} $\lambda$ is a Boolean formula, consisting
of tuple identifiers and the three Boolean connectives $\neg$
(``not"), $\land$ (``and") and $\lor$ (``or").  Tuple identifiers
represent Boolean random variables among which we assume independence
\cite{DyllaMT13, OlteanuHK09, DalviS07}).  For a base tuple $r$,
$r.\lambda$ is an atomic expression consisting of just $r$ itself. For
a result tuple $\tilde{r}$ derived from one or more TP operations,
$\tilde{r}.\lambda$ is a Boolean expression as defined above. For a
result tuple, lineage is determined by the tem\-poral-probabilistic
operators (formally defined in Section~\ref{sec:querySemantics}) that
were applied to derive that tuple from the base tuples. The
probability of a result tuple is computed via a probabilistic
valuation of the tuple's lineage expression, using either exact (see,
e.g., \cite{DalviS07, DalviS12, olteanu2008using}) or approximate
(see, e.g., \cite{FinkHO13,FinkO11, gatterbauer2014oblivious,
  GatterbauerS15, OlteanuHK10}) algorithms.  For example, in the
result relation of Fig.~\ref{fig:tpdbQueryResult}, the lineage
$c_1 \land \lnot a_1$ yields a marginal probability of $0.6 \cdot 
(1-0.3) = 0.42$ by assuming independence among the base tuples $c_1$
and $a_1$ (see Fig.~\ref{fig:tpdbRelations}). 

  Finally, we write $\lambda^{\mathbf{r}, f}_{t}$ as an
  abbreviation for:

\begin{equation}
   \lambda^{\mathbf{r}, f}_{t} =\begin{cases}
    r.\lambda & \text{iff \ $r \in {\bf r}\ \land r.F = f\ \land\ t\in r.T$} \\
    \mathtt{null} & \text{iff $\nexists\ r \in {\bf r}\ (r.F = f\ \land\ t\in r.T)$}.
  \end{cases}
\end{equation}

Thus, $\lambda^{\mathbf{r}, f}_{t}$ refers to the lineage expression
of a tuple in relation $\mathbf{r}$ with fact $f$ that is valid at
time point $t$.  If there are no tuples in $\mathbf{r}$ with fact $f$
at time point $t$, we write
$\lambda^{\mathbf{r}, f}_{t} = \mathtt{null}$.

\section{Query Semantics}
\label{sec:querySemantics}
For our query semantics, we adopt both the \emph{sequenced
  semantics}~\cite{BohJen2009}, widely used for the temporal
dimension, and the \emph{possible-worlds
  semantics}~\cite{SuciuKochOlteanuReBook}, commonly used for the
probabilistic dimension. The sequenced semantics is consistent with
viewing a temporal database as a sequence of atemporal databases (the
``snapshots''), one for each time point $t$ in $\Omega^T$.
Conceptually, query evaluation then resolves to evaluating a query
against each of these snapshots and producing maximal output intervals
according to time points with equivalent {\em data lineage}. Thus, an
output interval consists of time points, in which the corresponding
fact has been derived based on the same input tuples. The
possible-worlds semantics defines a probabilistic database as a
probability distribution over a finite set of possible states (aka.\
``worlds'') in which the probabilistic database could
be. Conceptually, a query is evaluated against each of the possible
worlds. The marginal probability of an answer tuple then is defined as
the sum of the possible-worlds probabilities, for which the answer
tuple exists. Data
lineage~\cite{DBLP:journals/vldb/BenjellounSHTW08,SarmaTW08}, in the
form of a Boolean expression, serves as a concise condition that is
satisfied over the possible worlds in which each answer tuple exists.

The query semantics of our sequenced TP data model is based on an
intriguing analogy between the temporal and probabilistic semantics:
rather than iterating over snapshots or possible worlds, they both use
the notion of data lineage to define their operational
semantics. Given a TP relation $\mathbf{r}$, a tuple
$r \in \mathbf{r}$ is valid at every time point $t$ included in its
time interval $r.T$ with probability $r.p$.  Thus, all tuples of a TP
relation $\mathbf{r}$ that are valid at time point $t$ with a given
probability are included in the \emph{probabilistic snapshot} of
$\mathbf{r}$ at $t$. Specifically, we obtain the probabilistic
snapshot of a TP relation $\mathbf{r}$ with schema $R^\kat{Tp}$ =
($F$, $\lambda$, $T$, $p$) at time point $t$ by applying the {\em
  timeslice operator} $\tau^\kat{p}_t$, which is defined as:

\[
  \tau^\kat{p}_t ({\bf r}^\kat{Tp}) = \{ (r.F, r.\lambda, [t,t+1), r.p)
  \,|\, r \in \mathbf{r} \land t \in r.T\}
\]

In Fig.~\ref{fig:probSnaps}, we illustrate the probabilistic
snapshots of the relations $\mathbf{a}$ and $\mathbf{c}$ of
Fig.~\ref{fig:tpdbRelations} at time point $t=2$. The probabilistic
snapshot of relation $\mathbf{b}$ at this time point is
$\mathtt{null}$ since there is no tuple of $\mathbf{b}$ valid.

\begin{figure}[!ht]\centering
  \setlength{\tabcolsep}{5pt}
  \def\arraystretch{1.2}
    \scalebox{0.85}{
      \begin{tabular}{c|c|c|c}
        \multicolumn{4}{l}{$\mathbf{a}$ {\bf (productsBought)}} \\
        \hline
        \textit{Product} & $\lambda$ & $T$ & $p$ \\ \hline
        \bsq{milk} & {$a_1$} & {[2,3)} & {0.3} \\
        \bsq{dates}  & {$a_3$} & {[2,3)} & {0.6} \\ \hline
      \end{tabular}}
    \qquad
    \scalebox{0.85}{
      \begin{tabular}{c|c|c|c}
        \multicolumn{4}{l}{$\mathbf{c}$ {\bf (productsInStock)}} \\
        \hline
        \textit{Product} & $\lambda$ & $T$ & $p$ \\ \hline
        \bsq{milk} & {$c_1$} & {[2,3)} & {0.6} \\ \hline
      \end{tabular}}
    \vspace*{0.1cm}
    \caption{Probabilistic Snapshots $\tau^\kat{p}_2({\bf a})$
    		 	  and $\tau^\kat{p}_2({\bf c})$}
    \label{fig:probSnaps}
\end{figure}

\begin{definition} {\it{\bf(TP Snapshot Reducibility)}}
{\em Let $\mathbf{r}_1, \ldots, \mathbf{r}_m$ be a set of TP
relations, let $op^\kat{Tp}$ be an $m$-ary temporal-probabilistic
operator, let $op^\kat{p}$ be the corresponding probabilistic
operator, let $\Omega^T$ be the time domain, and let
$\tau^\kat{p}_t ({\bf r})$ be the timeslice operator. The operator
$op^\kat{Tp}$ is {\em snapshot reducible} to $op^\kat{p}$ iff, for
all $t\in\Omega^T$, it holds that:}

\[
  \tau^\kat{p}_t (op^\kat{Tp}(\mathbf{r}_1, \ldots, \mathbf{r}_m))
  \equiv op^\kat{p}(\tau^\kat{p}_t (\mathbf{r}_1), \ldots,
  \tau^\kat{p}_t (\mathbf{r}_m))
\]
\label{def:snapshot}
\end{definition}

\vspace*{-0.2cm} {Snapshot reducibility} states that a probabilistic
snapshot of the result of an $m$-ary TP operation
$op^\kat{Tp}(\mathbf{r}_1, \ldots, \mathbf{r}_m)$ at any time point
$t$ is equivalent to the result derived from the corresponding
probabilistic operation $op^\kat{p}$ on the probabilistic snapshots 
of the input relations at $t$. Applying an atemporal operation over
all probabilistic snapshots thus is consistent with snapshot
reducibility in temporal databases and implies that the result at
any time point $t$, both in terms of probability values and facts,
is determined only by the input tuples that are valid at $t$. The
application of $op^\kat{p}$ guarantees that the computations at each
time point will yield Boolean lineage expressions that are
consistent with the possible-worlds
semantics~\cite{SarmaTW08,DBLP:journals/vldb/BenjellounSHTW08}.

As example, consider the query of Fig.~\ref{fig:tpdbQueryPlan} over
the relations of Fig.~\ref{fig:tpdbRelations}. According to the
lineage expression of tuple (\bsq{milk}, [2,4), $c_1\land \neg a_1$,
0.42), at $t=2$, the fact \bsq{milk} has been derived from the input
tuples $a_1$ and $c_1$, i.e., the only input tuples of the 
probabilistic snapshot at $t=2$ (Fig.~\ref{fig:probSnaps} that include
the fact \bsq{milk}. Since the probability of \bsq{milk} at $t=2$ is
only affected by the probabilities of $a_1$ and $c_1$, it can be
computed based on the lineage expression $c_1\land \neg a_1$. 

\begin{definition} {\it{\bf(TP Change Preservation)}} {\em Let
$\mathbf{r}_1, \ldots, \mathbf{r}_m$ be a set of TP relations,
let $op^\kat{Tp}$ be an $m$-ary temporal-probabilistic operator,
and let $u.T_s$, $u.T_e$ denote the start and end points of an
interval associated with a tuple $u$. For each tuple $u \in 
\mathbf{u}$, where ${\bf u} = op^\kat{Tp}(\mathbf{r}_1, \ldots,
\mathbf{r}_m)$, it holds that:}
  \begin{align*}
    &\forall t,t' \in u.T (\lambda^{\mathbf{u},
      u.F}_{t}  \equiv \lambda^{\mathbf{u}, u.F}_{t'}) \ \land\ \\
    &\nexists u'\in{\bf u}((u'.T_e = u.T_s \vee u'.T_s =  u.T_e) \land
    (u'.\lambda \equiv u.\lambda))
  \end{align*}
  \label{def:change}
\end{definition}

Intuitively, change preservation ensures that only consecutive time
points of tuples with equivalent lineage expressions are grouped into
intervals. For example, the output tuples (\bsq{milk}, [1,2), $c_1$,
0.6) and (\bsq{milk}, [2,4), $c_1 \land \neg a_1$, 0.42) are not
merged into the interval $[1,4)$, since they do not have equivalent
lineages. Change preservation guarantees that a fact is valid over the
same possible worlds with maximal intervals. The first line of
Def.~\ref{def:change} ensures that the lineage expression at all time
points in the interval of a result tuple is the same. The second line
ensures that the time intervals produced by coalescing time points
with the equivalent lineage expressions are maximal.\footnote{Rather
  than performing logical equivalence checks among Boolean formulas,
  which are co-NP-complete, we resort to a syntactic comparison of the
  lineage sets in our implementation.}

\section{TP Set Operations \& Queries}
\label{sec:tpsetOps}
\subsection{TP Set Operations}

In TP databases, the result of a \emph{TP set union} includes, at each
time point $t \in \Omega^T$, the facts for which there is a non-zero probability to be
in $\mathbf{r}$ or in $\mathbf{s}$; the result of a \emph{TP set
  intersection} includes, at each time point, the facts for which
there is a non-zero probability to be in $\mathbf{r}$ and in
$\mathbf{s}$; and the result of a \emph{TP set difference} between two
TP relations $\mathbf{r}$ and $\mathbf{s}$ includes, at each time
point, the facts for which there is a non-zero probability to be in
$\mathbf{r}$ and not in $\mathbf{s}$.

\begin{definition} \label{def:tpSetOpDef} {\it{\bf(TP Set Operations)}}
  {\em Let ${\bf r}$ and ${\bf s}$ be tem\-poral-probabilistic
    relations with schema ($F$, $\lambda$, $T$, $p$), and let
    $\lambda^{\mathbf{r},f}_{t}$ denote the lineage expression of the
    tuple in relation $\mathbf{r}$ that includes fact $f$ and is valid
    at time point $t$.  Given a result tuple $\tilde{r}$ and the
    lineage-concatenation functions depicted in
    Table~\ref{tab:lineageFunctions}, we define the three TP set
    operations $\mathbf{r} \cup^\kat{Tp} \mathbf{s}$,
    $\mathbf{r} \cap^\kat{Tp} \mathbf{s}$ and
    $\mathbf{r} -^\kat{Tp} \mathbf{s}$ as follows:}
  \begin{align*}
    \tilde{r} \in \mathbf{r} \cup^\kat{Tp} \mathbf{s}
    \Longleftrightarrow\ 
    & \forall t \in \tilde{r}.T ( (\lambda^{\mathbf{r}, \tilde{r}.F}_{t}
      \neq \mathtt{null} \
      \vee \ \lambda^{\mathbf{s}, \tilde{r}.F}_{t} \neq \mathtt{null}) \ \wedge \\
    & \hspace*{1.5cm}
      \tilde{r}.\lambda \equiv \textbf{or}(\lambda^{\mathbf{r}, \tilde{r}.F}_{t},
      \lambda^{\mathbf{s}, \tilde{r}.F}_{t})) \ \wedge \\
    & \forall t' \notin \tilde{r}.T ( \tilde{r}.\lambda \not\equiv
      \textbf{or}(\lambda^{\mathbf{r}, \tilde{r}.F}_{t'},
      \lambda^{\mathbf{s}, \tilde{r}.F}_{t'}))
    \\[10pt]
    \tilde{r} \in \mathbf{r} \cap^\kat{Tp} \mathbf{s}
    \Longleftrightarrow\ 
    & \forall t \in \tilde{r}.T (
      \lambda^{\mathbf{r}, \tilde{r}.F}_{t} \neq \mathtt{null} \wedge
      \lambda^{\mathbf{s}, \tilde{r}.F}_{t}  \neq \mathtt{null}\ \wedge\\
    & \hspace*{1.5cm}
      \tilde{r}.\lambda \equiv \textbf{and}(\lambda^{\mathbf{r},
      \tilde{r}.F}_{t}, \lambda^{\mathbf{s}, \tilde{r}.F}_{t})) \ \wedge\\
    & \forall t' \notin \tilde{r}.T (
      \tilde{r}.\lambda \not\equiv \textbf{and}(\lambda^{\mathbf{r}, \tilde{r}.F}_{t'},
      \lambda^{\mathbf{s}, \tilde{r}.F}_{t'}))
    \\[10pt]
    \tilde{r} \in \mathbf{r} -^\kat{Tp} \mathbf{s}
    \Longleftrightarrow\ 
    & \forall t \in \tilde{r}.T ( 
      \lambda^{\mathbf{r}, \tilde{r}.F}_{t} \neq \mathtt{null} \ \wedge \\
    & \hspace*{1.5cm}
      \tilde{r}.\lambda \equiv \textbf{andNot}(\lambda^{\mathbf{r}, \tilde{r}.F}_{t},
      \lambda^{\mathbf{s}, \tilde{r}.F}_{t})) \ \wedge
    \\
    & \forall t' \notin \tilde{r}.T  ( 
      \tilde{r}.\lambda  \not\equiv \textbf{andNot}(\lambda^{\mathbf{r},
      \tilde{r}.F}_{t'}, \lambda^{\mathbf{s}, \tilde{r}.F}_{t'}))
  \end{align*}
\end{definition}

\begin{table}[!h]\centering
  \caption{Definition of lineage-concatenation functions.}
  \label{tab:lineageFunctions}
  \begin{tabular}{M{0.6in} M{2.1in}  @{}m{0pt}@{}}
      \hline
      $\textit{\textbf{and}}(\lambda_1,\lambda_2)$ &
        \hspace*{-1.12in}$~=~ (\lambda_1) \land (\lambda_2)$  &\\ [0.4cm]
      \hline
      $\textit{\textbf{andNot}}(\lambda_1,\lambda_2)$ &
                                                        $~~=~\left \{ 
                    \begin{array}{l l}
                     (\lambda_1) & \quad \text{if $\lambda_2 = \mathtt{null}$}\\
                     (\lambda_1) \land \neg (\lambda_2) & \quad \text{otherwise}\\
                    \end{array} \right.$ &\\[0.8cm]
          \hline
          $\textit{\textbf{or}}(\lambda_1,\lambda_2)$ & $~=~ \left \{ 
          \begin{array}{l l}
            (\lambda_1) & \quad \text{if $\lambda_2 = \mathtt{null}$}\\
            (\lambda_2) & \quad \text{if $\lambda_1 = \mathtt{null}$}\\
            (\lambda_1) \lor (\lambda_2) & \quad \text{otherwise}\\
          \end{array} \right.$ &\\ [1.1cm]
          \hline                 
    \end{tabular}
\end{table}

The above definition of TP set operations specifies the intervals and
lineage expressions of a result tuple $\tilde{r}$.  The first line of
the definition of each operation relates to Def.~\ref{def:snapshot}.
It states that, at any time point $t \in \tilde{r}.T$, fact
$\tilde{r}.F$ must be included in the corresponding input tuples from
$\mathbf{r}$ and $\mathbf{s}$.  Consequently, the lineage expression
of the output tuple $\tilde{r}$ at each time point $t \in \tilde{r}.T$
(cf.\ second line) is computed based on the same input tuples, according to
the lineage-concatenating functions of
Table~\ref{tab:lineageFunctions}.  In the case of set union, there
must exist at least one tuple in either one of the two input relations
that also includes $\tilde{r}.F$ over $\tilde{r}.T$.  For set
intersection, there must exist corresponding tuples in both input
relations.  For set difference, an output tuple is produced at all
time points $t$, at which there exists a tuple of the left relation
$r$ that is valid at $t$ in $r.T$.  This happens in two cases: (a) if
a fact $f$ is included in a tuple of $\mathbf{r}$ but in no tuple in
$\mathbf{s}$, and (b) if a fact $f$ is included in a tuple of
$\mathbf{r}$ but, with a probability of less than 1, also in a tuple
of $\mathbf{s}$. %
The first case resembles the definition of temporal set difference,
where, at each time point in the output, there exist facts that are
included in tuples of $\mathbf{r}$ and not in tuples of $\mathbf{s}$.
The second case occurs due to the probabilistic dimension. The result
of a probabilistic set difference between $\mathbf{r}$ and
$\mathbf{s}$ includes all facts, which have a non-zero
probability to be in $\mathbf{r}$ and not in $\mathbf{s}$.

\begin{example}
  Figure~\ref{fig:defExample} shows the relations $\mathbf{a}$ and
  $\mathbf{c}$ of Fig.~\ref{fig:tpdbRelations} as well as selected
  output tuples of $\mathbf{a} - ^\kat{Tp} \mathbf{c}$. Different
  colors are used for different facts: green is used for \bsq{milk},
  blue for \bsq{dates} and red for \bsq{chips}. Output tuples are
  drawn below the time axis. For example, the output tuple (\bsq{milk},
  $a_1 \land \lnot c_2$, $[6,8)$, $0.09$) satisfies
  Def.~\ref{def:tpSetOpDef}: for all time points in $[6,8)$, it holds
  that
  $\lambda^{\mathbf{a}, \kat{\bsq{milk}}}_{t}= {a_1} \neq
  \mathtt{null}$ and
  $\lambda^{\mathbf{c}, \kat{\bsq{milk}}}_{t}= {c_2}$. Thus,
  $\forall t \in [6,8)$,
  $\textbf{andNot}(\lambda^{\mathbf{a}, \kat{\bsq{milk}}}_{t},
  \lambda^{\mathbf{c}, \kat{\bsq{milk}}}_{t}) \equiv a_1 \land \lnot c_2$.

  \begin{figure}[ht]\centering
    \scalebox{0.58} {
      \begin{tikzpicture}
        \draw (0,0) [->, line width = 1pt]-- coordinate (x axis mid) (14.2,0);
  
        \pgfmathsetmacro{\shift}{0.5}
        \foreach \j in {1,...,10}{
          \pgfmathsetmacro{\divRes}{int(\j/2)}
          \pgfmathsetmacro{\modRes}{1-(\j-\divRes*2)}
          \pgfmathsetmacro{\xPos}{\j-1+(\divRes-\modRes*\shift)}
          \draw (\xPos , 1pt)  --  (\xPos ,-3pt)
            node[anchor=north,font=\relsize{2}] {};
        }

        \pgfmathsetmacro{\shift}{0.5}
        \foreach \j in {1,...,9}{
          \pgfmathsetmacro{\divRes}{int(\j/2)}
          \pgfmathsetmacro{\modRes}{1-(\j-\divRes*2)}
          \pgfmathsetmacro{\xPos}{\j-1+(\divRes-\modRes*\shift)}
          \draw (\xPos+0.75 , -3pt)
            node[anchor=north,font=\relsize{1}] {\bf \j};
        }

        \pgfmathsetmacro{\xAxisOne}{-0.2}
        \pgfmathsetmacro{\xAxisTwo}{13.7}
        \pgfmathsetmacro{\windowHeight}{1.5}
 
        \pgfmathsetmacro{\yAxis}{2.1}
        \draw [line width=1.5,color=c1] (1.5, 2.7) -- (13.5, 2.7)
          node[pos=.2,above=-1pt]{\large $a_1$}; 
        \draw [line width=1.5,color=c3] (4.5, 2.1) -- (8.95, 2.1)
          node[pos=.2,above=-1pt]{\large $a_2$}; 
        \draw [line width=1.5,color=c4] (0, 2.1) -- (3, 2.1)
          node[pos=.2,above=-1pt]{\large $a_3$}; 
        \node at (\xAxisOne-0.3, \yAxis+0.4) {\large $\mathbf{a}$}; 
        \pgfmathsetmacro{\yAxis}{\yAxis-0.2}
        \draw [dashed] (\xAxisOne,\yAxis) --
                       (\xAxisOne,\yAxis + \windowHeight - 0.2) --
                       (\xAxisTwo,\yAxis + \windowHeight - 0.2) --
                       (\xAxisTwo, \yAxis + 0) --
                       cycle;

    \pgfmathsetmacro{\yAxis}{0.4}
    \draw [line width=1.5,color=c1] (0, 1) -- (4.5,1)
    node[pos=.15,above=-1pt]{\large $c_1$}; 
    \draw [line width=1.5,color=c1] (7.5, 1) -- (10.5,1)
    node[pos=.2,above=-1pt]{\large  $c_2$}; 
    \draw [line width=1.5,color=c3] (4.5, 0.4) -- (6,0.4)
    node[pos=.2,above=-1pt]{\large $c_3$}; 
    \draw [line width=1.5,color=c3] (9, 0.4) -- (12,0.4)
    node[pos=.2,above=-1pt]{\large $c_4$}; 
    \node at (\xAxisOne-0.3, \yAxis+0.4) {\large $\mathbf{c}$}; 
    \pgfmathsetmacro{\yAxis}{\yAxis-0.2}
 	\draw [dashed]
	(\xAxisOne,\yAxis) -- (\xAxisOne,\yAxis + \windowHeight - 0.2) -- (\xAxisTwo,\yAxis + 	\windowHeight - 0.2) -- (\xAxisTwo, \yAxis + 0)  -- cycle;
    \pgfmathsetmacro{\yAxis}{3.5}
    \draw[color=gray,dashed] (0, \yAxis)  --  (0,-1.5) ;    
   	\draw[color=gray,dashed] (3, \yAxis)  --  (3,-1.5) ;

    \draw[color=gray,dashed] (4.5, \yAxis)  --  (4.5,-1.5) ;
	\draw[color=gray,dashed] (6, \yAxis)  --  (6,-1.5) ;
	
	\draw[color=gray,dashed] (7.5, \yAxis)  --  (7.5,-1.5) ;
    \draw[color=gray,dashed] (10.5, \yAxis)  --  (10.5,-1.5) ;
%
    \draw [line width=1.5,color=c4] (0,-1.2) -- (3,-1.2)
    node[pos=.5,above=-1pt]{\large $(a_3, 0.6)$};
    \draw [line width=1.5,color=c3] (4.5,-1.2) -- (6,-1.2)
    node[pos=.5,above=-1pt]{\large $(a_2 \land \neg \ c_3, 0.24)$};
    \draw [line width=1.5,color=c1] (7.5,-1.2) -- (10.5,-1.2)
    node[pos=.5,above=-1pt]{\large $(a_1 \land \neg \  c_2, 0.09)$};
  \end{tikzpicture}}

  \caption{Selected output tuples of $\mathbf{a} - ^\kat{Tp} \mathbf{c}$.}
  \label{fig:defExample}
\end{figure}
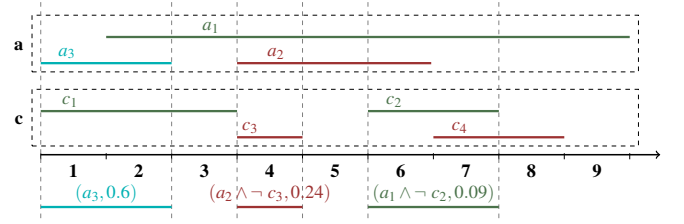
\end{example} 

The third line of the definition of each TP set operator is a direct
consequence of Def.~\ref{def:change}.  It guarantees that, when
merging consecutive time points into an interval, we consider only the
ones for which the condition in the first line is satisfied.  In other
words, a new interval is created whenever there is a change in the
validity of a tuple from either $\mathbf{r}$ or $\mathbf{s}$ at the
currently considered time point. In Example~\ref{fig:defExample}, at
time points $t=5$ and $t=8$,
$\lambda^{\mathbf{a}, \kat{\bsq{milk}}}_{t} = {a_1}$ and
$\lambda^{\mathbf{c}, \kat{\bsq{milk}}}_{t} = \mathtt{null}$. Thus,
outside the interval $[6,8)$ of tuple (\bsq{milk}, [6,8),
$a_1 \land \neg c_2$, 0.09), there are no time points for which
$\textbf{andNot}(\lambda^{\mathbf{a}, \kat{\bsq{milk}}}_{t},
\lambda^{\mathbf{c}, \kat{\bsq{milk}}}_{t}) \equiv a_1 \land \neg c_2$.
Fig.~\ref{fig:tpSetOps} shows the result of all TP set operations
between relations $\mathbf{a}$ and $\mathbf{c}$ in
Fig.~\ref{fig:tpdbRelations}.

\begin{figure}[ht] \centering 
  \fontsize{8pt}{10pt}\selectfont
  \setlength{\tabcolsep}{2pt}
  \begin{minipage}{0.45\linewidth} \centering
    \begin{tabular}{c|c|l|l}
      \multicolumn{4}{l}{$\mathbf{a} \cup^\kat{Tp} \mathbf{c}$}\\
      \hline
      \textit{Product}  & $\lambda$ & T &  p\\
      \hline
      \bsq{milk}   & $c_1$               & [1,2)  & 0.6  \\
      \bsq{milk}   & $a_1 \lor c_1$  & [2,4)  & 0.72 \\
      \bsq{milk}   & $a_1$               & [4,6)  & 0.3  \\
      \bsq{milk}   & $a_1 \lor c_2$  & [6,8)  & 0.79 \\
      \bsq{milk}   & $a_1$               & [8,10) & 0.3  \\
      \bsq{chips}  & $a_2 \lor c_3$ & [4,5)  & 0.94 \\
      \bsq{chips}  & $a_2$              & [5,7)  & 0.8  \\
      \bsq{chips}  & $c_4$              & [7,9)  & 0.8  \\
      \bsq{dates}  & $a_3$             & [1,3)  & 0.6  \\ 
      \hline
    \end{tabular}
  \end{minipage}
  \quad
  \begin{minipage}{0.4\linewidth} \centering
    \begin{tabular}{c|c|l|l}
      \multicolumn{4}{l}{$\mathbf{a} -^\kat{Tp} \mathbf{c}$}\\
      \hline
      \textit{Product}  & $\lambda$ & T &  p\\
      \hline
      \bsq{milk}   & $a_1 \land \neg c_1$  & [2,4)  & 0.12 \\ 
      \bsq{milk}   & $a_1$                         & [4,6)  & 0.3 \\ 
      \bsq{milk}   & $a_1 \land \neg c_2$  & [6,8)  & 0.09 \\ 
      \bsq{milk}   & $a_1$                         & [8,10) & 0.3 \\ 
      \bsq{chips}  & $a_2 \land \neg c_3$  & [4,5)  & 0.24 \\ 
      \bsq{chips}  & $a_2$                         & [5,7)  & 0.8 \\ 
      \bsq{dates}  & $a_3$                        & [1,3)  & 0.6 \\ 
      \hline
\end{tabular}

\end{minipage}

\vspace*{0.5cm}
\begin{tabular}{c|c|l|l}
      \multicolumn{4}{l}{$\mathbf{a} \cap^\kat{Tp} \mathbf{c}$}\\
      \hline
      \textit{Product}  & $\lambda$ & T &  p\\
      \hline
      \bsq{milk}  & $a_1 \land c_1$ &  [2,4) & 0.18\\
      \bsq{milk}  & $a_1 \land c_2$ &  [6,8) & 0.21\\
      \bsq{chips} & $a_2 \land c_3$ &  [4,5) & 0.56\\
      \hline
\end{tabular}
\medskip

\caption{TP set operations computed for the relations of
  Fig.~\ref{fig:tpdbRelations}.}
 \vspace{-0.3cm}
\label{fig:tpSetOps}
\end{figure}

\subsection{TP Set Queries \& Complexity}
\label{sec:complexity}
Having defined TP set operations, we now move on to TP set queries,
which are expressions of TP set operations over TP relations.

\begin{definition}
  \label{def:tpSetQueryDef} {\it\bf{(TP Set Query)}} {\em Let
    $\mathbf{r}_1,\ldots, \mathbf{r}_m$ be dup\-licate-free TP
    relations. 
    A {\em TP set query} $Q$ is any expression of TP set operators
    that adheres to the following grammar:}
  $$ Q ::= \mathbf{r}_i \mid Q \cup^\kat{Tp} Q \mid Q \cap^\kat{Tp} Q
           \mid Q -^\kat{Tp} Q \mid (Q) $$
\end{definition}

The following theorem and corollary establish an interesting
relationship between {\em safe queries} \cite{DalviS07,DalviS12} in
probabilistic databases and tractable queries in our TP setting. The
theorem is based on the observation that repeated applications of TP
set operations create regular lineage expressions, which are in {\em
one-occurrence form} (1OF) \cite{SuciuKochOlteanuReBook} if none of
the input relations occurs more than once in a TP set query. Formally,
a formula is in 1OF iff no tuple identifier occurs more than once in
the formula. Correspondingly, we call a TP set query $Q$ {\em
  non-repeating} iff every input relation $\mathbf{r}_i$ occurs at
most once in $Q$. \\

\begin{theorem}
  \label{thm:readOnce} {\em Any non-repeating TP set query $Q$ over
    duplicate-free TP relations yields lineage formulas in 1OF.}
\end{theorem}

\begin{proof}
  Consider a TP set operation over two TP relations $\mathbf{r}$ and
  $\mathbf{s}$, both having schema ($F$, $\lambda$, $T$, $p$). Since
  $\mathbf{r}$ and $\mathbf{s}$ are duplicate-free, we cannot have two
  tuples in either $\mathbf{r}$ or $\mathbf{s}$ that share the same fact at
  overlapping time intervals. Assume we have $n_1$ tuples in
  $\mathbf{r}$ and $n_2$ tuples in $\mathbf{s}$ with the same fact
  $f$, but each with non-overlapping time intervals.  Then, for
  $n = n_1 + n_2$ input intervals, we can at most obtain $2\,n-1$
  output intervals. According to change preservation
  (Def.~\ref{def:change}), we create the same amount of output tuples,
  one for each output interval and each with a different combination
  of tuple identifiers in their lineage (Def.~\ref{def:tpSetOpDef}).
  Next, inductively, during any further application of a TP set operation
  (over non-repeating subgoals), change preservation will only merge
  two consecutive time intervals iff their lineages are
  equivalent. This cannot occur, since all of the lineages that are
  created by an individual TP set operator are different. That is, for
  a non-repeating TP set query, each tuple identifier can occur at
  most once in the lineage of a result tuple, which means that the
  lineages are in 1OF.
\end{proof}

\begin{corollary}
\label{cor:ptime}
{\em Any non-repeating TP set query $Q$ over dup\-licate-free TP relations has PTIME data complexity.}
\end{corollary}

The proof of the corollary follows directly from
Theorem~\ref{thm:readOnce}, since computing the marginal probability
of a Boolean formula in 1OF can be done in linear time in the size of
the formula for independent random variables
\cite{SuciuKochOlteanuReBook}.  Also, all temporal alignment
operations are of polynomial complexity (see
\cite{DignosTODS16,DignosBG12} as well as the algorithms in
Section~\ref{sec:sweepingIUE} and Section~\ref{sec:basicIUE}).

The above class of non-repeating TP set queries over duplicate-free TP
relations nicely complements the dichotomy theorem
\cite{DalviS07,DalviS12} established for unions of conjunctive queries
(UCQs) in probabilistic databases.  Each individual TP set operation
over two compatible relation schemas resolves to (a union of) at most
two conjunctive queries, in which no intermediate duplicates due to a
projection onto a subset of attributes in $F$ may arise. Although
repeated applications of TP set operations in a query do not
necessarily form UCQs, the overall query remains hierarchical
\cite{SuciuKochOlteanuReBook}, since all attributes in $F$ are
propagated through the operations. Change preservation, on the other
hand, which is required for a sequenced temporal semantics, preserves
these complexity considerations by merging only intervals with
equivalent lineage expressions into a single output interval. TP set
queries with repeating subgoals however remain \#P-hard as shown in
\cite{DBLP:journals/pvldb/KhannaRT11} (consider, e.g., the query
$\left(\mathbf{r}_1 \cup^\kat{Tp} \mathbf{r}_2\right) -^\kat{Tp}
\left(\mathbf{r}_1 \cap^\kat{Tp} \mathbf{r}_3\right)$).

\section{Lineage-Aware Temporal Windows}
\label{sec:windows}
 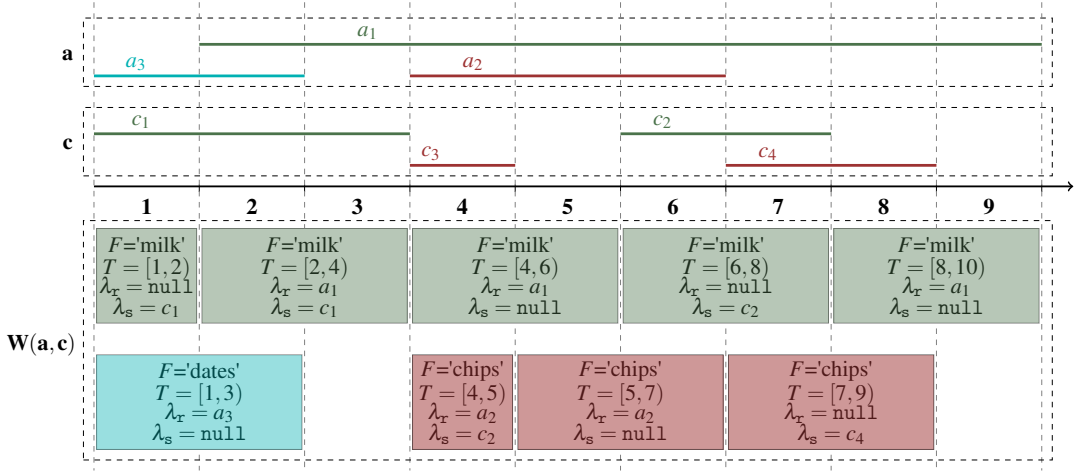
\begin{figure*}[ht]\centering
    \scalebox{0.7} {
      \begin{tikzpicture}
        \draw (0,0) [->, line width = 1pt]-- coordinate (x axis mid) (18.6,0);
  
        \pgfmathsetmacro{\shift}{1}
        \foreach \j in {1,...,9}{
          \draw (\j*2 , 1pt)  --  (\j*2 ,-3pt)
            node[anchor=north,font=\relsize{2}] {};
        }

        \foreach \j in {1,...,9}{
          \pgfmathsetmacro{\divRes}{int(\j/2)}
          \pgfmathsetmacro{\modRes}{1-(\j-\divRes*2)}
          \pgfmathsetmacro{\xPos}{2*(\j-1)}
          \draw (\xPos+1 , -3pt)
            node[anchor=north,font=\relsize{1}] {\bf \j};
        }

        \pgfmathsetmacro{\xAxisOne}{-0.2}
        \pgfmathsetmacro{\xAxisTwo}{18.2}
        \pgfmathsetmacro{\windowHeight}{1.5}

        \pgfmathsetmacro{\yAxis}{2.1}
        \draw [line width=1.5,color=c1] (2*2-2, 2.7) -- (10*2-2, 2.7)
          node[pos=.2,above=-1pt]{\large $a_1$}; 
        \draw [line width=1.5,color=c3] (4*2-2, 2.1) -- (7*2-2, 2.1)
          node[pos=.2,above=-1pt]{\large $a_2$}; 
        \draw [line width=1.5,color=c4] (1*2-2, 2.1) -- (3*2-2, 2.1)
          node[pos=.2,above=-1pt]{\large $a_3$}; 
        \node at (\xAxisOne-0.3, \yAxis+0.4) {\large $\mathbf{a}$}; 
        \pgfmathsetmacro{\yAxis}{\yAxis-0.2}
        \draw [dashed] (\xAxisOne,\yAxis) --
                       (\xAxisOne,\yAxis + \windowHeight - 0.2) --
                       (\xAxisTwo,\yAxis + \windowHeight - 0.2) --
                       (\xAxisTwo, \yAxis + 0) --
                       cycle;

    \pgfmathsetmacro{\yAxis}{0.4}
    \draw [line width=1.5,color=c1] (1*2-2, 1) -- (4*2-2,1)
    node[pos=.15,above=-1pt]{\large $c_1$}; 
    \draw [line width=1.5,color=c1] (6*2-2, 1) -- (8*2-2,1)
    node[pos=.2,above=-1pt]{\large  $c_2$}; 
    \draw [line width=1.5,color=c3] (4*2-2, 0.4) -- (5*2-2,0.4)
    node[pos=.2,above=-1pt]{\large $c_3$}; 
    \draw [line width=1.5,color=c3] (7*2-2, 0.4) -- (9*2-2,0.4)
    node[pos=.2,above=-1pt]{\large $c_4$}; 
    \node at (\xAxisOne-0.3, \yAxis+0.4) {\large $\mathbf{c}$}; 
    \pgfmathsetmacro{\yAxis}{\yAxis-0.2}
 	\draw [dashed]
	(\xAxisOne,\yAxis) -- (\xAxisOne,\yAxis + \windowHeight - 0.2) -- (\xAxisTwo,\yAxis + 	\windowHeight - 0.2) -- (\xAxisTwo, \yAxis + 0)  -- cycle;
    \pgfmathsetmacro{\yAxis}{3.5}
%
%
%
    \pgfmathsetmacro{\yPosB}{-0.7}
    \pgfmathsetmacro{\yPos}{\yPosB-0.4}
    \node at (1, \yPos) {\large $F$=\bsq{milk}}; 
    \node at (1, \yPos-0.45) {\large $T = [1,2)$}; 
    \node at (1, \yPos-0.8) {\large $\mathtt{\lambda_r} = \mathtt{null}$}; 
    \node at (1, \yPos-1.2) {\large $\mathtt{\lambda_s} = c_1$}; 
  	\draw [fill=c1, opacity = 0.4] (0.05,\yPos+0.3) -- (0.05,\yPos - \windowHeight) -- 
  	(1.95,\yPos - \windowHeight) -- (1.95, \yPos + 0.3)  -- cycle;

    \node at (4, \yPos) {\large $F$=\bsq{milk}}; 
    \node at (4, \yPos-0.45)  {\large $T = [2,4)$}; 
    \node at (4, \yPos-0.8) {\large $\mathtt{\lambda_r} = a_1$}; 
    \node at (4, \yPos-1.2) {\large $\mathtt{\lambda_s} = c_1$}; 
   	\draw [fill=c1, opacity = 0.4] (2+0.05,\yPos+0.3) -- (2+0.05,\yPos - \windowHeight) --
   	(6-0.05,\yPos - 	\windowHeight) -- (6-0.05, \yPos + 0.3)  -- cycle;

    \node at (8, \yPos) {\large $F$=\bsq{milk}}; 
    \node at (8, \yPos-0.45) {\large $T = [4,6)$}; 
    \node at (8, \yPos-0.8) {\large $\mathtt{\lambda_r} = a_1$}; 
    \node at (8, \yPos-1.2) {\large $\mathtt{\lambda_s} = \mathtt{null}$}; 
   	\draw [fill=c1, opacity = 0.4](6+0.05,\yPos+0.3) -- (6+0.05,\yPos - \windowHeight) --
   	(10-0.05,\yPos - 	\windowHeight) -- (10-0.05, \yPos + 0.3)  -- cycle;

    \node at (12, \yPos) {\large $F$=\bsq{milk}};    
    \node at (12, \yPos-0.45) {\large $T = [6,8)$}; 
    \node at (12, \yPos-0.8)  {\large $\mathtt{\lambda_r} = \mathtt{null}$}; 
    \node at (12, \yPos-1.2) {\large $\mathtt{\lambda_s} = c_2$};     
   	\draw [fill=c1, opacity = 0.4] (10+0.05,\yPos+0.3) -- (10+0.05,\yPos - \windowHeight) --
   	(14-0.05,\yPos - 	\windowHeight) -- (14-0.05, \yPos + 0.3)  -- cycle;

    \node at (16, \yPos) {\large $F$=\bsq{milk}}; 
    \node at (16, \yPos-0.45)  {\large $T = [8,10)$}; 
    \node at (16, \yPos-0.8) {\large $\mathtt{\lambda_r} = a_1$}; 
    \node at (16, \yPos-1.2) {\large $\mathtt{\lambda_s} = \mathtt{null}$};     
   	\draw [fill=c1, opacity = 0.4] (14+0.05,\yPos+0.3) -- (14+0.05,\yPos - \windowHeight) --
   	(18-0.05,\yPos - \windowHeight) -- (18-0.05, \yPos + 0.3)  -- cycle;

	\pgfmathsetmacro{\yAxis}{3.5}    
	\pgfmathsetmacro{\temp}{\yPosB}   
	\pgfmathsetmacro{\yPosA}{\yPos}
	\pgfmathsetmacro{\yPosB}{\yPos-4.3}
	\pgfmathsetmacro{\yPos}{\temp-2.8}

    \node at (2, \yPos) {\large $F$=\bsq{dates}}; 
    \node at (2, \yPos-0.45) {\large $T = [1,3)$}; 
    \node at (2, \yPos-0.8) {\large $\mathtt{\lambda_r} = a_3$}; 
    \node at (2, \yPos-1.2) {\large $\mathtt{\lambda_s} = \mathtt{null}$}; 
   	\draw [fill=c4, opacity = 0.4] (0+0.05,\yPos+0.3) -- (0+0.05,\yPos - \windowHeight) --
   	(4-0.05,\yPos - 	\windowHeight) -- (4-0.05, \yPos + 0.3)  -- cycle;

    \node at (7, \yPos) {\large $F$=\bsq{chips}}; 
    \node at (7, \yPos-0.45) {\large $T = [4,5)$}; 
    \node at (7, \yPos-0.8) {\large $\mathtt{\lambda_r} = a_2$}; 
    \node at (7, \yPos-1.2) {\large $\mathtt{\lambda_s} = c_2$}; 
   	\draw [fill=c3, opacity = 0.5] (6+0.05,\yPos+0.3) -- (6+0.05,\yPos - \windowHeight) --
   	(8-0.05,\yPos - \windowHeight) -- (8-0.05, \yPos + 0.3)  -- cycle;
   	
    \node at (10, \yPos) {\large $F$=\bsq{chips}};    
    \node at (10, \yPos-0.45) {\large $T = [5,7)$}; 
    \node at (10, \yPos-0.8)  {\large $\mathtt{\lambda_r} = a_2$}; 
    \node at (10, \yPos-1.2) {\large $\mathtt{\lambda_s} = \mathtt{null}$};     
   	\draw [fill=c3, opacity = 0.5] (8+0.05,\yPos+0.3) -- (8+0.05,\yPos - \windowHeight) --
   	(12-0.05,\yPos - \windowHeight) -- (12-0.05, \yPos + 0.3)  -- cycle;

    \node at (14, \yPos) {\large $F$=\bsq{chips}}; 
    \node at (14, \yPos-0.45)  {\large $T = [7,9)$}; 
     \node at (14, \yPos-0.8) {\large $\mathtt{\lambda_r} = \mathtt{null}$}; 
     \node at (14, \yPos-1.2) {\large $\mathtt{\lambda_s} = c_4$};     
   	\draw [fill=c3, opacity = 0.5] (12+0.05,\yPos+0.3) -- (12+0.05,\yPos - \windowHeight) --
   	(16-0.05,\yPos - \windowHeight) -- (16-0.05, \yPos + 0.3)  -- cycle;

	\draw[color=gray,dashed] (0,\yAxis) -- (0,\yPosB);

    \draw[color=gray,dashed] (2,\yAxis)   -- (2,\yPos+0.3);
    \draw[color=gray,dashed] (2,\yPos-\windowHeight-0.1)   -- (2,\yPosB);

    \draw[color=gray,dashed] (4,\yAxis)   -- (4,\yPosA+0.3);
    \draw[color=gray,dashed] (4,\yPosA-\windowHeight-0.1)   -- (4,\yPosB);

	\draw[color=gray,dashed] (6,\yAxis) -- (6,\yPosB);

    \draw[color=gray,dashed] (8,\yAxis)   -- (8,\yPosA+0.3);
    \draw[color=gray,dashed] (8,\yPosA-\windowHeight-0.1)   -- (8,\yPosB);

    \draw[color=gray,dashed] (10,\yAxis)  -- (10,\yPos+0.3);
    \draw[color=gray,dashed] (10,\yPos-\windowHeight-0.1)   -- (10,\yPosB);

    \draw[color=gray,dashed] (12,\yAxis)   -- (12,\yPosA+0.3);
    \draw[color=gray,dashed] (12,\yPosA-\windowHeight-0.1)   -- (12,\yPosB);

    \draw[color=gray,dashed] (14,\yAxis)  -- (14,\yPos+0.3);
    \draw[color=gray,dashed] (14,\yPos-\windowHeight-0.1)   -- (14,\yPosB);

    \draw[color=gray,dashed] (16,\yAxis)   -- (16,\yPosA+0.3);
    \draw[color=gray,dashed] (16,\yPosA-\windowHeight-0.1)   -- (16,\yPosB);

	\draw[color=gray,dashed] (18,\yAxis) -- (18,\yPosB);

    \node at (\xAxisOne-0.8, \yPos+0.5) {\large $\mathbf{W}(\mathbf{a}, \mathbf{c})$}; 
 	\draw [dashed]
	(\xAxisOne,\yPosA+0.45) -- (\xAxisOne,\yPosB+0.2)   -- 
	(\xAxisTwo,\yPosB+0.2)  -- (\xAxisTwo, \yPosA+0.45) -- cycle;

  \end{tikzpicture}}

\caption{Lineage-Aware Temporal Windows $\mathbf{W}(\mathbf{a}, \mathbf{c})$}
  \label{fig:windowsAC}
\end{figure*}

  The result of all TP set operations includes facts whose
  probability is computed over maximal intervals, i.e., intervals
  during which the same input tuples are valid. The computation of
  such intervals in temporal databases is performed by adjusting the
  intervals of each input relation based on the tuples of the other
  input relation that are valid. Combining the adjusted results to
  identify the intervals when, for example, tuples of both relations
  are valid~\cite{DignosBG14}, and concatenating their lineages for
  probability computation~\cite{DyllaMT13, DignosBG14} must be
  performed with joins.  In this section, we introduce the
  \emph{lineage-aware temporal window}, a novel mechanism that
  directly associates candidate output intervals with the lineage
  expressions of the valid input tuples of both relations. We show
  that a window contains all the information to produce an output
  tuple of a TP set operation $op^\kat{Tp}$, and that the set of all
  windows is a common core based on which all set operations can be
  computed using simple filtering and lineage-concatenation functions.

  A lineage-aware temporal window has schema ($F$, $T$,
  $\mathtt{\lambda_r}$, $\mathtt{\lambda_s}$). $F$ is a fact included
  in tuples over interval $T$. $\mathtt{\lambda_r}$ and
  $\mathtt{\lambda_s}$ are the lineage expressions of the input tuples
  of the left input relation $\mathbf{r}$ and the right input relation
  $\mathbf{s}$, respectively, which are valid over
  $[\mathtt{winTs}, \mathtt{winTe})$ and include $F$. \\

\begin{definition} \label{def:line} {\it 
(Lineage-Aware Windows) Let ${\bf r}$ and ${\bf s}$ be TP
relations with schema ($F$, $\lambda$, $T$, $p$). The set
of lineage-aware windows ${\bf W}({\bf r} , {\bf s})$ of
${\bf r}$ with respect to ${\bf s}$ with schema 
($F$, $T$, $\lambda_r$, $\lambda_s$) is defined as follows:}
  \begin{align*}
    \tilde{w} \in \mathbf{W} \Longleftrightarrow\ 
    & \forall t \in \tilde{w}.T ( \ 
	 (\lambda^{\mathbf{r}, \tilde{w}.F}_{t} \neq \mathtt{null} \ \lor \
      \lambda^{\mathbf{s}, \tilde{w}.F}_{t} \neq \mathtt{null}) \ \land \\ 
    & \hspace*{1.5cm}
     (\tilde{w}.\lambda_r = \lambda^{\mathbf{r}, \tilde{w}.F}_{t} \land \ 
      \tilde{w}.\lambda_s = \lambda^{\mathbf{s}, \tilde{w}.F}_{t})\ )
     \ \land \\      
    & \forall t' \notin \tilde{w}.T 
     (\tilde{w}.\lambda_r \neq \lambda^{\mathbf{r}, \tilde{w}.F}_{t} \lor \ 
      \tilde{w}.\lambda_s \neq \lambda^{\mathbf{s}, \tilde{w}.F}_{t})
    \\[2pt]
  \end{align*}
\end{definition}

For a window $\tilde{w}$ to be created over $\tilde{w}.T$, at least a
tuple of one of the input relations must be valid (Line 1). Each window
$\tilde{w}$ in ${\bf W}({\bf r}, {\bf s})$ spans over the interval or a
subinterval of a tuple $r$ in ${\bf r}$ or a tuple $s$ in ${\bf s}$
that include the fact $\tilde{w}.F$ and as stated in the second line of
the definition these tuples will determine $\tilde{w}.\lambda_r$ and
$\tilde{w}.\lambda_s$ respectively.  Finally, according to line 3 of
Definition~\ref{def:line}, the interval of window $\tilde{w}$ is a
maximal subinterval of an input tuple. In other words, at every time
point outside the $\tilde{w}.T$, either an input tuple that was valid
over $\tilde{w}.T$ stops being valid or an input tuple that was not
valid over $\tilde{w}.T$ starts being valid.

\begin{example}
{
  In Fig.~\ref{fig:windowsAC}, the TP relations ${\bf a}$ and
  ${\bf c}$ of Fig.~\ref{fig:tpdb} are illustrated along with the
  lineage-aware temporal windows of these two relations.  Different
  colors are used for different facts: green for \bsq{milk}, red for
  \bsq{chips}, and blue for \bsq{dates}.  A rectangle represents a
  window, filled in the color of the tuples including the
  corresponding fact. The window $w_1$ = (\bsq{milk}, [1,2), $c_1$,
  $\mathtt{null}$) is colored green since it includes the fact
  $w_1.F$ = \bsq{milk}. It indicates that, over interval [1,2), fact
  \bsq{milk} is included in tuple $c_1$ of relation ${\bf c}$
  ($w_1.\lambda_r$ = $c_1$) but in no tuple of relation
  ${\bf a}$ ($w_1.\lambda_s$ = $\mathtt{null}$). The window
  $w_1$ only spans the maximal interval [1,2), since at time point
  $t=2$, tuple $a_1$ starts being valid and thus, there is a change
  in the tuples of the two relations that are valid at $t=2$ and
  include fact \bsq{milk}.}
\end{example}

\smallskip
\begin{theorem}
\label{thm:mapping} {\em Let ${\bf r}$ and ${\bf s}$ be TP
relations with schema ($F$, $\lambda$, $T$, $p$), $op^\kat{Tp}$ a
TP set operation, and ${\bf W}({\bf r} , {\bf s})$ the lineage-aware
windows of ${\bf r}$ and ${\bf s}$. Given the output of the TP
set-operation ${\bf r} \ op^\kat{Tp} \ {\bf s}$, there exists a
window $w$ in $W$ that contains all the necessary information to
produce a tuple $u$ in ${\bf r} \ op^\kat{Tp} \ {\bf s}$.}
\end{theorem}

\begin{proof}
  We assume that $op^\kat{Tp}$ is a TP set-intersection
  ($\cap^\kat{Tp}$) and $u$ is an output tuple in
  ${\bf r} \ \cap^\kat{Tp} \ {\bf s}$. According to the definition of
  this operation and since, at each time point, only one tuple of each
  relation can include a fact, at each time point in $u.T$, there is
  exactly one tuple of ${\bf r}$ and one ${\bf s}$ valid and include
  $u.F$. Each window in ${\bf W}({\bf r}, {\bf s})$ records, for each
  fact $F$ and time point $t$, the tuples of each relation that
  include $F$ at $t$. Thus, windows are only created over time points
  when there is at least one valid input tuple. In order for $u$ to
  map to at least one window $w \in {\bf W}$, there must exist a
  window $w$ with the same fact ($u.F = w.F$) and interval
  ($u.T = w.T$) as $u$, and for which it holds that
  $w.\lambda_r = \lambda^{\mathbf{r}, u.F}_{t}$ and
  $w.\lambda_s = \lambda^{\mathbf{s}, u.F}_{t}$. Assuming that there
  is no such window, i.e., assuming that one of the above mentioned
  conditions is not satisfied, we conclude that there are
  no valid tuples including $u.F$ or the interval $u.T$ is not maximal.
  This contradicts our initial assumption of $u$ being a valid output
  tuple and of exactly one tuple of ${\bf r}$ and one ${\bf s}$ being
  valid over $u.T$ and including $u.F$. Consequently, there is at least
  one window $w \in {\bf W}$ to which we can map $u$.  In turn, we
  assume that $u$ maps to two windows $w_1$ and $w_2$ of ${\bf W}$. 
  This means that $u$ has the same fact and interval with both $w_1$
  and $w_2$ and that $w_1.\lambda_r = \lambda^{\mathbf{r}, u.F}_{t} =
  w_2.\lambda_r$ and $w_1.\lambda_s = \lambda^{\mathbf{s}, u.F}_{t} =
  w_2.\lambda_s$. Consequently, window $w_1$ coincides with $w_2$, and
  this proves that there is exactly one window $w \in {\bf W}$ that
  contains all the information needed to produce an output tuple $u$
  for TP set-intersection.  Similarly, we can prove that the same
  holds for an output tuple of any TP set operation.
\end{proof}

The flexibility of \emph{lineage-aware temporal windows} relies on two
characteristics: the lineages of valid tuples of each input relation
are directly associated with a maximal interval, and they are
separately recorded.  These two characteristics allow for an efficient
computation of the output tuples by using simple filtering conditions
and lineage-concatenating functions instead of the additional joins
performed in related approaches~\cite{DyllaMT13, DignosBG14}. Given a
TP set operation, $\mathtt{\lambda_r}$ and $\mathtt{\lambda_s}$ can be
used to determine whether fact $F$ and interval $[\mathtt{winTs}$,
$\mathtt{winTe})$ yield an output tuple.  If this is the case,
$\mathtt{\lambda_r}$ and $\mathtt{\lambda_s}$ are combined to the
lineage expression of this output tuple.

\begin{theorem}
\label{thm:reduction} {\em  Let ${\bf r}$ and ${\bf s}$ be TP
relations with schema ($F$, $\lambda$, $T$, $p$), $op^\kat{Tp}$
a TP set operation, and ${\bf W}({\bf r} , {\bf s})$ the set
of lineage-aware windows of ${\bf r}$ and ${\bf s}$. Given the
filtering conditions $\lambda_{filter}$ in Table~\ref{tab:filters}
and the lineage-concatenating functions $\lambda_{function}$ of
Definition~\ref{def:tpSetOpDef}, the computation of $op^\kat{Tp}$
is reduced to:

\begin{equation} {\bf r} \ op^\kat{Tp} \ {\bf s} \ = \ \pi_{F, T,
    \lambda_{function}(\lambda_r, \lambda_s)}(
  \sigma_{\lambda_{filter}} ({\bf W}({\bf r}, {\bf s})))
\end{equation}

\begin{table}[!h]\centering
  \caption{Definition of filtering conditions.}
  \label{tab:filters}
  \begin{tabular}{ M{0.4in} | M{1.2in} | M{0.8in}  @{}m{0pt}@{}}
	\hline
      $op^\kat{Tp}$ & $\lambda_{filter}$ & $\lambda_{function}$ & \\ [0.3cm]
      \hline
      ${\bf r} \ \cap^\kat{Tp} \ {\bf s}$ &
      $\lambda_r \neq \mathtt{null} \ \land \ \lambda_s \neq \mathtt{null}$ & 
      $\textit{\textbf{and}}(\lambda_r,\lambda_s)$ & \\ [0.3cm]
      ${\bf r} \ -^\kat{Tp} \ {\bf s}$ & 
      $\lambda_r \neq \mathtt{null}$ &
      $\textit{\textbf{andNot}}(\lambda_r,\lambda_s)$ &  \\ [0.3cm]
      ${\bf r} \ \cup^\kat{Tp} \ {\bf s}$ &
      $\lambda_r \neq \mathtt{null} \ \lor \ \lambda_s \neq \mathtt{null}$ &
      $\textit{\textbf{or}}(\lambda_r,\lambda_s)$ & \\ [0.3cm]
     \hline
    \end{tabular}
\end{table}
\vspace*{0.3cm}
}
\end{theorem}

\begin{proof}
  We assume that $op^\kat{Tp}$ is a TP set-intersection
  ($\cap^\kat{Tp}$), and a tuple $u$ that is produced by the algebraic
  expression $\pi_{F, T, \textit{\textbf{and}}(\lambda_r, \lambda_s)}
  (\sigma_{\lambda_r \neq \mathtt{null}\, \land\, \lambda_s \neq
  \mathtt{null}} ({\bf W}({\bf r}, {\bf s})))$. As a result, $u$ has
  been produced from a window in ${\bf W}({\bf r} , {\bf s})$ for
  which $w.\lambda_r \neq \mathtt{null}$ and $w.\lambda_s \neq
  \mathtt{null}$.  Also, $u.\lambda = and(w.\lambda_r, w.\lambda_s)$.
  Assuming that $u \notin {\bf r} \ \cap^\kat{Tp} \ {\bf s}$ means that
  one of the conditions in Def.~\ref{def:tpSetOpDef} for TP
  set-intersection is not satisfied. This is not possible, since $u$
  has been produced based on a window $w$ and thus for all time points
  in $u.T$ or equivalently in $w.T$, $\lambda^{\mathbf{r}, u.F}_{t}
  \neq \mathtt{null}$, $\lambda^{\mathbf{s}, u.F}_{t} \neq
  \mathtt{null}$ and $u.\lambda = and(\lambda^{\mathbf{r}, u.F}_{t},
  \lambda^{\mathbf{s}, u.F}_{t})$. Similarly, the contradiction can be
  shown for the time points outside $u.T$ and it can be shown that all 
  tuples in ${\bf r} \ \cap^\kat{Tp} \ {\bf s}$ are created based on
  the algebraic expression $\pi_{F, T, \lambda_{function}(\lambda_r,
  \lambda_s)} (\sigma_{\lambda_{filter}} ({\bf W}({\bf r}, {\bf s})))$.
  We can prove that the same holds for an output tuple of any TP set
  operation.
\end{proof}

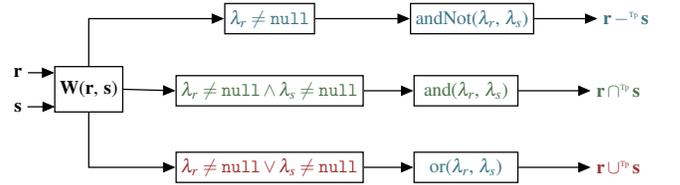
\begin{figure}[!h]
\tikzset{%
  algebra/.style    = {draw, thick, rectangle, minimum height = 2em,
    minimum width = 2em},
}
\tikzset{%
  lineage/.style    = {draw, thick, rectangle, minimum height = 2em,
    minimum width = 2.3cm},
}
\tikzset{%
  filter/.style    = {draw, thick, rectangle, minimum height = 2em,
    minimum width = 2.75cm},
}
\tikzset{%
  block/.style    = {draw, thick, rectangle, minimum height = 3em,
    minimum width = 3em},
}
\centering
\scalebox{0.6}{
\begin{tikzpicture}[auto, thick, node distance=2cm, >=triangle 45]
\draw
	node at (0,0)[right=-3mm] (relR) {\large $\mathbf{r}$}
	node [below = 0.3cm of relR] (relS) {\large $\mathbf{s}$}


	node [above right = 0.1cm and 0.8cm of relS.south, block] (relNR)
		 {\large {\bf W}($\mathbf{r}$, $\mathbf{s}$)}
 	node [right = 0.9cm of relR.west] (relNR1) {}
	node [right = 0.9cm of relS.west] (relNR2) {}

	node [above right = 0.7cm and 3cm of relNR.north, algebra] (setD)
	{\large \color{c2}{$\lambda_r \neq \mathtt{null}$}}
	
	node [below = 0.9cm of setD, algebra] (setI)
	{\large \color{c1}{$\lambda_r \neq \mathtt{null} \land \lambda_s \neq \mathtt{null}$}}

	node [below = 1cm of setI, algebra] (setU) {\large \color{c3}{$\lambda_r \neq \mathtt{null} \lor \lambda_s \neq \mathtt{null}$}}
	
	node [right = 1.1cm of setI, lineage] (setIL) {\large \color{c1}{and($\lambda_r$, $\lambda_s$)}}
	node [right = 2.1cm of setD, lineage] (setDL) {\large \color{c2}{andNot($\lambda_r$, $\lambda_s$)}}
	node [right = 1.1cm of setU, lineage] (setUL) {\large \color{c2}{or($\lambda_r$, $\lambda_s$)}}

	node [right =1.6cm of setUL] (resU) {\large \color{c3}{\large $\mathbf{r} \cup^\kat{Tp} \mathbf{s}$}}
	node [right =1.6cm of setIL] (resI) {\large \color{c1}{\large $\mathbf{r}\cap^\kat{Tp} \mathbf{s}$}}
	node [right =1.4cm of setDL] (resD) {\large \color{c2}{\large $\mathbf{r} -^\kat{Tp} \mathbf{s}$}};

\draw[->] (relR.east) -- (relNR1.center);
\draw[->] (relS.east) -- (relNR2.center);

\draw[->] (relNR.north) |- (setD.west);
\draw[->] (relNR.east) -- (setI.west);
\draw[->] (relNR.south) |- (setU.west);

\draw[->] (setI)  -- (setIL);
\draw[->] (setD) -- (setDL);
\draw[->] (setU) -- (setUL);

\draw[->] (setIL) -- (resI);
\draw[->] (setDL) -- (resD);
\draw[->] (setUL) -- (resU);

\end{tikzpicture}}
%
\caption{TP set operations using lineage-aware temporal windows.}
\label{fig:tpW}
\end{figure}

In Theorem~\ref{thm:reduction}, we reduce the computation of a TP set
operation ${\bf r} \ op^\kat{Tp} \ {\bf s}$ to the application of a
conventional projection and selection on the lineage-aware temporal
windows of ${\bf r}$ and ${\bf s}$. The filtering condition in the
selection as well as the lineage concatenating-function used in the
projection are directly derived from the definition of TP set
operations (Definition~\ref{def:tpSetOpDef}). The computation process
is illustrated in Fig.~\ref{fig:tpW}. In comparison to existing
temporal or probabilistic approaches used for set operations
(cf.\ Fig.~\ref{fig:temporalDBs} and Fig.~\ref{fig:probabilisticDB}),
the set of lineage-aware temporal windows constitutes a computational
core that only needs to be computed once and does not suffer from the quadratic complexity of previous approaches, as shown in Section~\ref{sec:sweepingIUE}.

\section{Lineage-Aware Window Advancer}
\label{sec:sweepingIUE}
In this section, we present the {\em lineage-aware window-advancer}
(LAWA), an algorithm that produces all lineage-aware temporal
windows of two TP relations. Each lineage-aware temporal window $w$
in ${\bf W}({\bf r}, {\bf s})$ records the lineage expression of the
tuple of each input relation that is valid over $w.T$ and that
includes $w.F$. Since the interval of each window is maximal, a new
window should be created when there is a change in the tuples of the
input relations that are valid and include a given fact. Such a
change only takes place when an input tuples starts or stops being
valid, i.e., at the starting and ending points of input intervals,
and this observation directly points to the use of a sweeping
technique.

\begin{center}
\begin{algorithm2e}[!htbp]
\small 
$(\mathtt{prevWinTe, currFact, rValid, sValid, r, s}) = \mathtt{status}$\;
\BlankLine
\If {$\mathtt{rValid}$ = $\mathtt{null}$ $\wedge$ $\mathtt{sValid}$ = $\mathtt{null}$} {  \label{line:newFactA}
    \If (\tcp*[f]{Case 1}){$\mathtt{r}=\mathtt{null} \wedge \mathtt{s}=\mathtt{null}$ } 
    {\label{line:bothNull}  
		\Return ($\mathtt{null}, \mathtt{null}$) 
	}
	

    \ElseIf (\tcp*[f]{Case 2}) {$\mathtt{r}=\mathtt{null} \wedge \mathtt{s} \neq \mathtt{null}$}
    {\label{line:oneNullA}
      $\mathtt{winTs}$ = $\mathtt{s.Ts}$;
      $\mathtt{currFact}$ = $\mathtt{s.F}$;
    } \ElseIf (\tcp*[f]{Case 3}){$\mathtt{r} \neq \mathtt{null} \wedge \mathtt{s}=\mathtt{null}$} {
      $\mathtt{winTs}$ = $\mathtt{r.Ts}$;
      $\mathtt{currFact}$ = $\mathtt{r.F}$;
\label{line:oneNullB}
    } \Else {
      \If  {$\mathtt{r}.F = \mathtt{currFact} \wedge
          \mathtt{s}.F \neq \mathtt{currFact}$} {\label{line:diffFactA}
           $\mathtt{winTs}$ = $\mathtt{r.Ts}$ \tcp*[f]{Case 4}} 
      \If {$\mathtt{r}.F \neq \mathtt{currFact} \wedge
          \mathtt{s}.F = \mathtt{currFact}$} {
           $\mathtt{winTs}$ = $\mathtt{s.Ts}$ \label{line:diffFactB}
           \tcp*[f]{Case 5}}
        \ElseIf (\tcp*[f]{Cases 6, 7}) {$\mathtt{r}.\mathtt{Ts} <  \mathtt{s}.\mathtt{Ts}$} {  \label{line:sameFactA}
      $\mathtt{winTs}$ = $\mathtt{r.Ts}$;
      $\mathtt{currFact}$ = $\mathtt{r.F}$;}
       \Else  { 
           $\mathtt{winTs}$ = $\mathtt{s.Ts}$;
	       $\mathtt{currFact}$ = $\mathtt{s.F}$;\label{line:sameFactB} 
	    } 
    }
}  \label{line:newFactB} 

  \lElse (\tcp*[f]{Case 8}){ $\mathtt{winTs}$ =  $\mathtt{prevWinTe}$   \label{line:adjacent} }
\BlankLine
\BlankLine

\If {$\mathtt{r} \neq \mathtt{null} \wedge \mathtt{r}.F = \mathtt{currFact} 
    \wedge \mathtt{r.Ts} = \mathtt{winTs}$  \label{line:rsValidA}} {
      $\mathtt{rValid}$ = $\mathtt{r}$;
      $\mathtt{r}$ = getNext($\mathtt{r}$);
}
\If {$\mathtt{s} \neq \mathtt{null} \wedge \mathtt{s}.F = \mathtt{currFact} 
    \wedge \mathtt{s.Ts} = \mathtt{winTs}$} {
      $\mathtt{sValid}$ = $\mathtt{s}$;
      $\mathtt{s}$ = getNext($\mathtt{s}$);
}\label{line:rsValidB}
\BlankLine
\BlankLine

$\mathtt{winTe}$ = min(minTs($\mathtt{r}$, $\mathtt{s}$), minTe($\mathtt{rValid}$, $\mathtt{sValid}$))\; \label{line:winTe}

\BlankLine
$\mathtt{\lambda_r}$ = $\mathtt{null}$; $\mathtt{\lambda_s}$ = $\mathtt{null}$; $\mathtt{window}$ = $\mathtt{null}$\; 
\BlankLine

\lIf {$\mathtt{rValid}$ $\neq$ $\mathtt{null}$  \label{line:lambdaA}}
{$\mathtt{\lambda_r}$ = $\mathtt{rValid}$.$\lambda$}
\lIf {$\mathtt{sValid}$ $\neq$ $\mathtt{null}$  \label{line:lambdaB}}
{$\mathtt{\lambda_s}$ = $\mathtt{sValid}$.$\lambda$}

\BlankLine
$\mathtt{window}$ = ($\mathtt{currFact}$, $\mathtt{winTs}$, $\mathtt{winTe}$, $\mathtt{\lambda_r}$ , $\mathtt{\lambda_s}$) \label{line:window} \;

\BlankLine
\BlankLine

 \lIf {$\mathtt{rValid}$ $\neq$ $\mathtt{null} \wedge \mathtt{rValid}$.$\mathtt{Te}$=$\mathtt{winTe}$}
 {$\mathtt{rValid}$ = $\mathtt{null}$}

 \lIf {$\mathtt{sValid}$ $\neq$ $\mathtt{null} \wedge \mathtt{sValid}$.$\mathtt{Te}$=$\mathtt{winTe}$}
 {$\mathtt{sValid}$ = $\mathtt{null}$}
\BlankLine

$\mathtt{prevWinTe}$=$\mathtt{winTe}$\; 
$\mathtt{status}$ = ($\mathtt{rValid, sValid, r, s, currFact, prevWinTe}$)\;
\BlankLine
\Return ($\mathtt{window}, \mathtt{status}$)\;
\caption{LAWA($\mathtt{status}$)}
\label{algo:lawa}
\end{algorithm2e}
\end{center}

In our approach, to produce all lineage-aware temporal windows, we
introduce LAWA, a sweeping algorithm we describe in Algorithm
\ref{algo:lawa}. Traditionally, sweeping algorithms use a vertical
sweepline, and they determine the output tuples based on the input
tuples that intersect with this sweepline \cite{Arge1998, 
PlatovICDE16}.  This works well for TP set intersection. However,
for TP set difference and set union, there are cases when the
interval of an output tuple is not determined only by the tuples
that intersect with the sweepline.  In order to handle such cases,
we use a \textit{sweeping window}.  The left and right boundaries
of the window correspond to the start and end points of a maximal
interval that is associated with a potential output interval.

LAWA processes the tuples of two duplicate-free TP relations
$\mathbf{r}$ and $\mathbf{s}$ with schema ($F$, $\lambda$, $T$, $p$)
that are sorted by their facts and starting points of their
intervals. It produces lineage-aware temporal windows whose left
($\mathtt{winTs}$) and right ($\mathtt{winTe}$) boundaries are
computed during a sweep of the start ($\mathtt{Ts}$) and end
($\mathtt{Te}$) points of the tuples.  The left boundary
$\mathtt{winTs}_i$ of a window $i$ is greater or equal to
$\mathtt{winTe}_{i-1}$ of the previous window. Its right boundary
$\mathtt{winTe}_{i}$ is the smallest among the end points of the
tuples expected to overlap with this window, i.e., tuples with
$\mathtt{Ts} \leq \mathtt{winTs}$ and $\mathtt{Te} >
\mathtt{winTs}$, and the start points of the tuples of the two
relations to be processed next.

The input of LAWA is a structure ($\mathtt{status}$) with the
necessary status information: the right boundary of the last
candidate window ($\mathtt{prevWinTe}$), the fact that is
currently being processed ($\mathtt{currFact}$), the current
tuples of $\mathbf{r}$ ($\mathtt{rValid}$) and $\mathbf{s}$
($\mathtt{sValid}$) that are valid over the sweeping window
$[\mathtt{winTs}, \mathtt{winTe})$, and the next tuples of
relations $\mathbf{r}$ ($\mathtt{r}$) and $\mathbf{s}$
($\mathtt{s}$).  All variables are initialized to $\mathtt{null}$
except for $\mathtt{r}$ and $\mathtt{s}$ that are initialized
to the first tuples of the corresponding relations. The value of
$\mathtt{prevWinTe}$ is initialized to $-1$.

\input{7a_lawa_figure.tex}

Initially, the left boundary $\mathtt{winTs}$ of the new window
is determined, and the cases considered are described in
Fig.~\ref{fig:algoCases}. If at least one tuple is valid
(Fig.~\ref{fig:still_valid}), the new window is adjacent to the
previous one, with $\mathtt{winTs}$ = $\mathtt{prevWinTe}$
(Case 8, Line~\ref{line:adjacent}). Otherwise, $\mathtt{winTs}$, and
potentially $\mathtt{currFact}$, are determined by the new tuples.
Five possible scenarios exist: (a) both relations have been scanned
(Case 1, Line \ref{line:bothNull}), (b) one of the two relations has
already been scanned (Cases 2 and 3, Lines \ref{line:oneNullA}--
\ref{line:oneNullB}), (c) there are available tuples from both
$\mathbf{r}$ and $\mathbf{s}$, but only one includes the same
fact as $\mathtt{currFact}$ (Cases 4 and 5,
Lines\ref{line:diffFactA}--\ref{line:diffFactB}), (d) there are
available tuples from both $\mathbf{r}$ and $\mathbf{s}$ and they
either both include different facts from $\mathtt{currFact}$ or
the same fact as $\mathtt{currFact}$, making two starting points
as candidates for $windTs$ (Cases 6 and 7,
Lines \ref{line:sameFactA}--\ref{line:sameFactB}).

Since the input relations are duplicate-free, i.e., no two tuples of
the same relation can include the same fact and be valid at the same
time point, $\mathtt{rValid}$ and $\mathtt{sValid}$ correspond to
exactly one input tuple each.  If $\mathtt{rValid}$ and
$\mathtt{sValid}$ are not $\mathtt{null}$, they correspond to tuples 
that were also overlapping with the previous window.  Otherwise,
they need to be updated to $\mathtt{r}$ or $\mathtt{s}$ if the
latter include a fact equal to $\mathtt{currFact}$ and have a start
point equal to $\mathtt{winTs}$ (Lines
\ref{line:rsValidA}--\ref{line:rsValidB}). The right boundary
$\mathtt{winTe}$ is updated to the minimum time point among the end
points of $\mathtt{rValid}$ and $\mathtt{sValid}$ and the current
start points of $\mathtt{r}$ and $\mathtt{s}$, i.e., the next tuples
to be processed (Line \ref{line:winTe}).  Here, the tuples
$\mathtt{r}$ and $\mathtt{s}$ must be considered because the start
point of an unprocessed tuple marks a change in the tuples that are
valid over that interval.


After $\lambda_r$ and $\lambda_s$ are extracted from
$\mathtt{rValid}$ and $\mathtt{sValid}$
(Lines \ref{line:lambdaA}--\ref{line:lambdaB}), all the information
for the creation of a lineage-aware temporal window is recorded
(Line \ref{line:window}). $\mathtt{rValid}$ and $\mathtt{sValid}$
are updated for the next call of LAWA based on whether the tuples
they correspond to are still valid outside the window, i.e., when
the end points of these tuples are larger than $\mathtt{winTe}$. 
Finally, LAWA also returns its $\mathtt{status}$, which is used in
the implementation of the actual TP set operations.

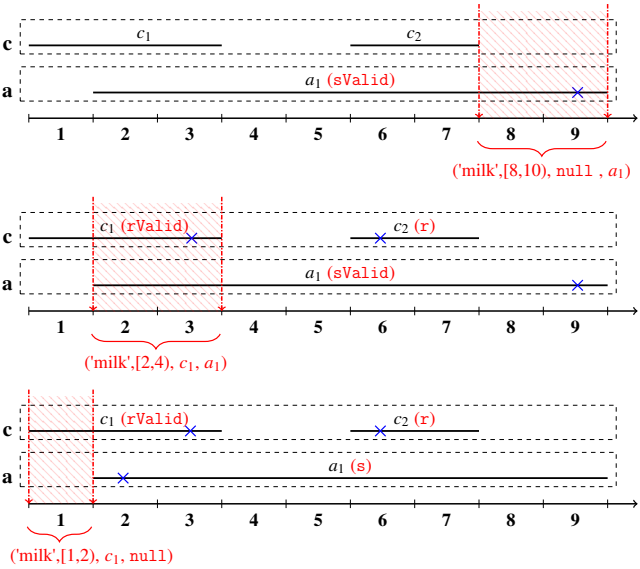
\begin{figure}[!htbp]
\centering
\scalebox{0.57} {
\tikzset{cross/.style={cross out, draw=blue, thick, minimum size=7pt, inner sep=0pt, outer sep=0pt},
cross/.default={1pt}}
\begin{tikzpicture}   

 \pgfmathsetmacro{\xAxisOne}{-0.2}
 \pgfmathsetmacro{\xAxisTwo}{13.7}
 \pgfmathsetmacro{\relationwindowHeight}{1}
\pgfmathsetmacro{\gap}{1.8}


\pgfmathsetmacro{\yAxis}{0}
\pgfmathsetmacro{\yPosA}{\yAxis  + 0.6}
\pgfmathsetmacro{\yPosB}{\yAxis  + 2*0.7 + 0.3}
\pgfmathsetmacro{\yPosC}{\yAxis  + 3*0.7 + 0.3}

 \draw (0,\yAxis) [->, line width = 1pt]
    -- coordinate (x axis mid) (14.2,\yAxis);

\pgfmathsetmacro{\shift}{0.5}
\foreach \j in {1,...,10}{
     \pgfmathsetmacro{\divRes}{int(\j/2)}
     \pgfmathsetmacro{\modRes}{1-(\j-\divRes*2)}
     \pgfmathsetmacro{\xPos}{\j-1+(\divRes-\modRes*\shift)}
     \draw (\xPos , \yAxis + 0.1)  --  (\xPos , \yAxis-0.1) node[anchor=north,font=\relsize{2}] {};
}

\pgfmathsetmacro{\shift}{0.5}
\foreach \j in {1,...,9}{
     \pgfmathsetmacro{\divRes}{int(\j/2)}
     \pgfmathsetmacro{\modRes}{1-(\j-\divRes*2)}
     \pgfmathsetmacro{\xPos}{\j-1+(\divRes-\modRes*\shift)}
     \draw (\xPos+0.75 , \yAxis-0.1) node[anchor=north,font=\relsize{1}] {\bf \j};
}

\draw [line width=1.2]
(0,   \yPosB) -- (4.5, \yPosB)
node[pos=0.6,above=-1pt]
{\large $c_1$\ {\color{red}($\mathtt{rValid}$)}};
\draw (4.47-.7,\yPosB) node[cross] {};
\draw [line width=1.2]
(7.5, \yPosB) -- (10.5,\yPosB)
node[pos=.5,above=-1pt]
{\large $c_2$\ {\color{red} ($\mathtt{r}$)}}; 
\draw (7.5+.7,\yPosB) node[cross] {};

\pgfmathsetmacro{\relationYAxis}{\yPosB}
\node at (\xAxisOne-0.3, \relationYAxis) {\Large $\mathbf{c}$}; 
\pgfmathsetmacro{\relationYAxis}{\relationYAxis-0.2}
\draw [dashed]
(\xAxisOne,\relationYAxis) -- (\xAxisOne,\relationYAxis + \relationwindowHeight - 0.2) -- (\xAxisTwo,\relationYAxis + \relationwindowHeight - 0.2) -- (\xAxisTwo, \relationYAxis + 0)  -- cycle;

\draw [line width=1.2]
(1.5, \yPosA) -- (13.5,\yPosA)
node[pos=.5,above=-1pt]
{\large $a_1$\ {\color{red} ($\mathtt{s}$)}}; 
\draw (1.5+.7,\yPosA) node[cross] {};

\pgfmathsetmacro{\relationYAxis}{\yPosA}
\node at (\xAxisOne-0.3, \relationYAxis) {\Large $\mathbf{a}$}; 
\pgfmathsetmacro{\relationYAxis}{\relationYAxis-0.2}
\draw [dashed]
(\xAxisOne,\relationYAxis) -- (\xAxisOne,\relationYAxis + \relationwindowHeight - 0.2) -- (\xAxisTwo,\relationYAxis + \relationwindowHeight - 0.2) -- (\xAxisTwo, \relationYAxis + 0)  -- cycle;


\pgfmathsetmacro{\winTs}{0}
\pgfmathsetmacro{\winTe}{1.5}
\pgfmathsetmacro{\windowHeight}{2.7}

\draw [<-,line width=1,densely dashdotted,red]
(\winTs, \yAxis) -- (\winTs,\yAxis+\windowHeight)
node[pos=1.2]
{};

\draw [<-,line width=1,densely dashdotted,red]
(\winTe, \yAxis) -- (\winTe,\yAxis +\windowHeight)
node[pos=1.2]
{};

\draw [ultra thick, draw=none, pattern color=red, pattern=north west lines, opacity=0.4]
(\winTs,\yAxis) -- (\winTs,\yAxis + \windowHeight - 0.2) -- (\winTe,\yAxis + \windowHeight - 0.2) -- (\winTe, \yAxis + 0)  -- cycle;

\draw [thick,red,decorate,
decoration={brace,amplitude=10pt,mirror},
xshift=-1pt,yshift=-1pt]
(\winTs,\yAxis-0.5) -- (\winTe,\yAxis-0.5)
node [red,pos=1,yshift=-0.7cm] 
{\large (\bsq{milk},[1,2), $c_1$, $\mathtt{null}$)};


\pgfmathsetmacro{\yAxis}{\yAxis +\windowHeight + \gap}
\pgfmathsetmacro{\yPosA}{\yAxis  + 0.6}
\pgfmathsetmacro{\yPosB}{\yAxis  + 2*0.7 + 0.3}
\pgfmathsetmacro{\yPosC}{\yAxis  + 3*0.7 + 0.3}

 \draw (0,\yAxis) [->, line width = 1pt]
    -- coordinate (x axis mid) (14.2,\yAxis);

\pgfmathsetmacro{\shift}{0.5}
\foreach \j in {1,...,10}{
     \pgfmathsetmacro{\divRes}{int(\j/2)}
     \pgfmathsetmacro{\modRes}{1-(\j-\divRes*2)}
     \pgfmathsetmacro{\xPos}{\j-1+(\divRes-\modRes*\shift)}
     \draw (\xPos , \yAxis + 0.1)  --  (\xPos , \yAxis-0.1) node[anchor=north,font=\relsize{2}] {};
}

\pgfmathsetmacro{\shift}{0.5}
\foreach \j in {1,...,9}{
     \pgfmathsetmacro{\divRes}{int(\j/2)}
     \pgfmathsetmacro{\modRes}{1-(\j-\divRes*2)}
     \pgfmathsetmacro{\xPos}{\j-1+(\divRes-\modRes*\shift)}
     \draw (\xPos+0.75 , \yAxis-0.1) node[anchor=north,font=\relsize{1}] {\bf \j};
}

\draw [line width=1.2]
(0,   \yPosB) -- (4.5, \yPosB)
node[pos=0.6,above=-1pt]
{\large $c_1$\ {\color{red}($\mathtt{rValid}$)}};
\draw (4.5-.7,\yPosB) node[cross] {};

\draw [line width=1.2]
(7.5, \yPosB) -- (10.5,\yPosB)
node[pos=.5,above=-1pt]
{\large $c_2$\ {\color{red} ($\mathtt{r}$)}}; 
\draw (7.5+.7,\yPosB) node[cross] {};

\pgfmathsetmacro{\relationYAxis}{\yPosB}
\node at (\xAxisOne-0.3, \relationYAxis) {\Large $\mathbf{c}$}; 
\pgfmathsetmacro{\relationYAxis}{\relationYAxis-0.2}
\draw [dashed]
(\xAxisOne,\relationYAxis) -- (\xAxisOne,\relationYAxis + \relationwindowHeight - 0.2) -- (\xAxisTwo,\relationYAxis + \relationwindowHeight - 0.2) -- (\xAxisTwo, \relationYAxis + 0)  -- cycle;

\draw [line width=1.2]
(1.5, \yPosA) -- (13.5,\yPosA)
node[pos=.5,above=-1pt]
{\large $a_1$\ {\color{red} ($\mathtt{sValid}$)}}; 
\draw (13.5-.7,\yPosA) node[cross] {};

\pgfmathsetmacro{\relationYAxis}{\yPosA}
\node at (\xAxisOne-0.3, \relationYAxis) {\Large $\mathbf{a}$}; 
\pgfmathsetmacro{\relationYAxis}{\relationYAxis-0.2}
\draw [dashed]
(\xAxisOne,\relationYAxis) -- (\xAxisOne,\relationYAxis + \relationwindowHeight - 0.2) -- (\xAxisTwo,\relationYAxis + \relationwindowHeight - 0.2) -- (\xAxisTwo, \relationYAxis + 0)  -- cycle;


\pgfmathsetmacro{\winTs}{1.5}
\pgfmathsetmacro{\winTe}{4.5}

\draw [<-,line width=1,densely dashdotted,red]
(\winTs, \yAxis) -- (\winTs,\yAxis+\windowHeight)
node[pos=1.2]
{};

\draw [<-,line width=1,densely dashdotted,red]
(\winTe, \yAxis) -- (\winTe,\yAxis +\windowHeight)
node[pos=1.2]
{};

\draw [ultra thick, draw=none, pattern color=red, pattern=north west lines, opacity=0.4]
(\winTs,\yAxis) -- (\winTs,\yAxis + \windowHeight - 0.2) -- (\winTe,\yAxis + \windowHeight - 0.2) -- (\winTe, \yAxis + 0)  -- cycle;

\draw [thick,red,decorate,
decoration={brace,amplitude=10pt,mirror},
xshift=-1pt,yshift=-1pt]
(\winTs,\yAxis-0.5) -- (\winTe,\yAxis-0.5)
node [red,midway,yshift=-0.7cm] 
{\large (\bsq{milk},[2,4), $c_1$, $a_1$)};


\pgfmathsetmacro{\yAxis}{\yAxis +\windowHeight + \gap}
\pgfmathsetmacro{\yPosA}{\yAxis  + 0.6}
\pgfmathsetmacro{\yPosB}{\yAxis  + 2*0.7 + 0.3}
\pgfmathsetmacro{\yPosC}{\yAxis  + 3*0.7 + 0.3}

 \draw (0,\yAxis) [->, line width = 1pt]
    -- coordinate (x axis mid) (14.2,\yAxis);

\pgfmathsetmacro{\shift}{0.5}
\foreach \j in {1,...,10}{
     \pgfmathsetmacro{\divRes}{int(\j/2)}
     \pgfmathsetmacro{\modRes}{1-(\j-\divRes*2)}
     \pgfmathsetmacro{\xPos}{\j-1+(\divRes-\modRes*\shift)}
     \draw (\xPos , \yAxis + 0.1)  --  (\xPos , \yAxis-0.1) node[anchor=north,font=\relsize{2}] {};
}

\pgfmathsetmacro{\shift}{0.5}
\foreach \j in {1,...,9}{
     \pgfmathsetmacro{\divRes}{int(\j/2)}
     \pgfmathsetmacro{\modRes}{1-(\j-\divRes*2)}
     \pgfmathsetmacro{\xPos}{\j-1+(\divRes-\modRes*\shift)}
     \draw (\xPos+0.75 , \yAxis-0.1) node[anchor=north,font=\relsize{1}] {\bf \j};
}

\draw [line width=1.2]
(0,   \yPosB) -- (4.5, \yPosB)
node[pos=0.6,above=-1pt]
{\large $c_1$};

\draw [line width=1.2]
(7.5, \yPosB) -- (10.5,\yPosB)
node[pos=.5,above=-1pt]
{\large $c_2$}; 

\pgfmathsetmacro{\relationYAxis}{\yPosB}
\node at (\xAxisOne-0.3, \relationYAxis) {\Large $\mathbf{c}$}; 
\pgfmathsetmacro{\relationYAxis}{\relationYAxis-0.2}
\draw [dashed]
(\xAxisOne,\relationYAxis) -- (\xAxisOne,\relationYAxis + \relationwindowHeight - 0.2) -- (\xAxisTwo,\relationYAxis + \relationwindowHeight - 0.2) -- (\xAxisTwo, \relationYAxis + 0)  -- cycle;

\draw [line width=1.2]
(1.5, \yPosA) -- (13.5,\yPosA)
node[pos=.5,above=-1pt]
{\large $a_1$\ {\color{red} ($\mathtt{sValid}$)}}; 
\draw (13.5-.7,\yPosA) node[cross] {};

\pgfmathsetmacro{\relationYAxis}{\yPosA}
\node at (\xAxisOne-0.3, \relationYAxis) {\Large $\mathbf{a}$}; 
\pgfmathsetmacro{\relationYAxis}{\relationYAxis-0.2}
\draw [dashed]
(\xAxisOne,\relationYAxis) -- (\xAxisOne,\relationYAxis + \relationwindowHeight - 0.2) -- (\xAxisTwo,\relationYAxis + \relationwindowHeight - 0.2) -- (\xAxisTwo, \relationYAxis + 0)  -- cycle;


\pgfmathsetmacro{\winTs}{10.5}
\pgfmathsetmacro{\winTe}{13.5}

\draw [<-,line width=1,densely dashdotted,red]
(\winTs, \yAxis) -- (\winTs,\yAxis+\windowHeight)
node[pos=1.2]
{};

\draw [<-,line width=1,densely dashdotted,red]
(\winTe, \yAxis) -- (\winTe,\yAxis +\windowHeight)
node[pos=1.2]
{};

\draw [ultra thick, draw=none, pattern color=red, pattern=north west lines, opacity=0.4]
(\winTs,\yAxis) -- (\winTs,\yAxis + \windowHeight - 0.2) -- (\winTe,\yAxis + \windowHeight - 0.2) -- (\winTe, \yAxis + 0)  -- cycle;

\draw [thick,red,decorate,
decoration={brace,amplitude=10pt,mirror},
xshift=-1pt,yshift=-1pt]
(\winTs,\yAxis-0.5) -- (\winTe,\yAxis-0.5)
node [red,midway,yshift=-0.7cm] 
{\large (\bsq{milk},[8,10), $\mathtt{null}$ , $a_1$)};
\end{tikzpicture}}
\caption{Three calls of LAWA for the input relations $\mathbf{c}$ and $\mathbf{a}$.}
\label{fig:algoRun}
\end{figure}

\begin{example}
In Fig.~\ref{fig:algoRun}, we illustrate three calls of LAWA with
the left and right relations being $\mathbf{c}$ and $\mathbf{a}$ of
Fig.\ref{fig:tpdbRelations}, respectively. Before the first call,
the input relations have been sorted by their facts and start
points. The time points used to determine the right boundary of a
window are annotated with a blue cross. In the first call of LAWA,
illustrated at the bottom, the left and right boundary of the window
are set to $\mathtt{winTs} = 1$ and $\mathtt{winTe} = 2$,
respectively. After $\mathtt{winTs}$ is determined, the only tuple
valid is $\mathtt{rValid} = c_1$. Thus, given that there is no valid
tuple in $\mathbf{a}$ yet, $\mathtt{winTe}$ is set to the start
point of $a_1$, i.e., the next tuple of $\mathbf{a}$ to be
processed. This time point is smaller than the end point
$\mathtt{Te} = 4$ of $\mathtt{rValid}$ or the start point
$\mathtt{Ts} = 6$ of the upcoming tuple of $\mathbf{c}$ ($c_2$).
In the second call of LAWA, illustrated in the middle, the left
boundary of the next window to be examined is equal to the right
boundary of the previous window, i.e., $\mathtt{winTs} = 2$, given
that the fact (\bsq{milk}) is still being processed. The tuples
valid after time point $t=2$ are $\mathtt{rValid} = c_1$ and
$\mathtt{sValid} = a_1$.  The right boundary of the window is the
minimum of $\mathtt{rValid}.\mathtt{Te} = 4$,
$\mathtt{sValid}.\mathtt{Te} = 10$ and $c_2.\mathtt{Ts} = 6$, and 
thus $\mathtt{winTe} = 4$. A similar pattern goes on until the last
call of LAWA, illustrated on the top of Fig.~\ref{fig:algoRun},
where $\mathtt{winTs} = 8$ and $\mathtt{winTe} = 10$. Then,
$\mathtt{rValid}$ and $\mathtt{sValid}$ are set to $\mathtt{null}$
and no further windows are produced.
\end{example}

\section{Basic TP Set Algorithms}
\label{sec:basicIUE}
In this section, we implement all TP-set operations by exploiting
the flexibility of {\em lineage-aware temporal windows} that enable
finalizing output lineages  and filtering out output intervals when
they are produced, thus avoiding redundant computations that occur
when these two steps are decoupled~\cite{DyllaMT13, DignosTODS16}. 
Based on Theorem~\ref{thm:reduction}, we reduce the implementation
of TP set operations into a four-step process
(Fig.~\ref{fig:algorithmSketch}). The sorting step is
a prerequisite for the creation of windows using LAWA. When a
window is created, a lineage-based filter ($\lambda_{filter}$) is
directly applied. The $\lambda_{filter}$ is different for each TP set
operation.  In contrast to previous works of either temporal or
probabilistic set operations, this step involves no application of
additional algebraic operations, no tuple replication and no redundant
interval comparisons. After the filtering step, the final lineage
expression of an output tuple is created by applying the
lineage-concatenating function ($\lambda_{function}$) of the respective
TP set operation (Def.~\ref{def:tpSetOpDef}) on $\mathtt{\lambda_r}$
and $\mathtt{\lambda_s}$.

\begin{figure}[!htbp]
\centering
\scalebox{0.8}{
\tikzstyle{block} = [draw, rectangle, 
    minimum height=3em, minimum width=3em]

\tikzstyle{input} = [coordinate]
\tikzstyle{output} = [coordinate]

\begin{tikzpicture}[node distance=2.2cm]

    \node [input] at (0,0) (rs)  {};
    \node [block, right = 1.5cm of rs, text width=1cm, align=center] (sort)
    {sort};
    \node [block, right of =sort, text width=1cm, align=center] (allInfo)
    {LAWA};

    \node [block, right of=allInfo,  text width=1cm, align=center] (filter) {$\lambda_{filter}$};
    \node [block, right of=filter,  text width=1.2cm, align=center] (lfunction) {$\lambda_{function}$};

    \node [output, right=1cm of lfunction] (output) {};

    \draw [->] (rs) -- node[above=0.05cm] {$\mathbf{r},\mathbf{s}, op$} (sort.west);
    \draw [->] (sort.east) -- node {} (allInfo.west);
    \draw [->] (allInfo.east) -- node {} (filter.west);
    \draw [->] (filter) -- node {} (lfunction);
    \draw [->] (lfunction) -- (output);
    
	\node [left = 0.3cm of lfunction] (y) {};
	\node [right = 0.3cm of lfunction] (y3) {};

	\node [right = 0.3cm of sort] (y2) {};
	\node [below = 0.6cm of y] (y1) {};

    \draw [->] (y.center) -- (y1.center);
    \draw [->] (y3.center) |- (y1.center);
    \draw [->] (y1.center) -| (y2.center);
\end{tikzpicture}}
\caption{Process overview.}
\label{fig:algorithmSketch}
\end{figure}
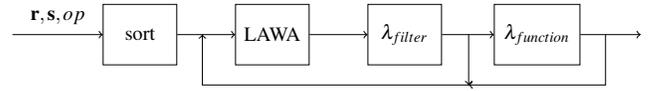

The algorithms \emph{Intersect}($\mathbf{r}$, $\mathbf{s}$),
\emph{Union}($\mathbf{r}$, $\mathbf{s}$) and \emph{Except}($\mathbf{r}$,
$\mathbf{s}$) correspond to $\mathbf{r}\cap^\kat{Tp}\mathbf{s}$,
$\mathbf{r}\cup^\kat{Tp}\mathbf{s}$ and $\mathbf{r}-^\kat{Tp}\mathbf{s}$, 
respectively. In all algorithms, input relations are initially sorted based 
on their facts $F$ and start points $\mathtt{Ts}$ (Line~\ref{line:sort})
when the status of LAWA is initialized. As long as the terminating condition
(Line~\ref{line:terminating}) is satisfied, LAWA passes through all start
and end points in a smaller-to-larger fashion and produces candidate
windows (Line~\ref{line:candidates}). The windows produced by LAWA are
filtered based on the lineages of the tuples that are valid during the 
interval it covers (Line~\ref{line:filtering}). The filter used for each
operation, as well as the terminating condition and the
lineage-concatenating function, directly stem from the definitions of the
operation. For example, in the case of set difference
$\mathbf{r}-^\kat{Tp}\mathbf{s}$,  windows are produced as long as
there are tuples in the outer relation (i.e., while $\mathtt{r}
\neq \mathtt{null}$). The interval of a lineage-aware temporal window
corresponds to an output tuple only if there is a tuple of the outer 
relation that is valid over $[\mathtt{winTs}, \mathtt{winTe})$
(i.e., when $\mathtt{\lambda_r} \neq \mathtt{null}$).

For \emph{Union}($\mathbf{r}$, $\mathbf{s}$) and
\emph{Except}($\mathbf{r}$, $\mathbf{s}$), when the while-loop
terminates, there might still be one more window, corresponding to the
subinterval of the last valid tuple of $\mathbf{r}$
($\mathtt{rValid}$) or the last valid tuple of $\mathbf{s}$
($\mathtt{sValid}$). Thus, LAWA is called one more time
(Line~\ref{line:extraCall}).

{\setlength{\algomargin}{0pt} 
\begin{algorithm2e}[!htbp]\small 
  \emph{sort}($\mathbf{r}$\{$F,\mathtt{Ts}$\});  \label{line:sort}
  \emph{sort}($\mathbf{s}$\{$F,\mathtt{Ts}$\})\;
  $\mathtt{status} = (\mathtt{-1,null,null,null,fetchRow(\mathbf{r}),fetchRow(\mathbf{s})})$\;
  \While{$\mathtt{status.r} \neq \mathtt{null} \wedge \mathtt{status.s} \neq \mathtt{null}$}{
    \label{line:terminating}
    $(\mathtt{w}, \mathtt{status})$  = LAWA($\mathtt{status}$)\;  \label{line:candidates}
    \If{$\mathtt{w.\lambda_r} \neq \mathtt{null}$ $\wedge$
      $\mathtt{w.\lambda_s} \neq \mathtt{null}$} {
      \label{line:filtering}
        $\mathtt{o}$ =
        $\mathtt{o}$ $\cup$
        \{($F$, \textit{\textbf{and}}($\mathtt{w.\lambda_r}$, $\mathtt{w.\lambda_s}$),
                                     [$\mathtt{w.winTs}$, $\mathtt{w.winTe}$))\}\;
    }
  }
  \Return $\mathtt{o}$\;
  \caption{\emph{Intersect}($\mathbf{r}$, $\mathbf{s}$)}
  \label{fig:intersect}
\end{algorithm2e}
}


{\setlength{\algomargin}{0pt} 
\begin{algorithm2e}[!htbp] \small 
  \emph{sort}($\mathbf{r}$\{$F,\mathtt{Ts}$\});
  \emph{sort}($\mathbf{s}$\{$F,\mathtt{Ts}$\})\;
  $\mathtt{status} = (\mathtt{-1,null,null,null,fetchRow(\mathbf{r}),fetchRow(\mathbf{s})})$\;
  \While{$\mathtt{status.r} \neq \mathtt{null} \vee \mathtt{status.s} \neq \mathtt{null}$}{
    $(\mathtt{w, status})$  = LAWA($\mathtt{status}$)\;
    \If{$\mathtt{w.\lambda_r} \neq \mathtt{null}$ $\vee$
        $\mathtt{w.\lambda_s} \neq \mathtt{null}$} {
          \hspace*{-0.3cm}
          $\mathtt{o}$ = $\mathtt{o}$ $\cup$ \{($\mathtt{w}.F$,
          \textit{\textbf{or}}($\mathtt{w.\lambda_r}$, $\mathtt{w.\lambda_s}$),
          [$\mathtt{w.winTs}$, $\mathtt{w.winTe}$))\}\;
  	}
  }
  \If{$\mathtt{status.rValid} \neq \mathtt{null} \vee
       \mathtt{status.sValid} \neq \mathtt{null}$}{
     $(\mathtt{w, status})$  = LAWA($\mathtt{status}$)\;  \label{line:extraCall}
     $\mathtt{o}$ = $\mathtt{o}$ $\cup$ \{($\mathtt{w}.F$,
     \textit{\textbf{or}}($\mathtt{w.\lambda_r}$, $\mathtt{w.\lambda_s}$),
       [$\mathtt{w.winTs}$, $\mathtt{w.winTe}$))\}\;
  }
  \Return $\mathtt{o}$\;
  \caption{\emph{Union}($\mathbf{r}$, $\mathbf{s}$)}
  \label{fig:union}
\end{algorithm2e}
}


{\setlength{\algomargin}{0pt} 
\begin{algorithm2e}[!htbp]
\small 
\emph{sort}($\mathbf{r}$\{$F,\mathtt{Ts}$\});
\emph{sort}($\mathbf{s}$\{$F,\mathtt{Ts}$\})\;
$\mathtt{status} = (\mathtt{-1, null, null, null, fetchRow(\mathbf{r}), fetchRow(\mathbf{s})})$\;
\While{$\mathtt{status.r} \neq \mathtt{null}$}{
	$(\mathtt{w}, \mathtt{status})$  = LAWA($\mathtt{status}$)\;
	\If{$\mathtt{w.\lambda_r} \neq \mathtt{null}$} {
          \hspace*{-0.3cm}
          $\mathtt{o}$ = $\mathtt{o}$ $\cup$ \{($\mathtt{w.F}$,
          \textit{\textbf{andNot}}($\mathtt{w.\lambda_r}$,$\mathtt{w.\lambda_s}$),
          [$\mathtt{w.winTs}$, $\mathtt{w.winTe}$))\}\;
  	 }
}
\If{$\mathtt{status.rValid} \neq \mathtt{null}$}{
	$(\mathtt{w}, \mathtt{status})$  = LAWA($\mathtt{status}$)\;  
		 $\mathtt{o}$ = $\mathtt{o}$ $\cup$ \{($\mathtt{w.F}$, \textit{\textbf{andNot}}($\mathtt{w.\lambda_r}$, $\mathtt{w.\lambda_s}$), [$\mathtt{w.winTs}$, $\mathtt{w.winTe}$))\}\;
}
\Return $\mathtt{o}$\;
\caption{\emph{Except}($\mathbf{r}$, $\mathbf{s}$)}
\label{fig:except}
\end{algorithm2e}

\begin{example}
  In Fig.~\ref{fig:processSetDif}, we illustrate the computation of
  set difference
  $\sigma_\kat{F = \bsq{milk}}(\mathbf{c}) -^\kat{TP} \sigma_\kat{F =
    \bsq{milk}} (\mathbf{a})$ for relations $\mathbf{c}$ and
  $\mathbf{a}$ in Fig.~\ref{fig:tpdbRelations}.  The first candidate
  window $[1,2)$ has $\mathtt{\lambda_s} = \mathtt{null}$ and
  $\mathtt{\lambda_r} = c_1$.  For set difference the current window
  yields a result tuple, since, over interval $[1,2)$, the fact
  (\bsq{milk}) is included in a tuple of the left input relation
  $\mathbf{c}$ with lineage $\mathtt{\lambda_s} = c_1$.  In contrast,
  the candidate (\bsq{milk}, $[4,6)$, $\mathtt{null}$, $a_1$) is
  rejected since (\bsq{milk}) is not included in a tuple of the left
  input relation $\mathbf{c}$ over $[4,6)$.
\end{example}

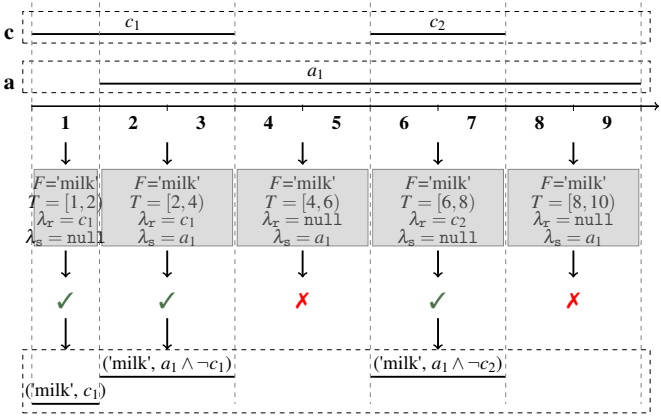
\begin{figure}[!htbp]
\scalebox{0.6} {
\begin{tikzpicture}   
   \draw (0,0) [->, line width = 1pt]-- coordinate (x axis mid) (13.9,0);
  
\pgfmathsetmacro{\shift}{0.5}
\foreach \j in {1,...,10}{
     \pgfmathsetmacro{\divRes}{int(\j/2)}
     \pgfmathsetmacro{\modRes}{1-(\j-\divRes*2)}
     \pgfmathsetmacro{\xPos}{\j-1+(\divRes-\modRes*\shift)}
     \draw (\xPos , 1pt)  --  (\xPos ,-3pt) node[anchor=north,font=\relsize{2}] {};
}

\pgfmathsetmacro{\shift}{0.5}
\foreach \j in {1,...,9}{
     \pgfmathsetmacro{\divRes}{int(\j/2)}
     \pgfmathsetmacro{\modRes}{1-(\j-\divRes*2)}
     \pgfmathsetmacro{\xPos}{\j-1+(\divRes-\modRes*\shift)}
     \draw (\xPos+0.75 , -3pt) node[anchor=north,font=\relsize{1}] {\bf \j};
}

 \pgfmathsetmacro{\xAxisOne}{-0.2}
 \pgfmathsetmacro{\xAxisTwo}{13.7}
 \pgfmathsetmacro{\windowHeight}{0.9}
 
    \pgfmathsetmacro{\yAxis}{0.5}
 \draw [line width=1.2] (1.5, \yAxis) -- (13.5,\yAxis) node[pos=.4,above=-1pt]{\large $a_1$};
  \node at (\xAxisOne-0.3, \yAxis) {\Large $\mathbf{a}$}; 
    \pgfmathsetmacro{\yAxis}{\yAxis-0.2}
 \draw [dashed]
(\xAxisOne,\yAxis) -- (\xAxisOne,\yAxis + \windowHeight - 0.2) -- (\xAxisTwo,\yAxis + \windowHeight - 0.2) -- (\xAxisTwo, \yAxis + 0)  -- cycle;

   \pgfmathsetmacro{\yAxis}{1.6}
 \draw [line width=1.2] (0,\yAxis) -- (4.5, \yAxis) node[pos=0.5,above=-1pt]{\large $c_1$};
 \draw [line width=1.2] (7.5, \yAxis) -- (10.5,\yAxis) node[pos=.5,above=-1pt]{\large $c_2$}; 
  \node at (\xAxisOne-0.3, \yAxis) {\Large $\mathbf{c}$}; 
  \pgfmathsetmacro{\yAxis}{\yAxis-0.2}
 \draw [dashed]
(\xAxisOne,\yAxis) -- (\xAxisOne,\yAxis + \windowHeight - 0.2) -- (\xAxisTwo,\yAxis + \windowHeight - 0.2) -- (\xAxisTwo, \yAxis + 0)  -- cycle;
 
    \pgfmathsetmacro{\yPosA}{-0.8}
    \pgfmathsetmacro{\yPosB}{\yPosA-0.5}
    \draw [line width=1.2, ->]   (0.75, \yPosA) --   (0.75,\yPosB) node[pos=.5,above=-1pt]{}; 
	\draw [line width=1.2, ->]      (3, \yPosA) --      (3,\yPosB) node[pos=.5,above=-1pt]{}; 
	\draw [line width=1.2, ->]      (6, \yPosA) --      (6,\yPosB) node[pos=.5,above=-1pt]{}; 
	\draw [line width=1.2, ->]      (9, \yPosA) --      (9,\yPosB) node[pos=.5,above=-1pt]{}; 
	\draw [line width=1.2, ->]     (12, \yPosA) --     (12,\yPosB) node[pos=.5,above=-1pt]{}; 
    
    \pgfmathsetmacro{\yPos}{\yPosB-0.4}
    \node at (0.75, \yPos) {\large $F$=\bsq{milk}}; 
    \node at (0.75, \yPos-0.45) {\large $T = [1,2)$}; 
    \node at (0.75, \yPos-0.8) {\large $\mathtt{\lambda_r} = c_1 $}; 
    \node at (0.75, \yPos-1.2) {\large $\mathtt{\lambda_s} = \mathtt{null} $}; 
    \node at (3, \yPos) {\large $F$=\bsq{milk}}; 
    \node at (3, \yPos-0.45)  {\large $T = [2,4)$}; 
    \node at (3, \yPos-0.8) {\large $\mathtt{\lambda_r} = c_1$}; 
    \node at (3, \yPos-1.2) {\large $\mathtt{\lambda_s} = a_1 $}; 
    \node at (6, \yPos) {\large $F$=\bsq{milk}}; 
    \node at (6, \yPos-0.45) {\large $T = [4,6)$}; 
    \node at (6, \yPos-0.8) {\large $\mathtt{\lambda_r} = \mathtt{null}$}; 
    \node at (6, \yPos-1.2) {\large $\mathtt{\lambda_s} = a_1 $}; 
    \node at (9, \yPos) {\large $F$=\bsq{milk}};    
    \node at (9, \yPos-0.45) {\large $T = [6,8)$}; 
    \node at (9, \yPos-0.8)  {\large $\mathtt{\lambda_r} = c_2 $}; 
    \node at (9, \yPos-1.2) {\large $\mathtt{\lambda_s} = \mathtt{null} $};     
    \node at (12, \yPos) {\large $F$=\bsq{milk}}; 
    \node at (12, \yPos-0.45)  {\large $T = [8,10)$}; 
     \node at (12, \yPos-0.8) {\large $\mathtt{\lambda_r} = \mathtt{null}$}; 
     \node at (12, \yPos-1.2) {\large $\mathtt{\lambda_s} = a_1 $};     
     
    \pgfmathsetmacro{\yPosA}{\yPos-1.2-0.3}
    \pgfmathsetmacro{\yPosB}{\yPosA-0.6}
    \draw [line width=1.2, ->]   (0.75, \yPosA) --   (0.75,\yPosB) node[pos=.5,above=-1pt]{}; 
	\draw [line width=1.2, ->]      (3, \yPosA) --      (3,\yPosB) node[pos=.5,above=-1pt]{}; 
	\draw [line width=1.2, ->]      (6, \yPosA) --      (6,\yPosB) node[pos=.5,above=-1pt]{}; 
	\draw [line width=1.2, ->]      (9, \yPosA) --      (9,\yPosB) node[pos=.5,above=-1pt]{}; 
	\draw [line width=1.2, ->]     (12, \yPosA) --     (12,\yPosB) node[pos=.5,above=-1pt]{};

  	\draw [fill=grey, opacity = 0.4] (0+0.05,\yPos+0.3) -- (0+0.05,\yPosA+0.1) --
  	(1.5-0.05,\yPosA+0.1) -- (1.5-0.05, \yPos + 0.3)  -- cycle;
  	\draw [fill=grey, opacity = 0.4] (1.5+0.05,\yPos+0.3) -- (1.5+0.05,\yPosA+0.1) --
  	(4.5-0.05,\yPosA+0.1) -- (4.5-0.05, \yPos + 0.3)  -- cycle;
  	\draw [fill=grey, opacity = 0.4] (4.5+0.05,\yPos+0.3) -- (4.5+0.05,\yPosA+0.1) --
  	(7.5-0.05,\yPosA+0.1) -- (7.5-0.05, \yPos + 0.3)  -- cycle;
  	\draw [fill=grey, opacity = 0.4] (7.5+0.05,\yPos+0.3) -- (7.5+0.05,\yPosA+0.1) --
  	(10.5-0.05,\yPosA+0.1) -- (10.5-0.05, \yPos + 0.3)  -- cycle;
  	\draw [fill=grey, opacity = 0.4] (10.5+0.05,\yPos+0.3) -- (10.5+0.05,\yPosA+0.1) --
  	(13.5-0.05,\yPosA+0.1) -- (13.5-0.05, \yPos + 0.3)  -- cycle;

    \pgfmathsetmacro{\yPos}{\yPosB-0.5}
    \node at (0.75, \yPos) [font=\relsize{2}] {\color{c1}{\cmark}}; 
    \node at (3, \yPos) [font=\relsize{2}]  {\color{c1}{\cmark}}; 
    \node at (6, \yPos) [font=\relsize{2}] {\color{red}{\xmark}}; 
     \node at (9, \yPos) [font=\relsize{2}] {\color{c1}{\cmark}}; 
    \node at (12, \yPos) [font=\relsize{2}] {\color{red}{\xmark}}; 
    
    \pgfmathsetmacro{\yPosA}{\yPos-0.4}
    \pgfmathsetmacro{\yPosB}{\yPosA-0.7}
    \draw [line width=1.2, ->]   (0.75, \yPosA) --   (0.75,\yPosB) node[pos=.5,above=-1pt]{}; 
	\draw [line width=1.2, ->]      (3, \yPosA) --      (3,\yPosB) node[pos=.5,above=-1pt]{}; 
	\draw [line width=1.2, ->]      (9, \yPosA) --      (9,\yPosB) node[pos=.5,above=-1pt]{}; 

    \pgfmathsetmacro{\yPos}{\yPosB-0.6}
    \draw [line width=1.2]       (0,\yPos-0.6) -- (1.5,\yPos-0.6) node[pos=.5,above=-1pt]{\large (\bsq{milk}, $c_1$)}; 
    \draw [line width=1.2] (1.5, \yPos) -- (4.5,\yPos) node[pos=.5,above=-1pt]{\large (\bsq{milk}, $a_1 \land \neg c_1$)}; 
	\draw [line width=1.2] (7.5, \yPos) -- (10.5,\yPos) node[pos=.5,above=-1pt]{\large (\bsq{milk}, $a_1 \land \neg c_2$)}; 

   \pgfmathsetmacro{\yAxis}{2.3}
   \pgfmathsetmacro{\yPos}{\yPos-0.6}
   \draw[color=gray,dashed] (0,\yAxis) -- (0,\yPos);
   \draw[color=gray,dashed] (1.5,\yAxis)   -- (1.5,\yPos);
   \draw[color=gray,dashed] (4.5,\yAxis)   -- (4.5,\yPos);
   \draw[color=gray,dashed] (7.5,\yAxis)   -- (7.5,\yPos);
   \draw[color=gray,dashed] (10.5,\yAxis) -- (10.5,\yPos);
    \draw[color=gray,dashed] (13.5,\yAxis) -- (13.5,\yPos);

       \pgfmathsetmacro{\yAxis}{\yPos-0.2}
        \pgfmathsetmacro{\windowHeight}{1.6}
 \draw [dashed]
(\xAxisOne,\yAxis) -- (\xAxisOne,\yAxis + \windowHeight - 0.2) -- (\xAxisTwo,\yAxis + \windowHeight - 0.2) -- (\xAxisTwo, \yAxis + 0)  -- cycle;

   
\end{tikzpicture}}
\caption{$\sigma_\kat{ F = \bsq{milk}}(\mathbf{c}) -^\kat{TP}
  \sigma_\kat{ F = \bsq{milk}}(\mathbf{a})$}
\label{fig:processSetDif}
\end{figure}

\smallskip
\noindent{\bf Time and Space Complexity:} The time complexity of all TP
set operations is determined by the complexity of the blocks presented
in Fig.~\ref{fig:algorithmSketch}. Sorting has complexity
$O(|\mathbf{r}| \log | \mathbf{r}| + |\mathbf{s}| \log |\mathbf{s}|)$
if it is comparison-based. A variant of counting-based sorting could
also be used~\cite{KaufmannTI} (which is the case if $\Omega^T$ fits
into main-memory), and in this case the corresponding complexity is even
linear. After sorting, LAWA will sweep over all tuples in the
sorted input relations $\mathbf{r}$ and $\mathbf{s}$, accessing two
input tuples at a time to determine the next window.

\begin{proposition}
  \label{prop:windowNumber} {\em Let $\mathbf{r}$, $\mathbf{s}$ be two
    duplicate-free temporal-probabilistic relations.  The upper bound
    of the number of windows produced by the window advancer is
    $n_r + n_s - f_d$ where $n_r$, $n_s$ are the number of start and
    end points in $\mathbf{r}$ and $\mathbf{s}$, and $f_d$ is number
    of distinct facts in these relations.}
\end{proposition}

By Proposition~\ref{prop:windowNumber}, the number of candidate
windows considered by the algorithm is linear in the number of time
intervals, and thus to the size of the input relations. Thus, LAWA has
a time complexity of $O(|\mathbf{r}| + |\mathbf{s}|)$, given that $|\mathbf{r}|$ and
$|\mathbf{s}|$ are the numbers of tuples in the input relations
$\mathbf{r}$ and $\mathbf{s}$, respectively.
Moreover, the filtering and lineage-concatenation step for each candidate output
tuple is performed in $O(1)$.  Thus, the overall time complexity for computing
TP set operations is
$O(|\mathbf{r}| \log | \mathbf{r}| + |\mathbf{s}| \log |\mathbf{s}|)$,
but may even be reduced to $O(|\mathbf{r}| + |\mathbf{s}|)$ if counting-based sorting is applicable.
The use of {\em lineage-aware temporal windows} not only avoids the
use for time-consuming additional operations for the filtering and
lineage-concatenation steps, but also allows them to be performed directly at
the time a window is created. That is, no intermediate buffers need to
be maintained (apart from very few pointers), and thus the space
complexity of all TP set operators is constant.

\section{Experimental Evaluation}
\label{sec:exper-eval}
In this section, we evaluate LAWA in comparison to both temporal and
temporal-probabilistic approaches that can be used for the computation
of TP set operations.  We perform experiments with real datasets as
well as with synthetic datasets in which we vary (i) the number of
facts in the input relations and (ii) the percentage of tuples whose
intervals overlap.  In all experiments, our approach empirically
scales according to the bounds we provide in
Section~\ref{sec:basicIUE}.  LAWA is the only scalable approach that
can be used for the computation of all three TP set operations,
outperforming all state-of-the-art approaches for input relations of
more than 10M tuples.  In contrast to existing techniques, LAWA is
robust, i.e., its performance behaves in a predictable manner with
respect to the aforementioned characteristics of the datasets.

\begin{figure*}[ht]\centering
  \begin{subfigure}[b]{0.3\linewidth}\centering
    \scalebox{0.5}{
      \begin{tikzpicture}
        \begin{axis}[scale only axis,
                     height=5cm,
                     width=1.7\linewidth,
                     xlabel=Number of Input Tuples {[K]},
                     ylabel=Runtime {[ms]},
                     label style={font=\large},
                     legend style={font=\large},
                     ymax=1*10^6,
                     tick label style={font=\large},
                     mark size=4pt,
                     xmin = 20,
                     xmax = 200]
          \addplot[color=red,mark=x,mark size=6pt]
            file[skip first]{graphs/synthetic_intersection_small_wal.tsv};
          \addplot[color=blue,mark=+,mark size=6pt]
            file [skip first]{graphs/synthetic_intersection_small_oip.tsv};
          \addplot[color=green,mark=*]
            file [skip first]{graphs/synthetic_intersection_small_ti.tsv};
          \addplot[color=cyan,mark=triangle*]
            file[skip first]{graphs/synthetic_intersection_small_tpdb.tsv};
          \addplot[color=black,mark=o]
            file [skip first]{graphs/synthetic_intersection_small_norm.tsv};
          \legend{LAWA,OIP, TI,TPDB,NORM}
        \end{axis}
      \end{tikzpicture}
    }
    \caption{Set Intersection}
    \label{fig:syntheticIntersectionSmall}
  \end{subfigure}
  \qquad
  \begin{subfigure}[b]{0.28\linewidth} \centering
    \scalebox{0.5}{
      \begin{tikzpicture}
        \begin{axis}[scale only axis,
	             height=5cm,
	             legend style={font=\large},
	             width=1.7\linewidth,
                     mark size=4pt,
                     label style={font=\Large},
                     xlabel=Number of Input Tuples {[K]},
                     ylabel=Runtime {[ms]},
                     tick label style={font=\large} ,
		     ymax=3.25*10^6,
		     xmin = 20,
		     xmax = 200]
          \addplot[color=red,mark=x,mark size=6pt]
            file [skip first] {graphs/synthetic_difference_small_wal.tsv};
          \addplot[color=black,mark=o]
            file [skip first]{graphs/synthetic_difference_small_norm.tsv};			
          \legend{LAWA,NORM}
	\end{axis}
      \end{tikzpicture}
    }
    \caption{Set Difference}
    \label{fig:syntheticDifferenceSmall}
  \end{subfigure}
  \qquad
  \begin{subfigure}[b]{0.28\linewidth}\centering
    \scalebox{0.5}{
      \begin{tikzpicture}
        \begin{axis}[scale only axis,
	             height=5cm,
	             width=1.7\linewidth,
                     mark size=4pt,
	             label style={font=\Large},
	             legend style={font=\large},
		     xlabel=Number of Input Tuples {[K]},
		     ylabel=Runtime {[ms]},
                     tick label style={font=\large} ,
		     ymax=27*10^4,
		     xmin = 20,
		     xmax = 200]		
          \addplot[color=red,mark=x,mark size=6pt]
            file[skip first]{graphs/synthetic_union_small_wal.tsv};
          \addplot[color=cyan,mark=triangle*]
            file[skip first]{graphs/synthetic_union_small_tpdb.tsv};
          \addplot[color=black,mark=o]
            file[skip first]{graphs/synthetic_union_small_norm.tsv};
	  \legend{LAWA,TPDB,NORM}
	\end{axis}
      \end{tikzpicture}
    }
    \caption{Set Union}
  \label{fig:syntheticUnionSmall}
  \end{subfigure}
  \caption{Synthetic Dataset [20K--200K]}
  \vspace*{-0.3cm}
  \label{fig:syntheticAllSmall}
\end{figure*}
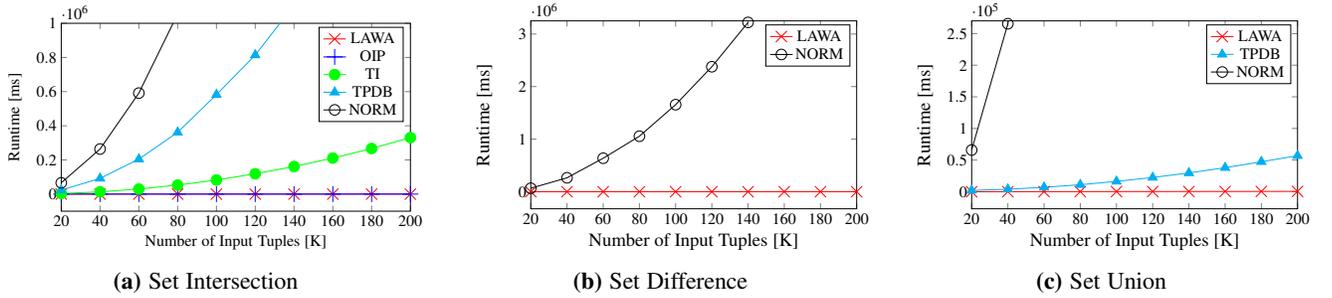

\subsection{Experimental Setup}

All of the following experiments were deployed on a 2xIntel(R) Xeon(R)
CPU E5-24400 @2.40GHz machine with 64GB main memory, running CentOS
6.7. LAWA has been implemented in C++ \footnote{http://www.ifi.uzh.ch/en/dbtg/research/tpset.html}, and all experiments were
performed in main-memory.  No indexes were used.  In cases where
PostgreSQL implementations were used, the maximum memory for sorting
as well as for shared buffers was set to 1GB.

\begin{table}[!htpb]
\caption{Approach Overview}
\label{tab:approachOverview}
\centering
\scalebox{0.9}{
\begin{tabular}{ | M{2cm} !{\VRule[1.5pt]} M{1cm} | M{1cm} | M{1cm}  | @{}m{0cm}@{}}
\hline
\textbf{Approach} &
$\textbf{r} \cup^\kat{Tp} \textbf{s}$ &
$\textbf{r} -^\kat{Tp} \textbf{s}$ &
$\textbf{r} \cap^\kat{Tp} \textbf{s}$ \\ [0.1cm] \specialrule{1.2pt}{0pt}{0pt}
LAWA  & \color{c1}{\cmark} & \color{c1}{\cmark} &
	    \color{c1}{\cmark} & \\ [0.12cm] \hline 
NORM  & \color{c1}{\cmark} & \color{c1}{\cmark} &
		\color{c1}{\cmark} & \\ [0.12cm] \hline 
TPDB  & \color{c1}{\cmark} & \color{c3}{\xmark} &
		\color{c1}{\cmark} & \\ [0.12cm] \hline 
OIP	  & \color{c3}{\xmark} & \color{c3}{\xmark} &
		\color{c1}{\cmark} & \\ [0.12cm] \hline 
TI	  & \color{c3}{\xmark} & \color{c3}{\xmark} &
		\color{c1}{\cmark} & \\ [0.12cm] \hline 
\end{tabular}}
\end{table}

The TP set operations that different approaches can compute are
presented in Table~\ref{tab:approachOverview}. Set difference is the
least-supported operation, followed by set union and set
intersection. Set intersection is the most-sup\-ported operation among
the available systems, since it can be reduced to an interval join
with an equality condition on the non-temporal
attributes. Specifically, we compare our implementation of TP set
operations using LAWA against:

\smallskip
\noindent{\bf Temporal-Probabilistic Database
  (TPDB)~\cite{DyllaMT13}:}
The implementation of TPDB is an application connected with a DBMS and
consists of three stages.  The first stage parses Datalog rules with
temporal predicates and translates them to SQL queries. The second
stage executes the SQL queries in the DBMS. Base relations are stored
in the DBMS, while lineage is kept as an internal data structure in
main-memory. The third stage focuses on lineage processing by
processing the base tuples with their Boolean connectives. We use the
authors' original implementation, connected to PostgreSQL 9.4.3.

\smallskip
\noindent \textbf{Normalize (NORM)~\cite{DignosTODS16}}: The
\emph{Normalize} operator is implemented in the kernel of PostgreSQL
by modifying its parser, executor and optimizer. We migrated the
authors' implementation to PostgreSQL 9.4.3 for a fair comparison. 
To support TP set operations, we introduced reduction rules that are
proper combinations of the temporal and probabilistic reduction 
rules (cf.\ \cite{DignosTODS16,FinkO11}) and we illustrate them in
Fig.~\ref{fig:tpNORM}.

\begin{figure}[!h]
\tikzset{%
  algebra/.style    = {draw, thick, rectangle, minimum height = 2em,
    minimum width = 2em},
}
\tikzset{%
  lineage/.style    = {draw, thick, rectangle, minimum height = 2em,
    minimum width = 2.3cm},
}
\tikzset{%
  filter/.style    = {draw, thick, rectangle, minimum height = 2em,
    minimum width = 2.75cm},
}
\tikzset{%
  block/.style    = {draw, thick, rectangle, minimum height = 2em,
    minimum width = 2em},
}
\centering
\scalebox{0.6}{
\begin{tikzpicture}[auto, thick, node distance=2cm, >=triangle 45]
\draw
	node at (0,0)[right=-3mm] (relR) {\large $\mathbf{r}$}
	node [right= 0.5cm of relR] (relR_line1) {}

	node [above right = 0.3cm and 0.8cm of relR_line1, block] (relNR)
		 {\large {\bf N} ($\mathbf{r}$, $\mathbf{s}$)}
 	node [above = 0.05cm of relNR.east] (relNR1) {}
	node [below = 0.05cm of relNR.east] (relNR2) {}
 	node [above = 0.05cm of relNR.west] (relNR3) {}
	node [below = 0.05cm of relNR.west] (relNR4) {}
	
	node [below = 0.2cm of relR] (relS) {\large $\mathbf{s}$}
	node [right= 0.25cm of relS] (relS_line1) {}

	node [below right = 0.3cm and 1.05cm of relS_line1, block] (relNS) 
		 {\large {\bf N} ($\mathbf{s}$, $\mathbf{r}$)}
 	node [above = 0.05cm of relNS.east] (relNS1) {}
	node [below = 0.05cm of relNS.east] (relNS2) {}
 	node [above = 0.05cm of relNS.west] (relNS3) {}
	node [below = 0.05cm of relNS.west] (relNS4) {}
	
	node [above right = 0.3cm and 3cm of relNR.north, algebra] (setD)
	{\large \color{c2}{$\LJoin$}}
	node [above = 0.05cm of setD.west] (setD1) {}
	node [below left = 0.05cm and 1.2cm of setD.west] (setD11) {}
	node [below = 0.05cm of setD.west] (setD2) {}
	
	node [below = 1.3cm of setD, algebra] (setI)
	{\large \color{c1}{$\bowtie$}}
	node [above left = 0.05cm and 0.5cm of setI.west] (setI11) {}
	node [below left = 0.05cm and 0.5cm of setI.west] (setI22) {}
	node [above = 0.05cm of setI.west] (setI1) {}
	node [below = 0.05cm of setI.west] (setI2) {}

	node [below = 1.3cm of setI, algebra] (setU) {\large \color{c3}{$\cup$}}
	node [above = 0.05cm of setU.west] (setU1) {}
	node [left = 0.75cm of setU1] (setU11) {}
	node [below = 0.05cm of setU.west] (setU2) {}
	
	node [right = 0.6cm of setI, lineage] (setIL) {\large \color{c1}{and($\lambda_r$, $\lambda_s$)}}
	node [right = 0.6cm of setD, lineage] (setDL) {\large \color{c2}{andNot($\lambda_r$, $\lambda_s$)}}
	node [right = 0.6cm of setU, lineage,white] (setUL) {}
	
	node [right = 0.1cm of setUL, lineage] (dupU) {\large $\vartheta$\color{c3}{or($\lambda$)}}

	node [right =0.8cm of dupU] (resU) {\large \color{c3}{\large $\mathbf{r} \cup^\kat{Tp} \mathbf{s}$}}
	node [right =3.2cm of setIL] (resI) {\large \color{c1}{\large $\mathbf{r}\cap^\kat{Tp} \mathbf{s}$}}
	node [right =2.7cm of setDL] (resD) {\large \color{c2}{\large $\mathbf{r} -^\kat{Tp} \mathbf{s}$}};

\draw[->] (relR.north) |- (relNR3.center);
\draw[-] (relR.east) -| (relS_line1.center);
\draw[->] (relS_line1.center) |- (relNS3.center);

\draw[->] (relS.south) |- (relNS4.center);
\draw[-] (relS.east) -| (relR_line1.center);
\draw[->] (relR_line1.center) |- (relNR4.center);

\draw[->] (relNR.north) |- (setD1.center);
\draw[-] (relNS1.center) -| (setD11.center);
\draw[->] (setD11.center) |- (setD2.center);

\draw[->] (relNS.south) |- (setU2.center);

\draw[-] (relNS2.center) -| (setI22.south);
\draw[->] (setI22.south) |- (setI2.center);
\draw[-] (relNR1.center) -| (setI11.north);
\draw[->] (setI11.north) |- (setI1.center);

\draw[-] (relNR2.center) -| (setU11.center);
\draw[->] (setU11.center) |- (setU1.center);

\draw[->] (setI)  -- (setIL);
\draw[->] (setD) -- (setDL);
\draw[->] (setU) -- (dupU);

\draw[->] (setIL) -- (resI);
\draw[->] (setDL) -- (resD);
\draw[->] (dupU) -- (resU);

\end{tikzpicture}}
\vspace{0.1cm}
\caption{TP set operations using \emph{NORM}. The approach adopted is a
combination of the processes described in Fig.~\ref{fig:temporalDBs} and Fig.~\ref{fig:probabilisticDB}}
\label{fig:tpNORM}
\end{figure}
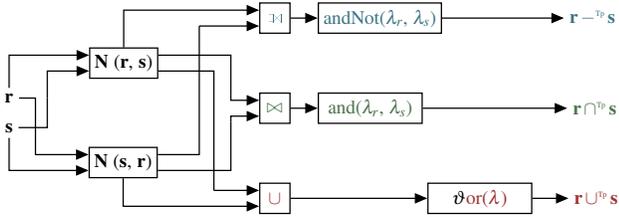

\smallskip
\noindent \textbf{Timeline Index (TI)~\cite{KaufmannTI}:}
This approach was used, in its original implementation, for the
computation of TP set intersection, by applying a temporal join with
an additional condition on the non-temporal attributes as well as
the lineage-concatenating function {\bf\emph{and}} (see Table.~\ref{tab:lineageFunctions}).

\smallskip
\noindent\textbf{Overlap Interval Partition Join
  (OIP)~\cite{DignosBG14}}:
This approach is designed for overlap joins but does not support an
additional filtering condition. For our experimental evaluation, we
extended the authors' implementation, so that an equality condition on
the non-temporal attributes of the tuples can be applied.  In order to
use OIP to compute set intersection, we first split each input
relation into groups based on the facts included in each tuple.  We
then applied the OIP partitioning and join over each of these groups
and merged the results.

\subsection{Synthetic Dataset}

The parameters that we consider to populate a relation of our dataset
are: (a) the length of the tuples' intervals, (b) the maximum time
distance between two tuples that are consecutive and include the same
fact, and (c) the number of different facts included in tuples of the
relation. Assume all tuples of relations $\mathbf{r}$ and $\mathbf{s}$
have the same fact $f$. We define the \emph{overlapping factor} of $f$
as the number of maximal subintervals during which a tuple from
$\mathbf{r}$ and $\mathbf{s}$ overlap, divided by the total number of
maximal subintervals.  Its value thus ranges in $[0,1]$. The higher
the value of this metric, the more pairs of input tuples form output
tuples, and therefore the more we stress-test the performance of the
various approaches for TP set operations. According to
Definition~\ref{def:tpSetOpDef}, overlapping time points are relevant
for all set operations, whereas time points for which a fact is only
included in the left input relation are only relevant for TP set
difference.

\smallskip
\noindent\textbf{1. Runtime.}
In the first setting, we fix the input tuples of all datasets to a
single fact. We fix the overlapping factor to 0.6, and we randomly
select the length of the intervals and the distance between two
consecutive intervals in $[0, 3]$. We then systematically increase the
number of input tuples. In Fig.~\ref{fig:syntheticAllSmall} and
Fig.~\ref{fig:syntheticAllLarge}, we illustrate the performance of all
the approaches for the computation of TP set operations for smaller
datasets with up to 100K tuples and for larger datasets with up to 50M
tuples, respectively.

\smallskip
\noindent\textbf{Smaller Datasets [20K--200K]:}
In Fig.~\ref{fig:syntheticAllSmall}, the datasets range from 20K to
200K tuples. Fig.~\ref{fig:syntheticIntersectionSmall} focuses on TP
set intersection. The runtimes of LAWA and OIP hardly increase
for the small datasets.  Both outperform NORM, TI and TPDB by a large
margin. OIP is specifically designed for the computation of an overlap
join, to which TP set intersection is reduced. NORM exhibits poor
performance even if the number of input tuples is only 50K. In this
approach, regardless of the operation, the two input relations need to
first be normalized, such that, in their adjusted versions, the
intervals would be either equal or disjoint.  The most expensive part
of the normalization of a relation $\mathbf{r}$ using relation
$\mathbf{s}$ is an outer join that uses inequality conditions on the
start and end points to guarantee an overlap of the
intervals. Although an additional inner join is applied in the case of
TP set intersection, the performance of NORM suffers because of the
outer join.  Since all tuples include the same fact, but not all of
them overlap, such a join has quadratic complexity~\cite{Khayyat2015}.

In TPDB, queries are expressed using Datalog. Each rule may contain a
conjunction of literals over the arithmetic predicates $=^T$, $\neq^T$
and $\leq^T$. In order to express TP set intersection, we use 6
reduction rules, one for each overlap relationship defined by
Allen~\cite{Allen1983}. TPDB then translates each rule to an inner
join that is submitted to PostgreSQL. Although there is an equality
condition on the non-temporal attributes, it is not used in the cases
examined in Fig.~\ref{fig:syntheticAllSmall} where all the tuples
include the same fact. Thus, the joins are only based on the
inequality conditions and perform a larger number of comparisons. TPDB
is slower than the other approaches, but it is still faster than NORM,
because the latter has to adjust each relation.

Although TI is faster than NORM and TPDB, it is one of the slowest
approaches for set intersection. The index allows for the avoidance of
redundant comparisons related to the interval overlap condition, and
its creation cost is a small percentage of its runtime. Given the
indexes of the input relations, TI performs a merge-join on them and
produces ($r_{id}$, $s_{id}$) pairs. In order to form the output
tuples, the input tuples corresponding to each pair need to be
retrieved. Given the value of the overlapping factor and the existence
of only one fact, a higher number of joined pairs is produced and thus
a higher number of lookups is required. OIP splits the tuples of each
input relation into partitions, based on the start/end points of their
interval and its duration. Consequently, it offers a mechanism that
performs interval comparisons between tuples only if their partitions
overlap. If the partitions overlap, OIP performs a nested loop between
the tuples of the two relations. As the overlapping factor is 0.6,
which indicates that most of the pairs produced in the nested loop
will indeed be output pairs, OIP has a very small percentage of false
hits. Although OIP is tailored for an overlap join, for datasets of up
to $200K$ tuples LAWA's performance is competitive, being on average
30 ms slower.

In the case of TP set difference, as illustrated in
Fig.~\ref{fig:syntheticDifferenceSmall}, LAWA clearly outperforms
NORM, for the same reasons as for TP set intersection.
Fig.~\ref{fig:syntheticUnionSmall} compares LAWA with NORM and TPDB
during the computation of TP set union. LAWA has the lowest runtime,
whereas NORM has the highest one, being 5 orders of magnitude slower
than LAWA. The window that sweeps over all the input tuples in LAWA
makes no false hits in this case, since all of the subintervals that
the window defines correspond to output intervals. NORM no longer
requires a join but a union after the relations have been
normalized. However, as in all the previous operations, NORM's
performance is hindered by the computation of the timestamp
adjustment. TPDB can also compute TP set union by using a deduction
rule that corresponds to a conventional union instead of joins, and
thus its performance is significantly better in comparison to TP set
intersection.

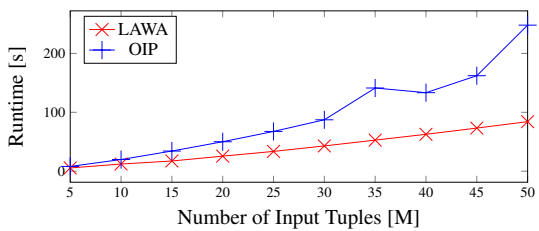
\begin{figure}[ht] \centering
  \scalebox{0.57}{
    \begin{tikzpicture}
      \begin{axis}[label style={font=\large},
                   scale only axis,
                   height=4cm,
                   mark size=4pt,
                   label style={font=\Large},
                   width=1.2\linewidth,
                   xlabel=Number of Input Tuples {[M]},
                   ylabel=Runtime {[s]},
		   legend style={font=\large},
                   legend style={legend pos=north west},
		   xmin = 5,
		   xmax = 50]
        \addplot[color=red,mark=x,mark size=6pt]
          file[skip first]{graphs/synthetic_intersection_large_wal.tsv};
        \addplot[color=blue,mark=+,mark size=6pt]
          file[skip first]{graphs/synthetic_intersection_large_oip.tsv};
	\legend{LAWA,OIP}
      \end{axis}
    \end{tikzpicture}
  }
  \caption{Synthetic Dataset [5M--50M]}
  \label{fig:syntheticAllLarge}
   \vspace{-0.4cm}
\end{figure}

\smallskip
\noindent\textbf{Larger Datasets [5M--50M]:} LAWA is the only
scalable approach that can be used for the computation of all three TP
set operations. In Fig.~\ref{fig:syntheticAllLarge}, we depict the
performance of LAWA for the computation of TP set intersection for
larger datasets. The overlapping factor of the datasets remains fixed
to 0.6, and the dataset sizes vary from 5M to 50M tuples. While OIP is
also considered, the other approaches that were included in
Fig.~\ref{fig:syntheticIntersectionSmall} are not taken into
consideration, since their runtimes were already two to five orders of
magnitude higher when applied on the smaller datasets.  After 30M
tuples, LAWA is at least 2 times faster than OIP and continues to
scale better.  OIP produced a small number of partitions that contain
many tuples each.  Such partitions are likely to overlap and the
nested loop that matches their tuples is computationally expensive.
As far as TP set difference and TP set union are concerned, LAWA has
similar runtime as in the case of TP set intersection and it is the
only scalable approach suitable for their computation within at most
100 seconds.

\smallskip
\noindent\textbf{2. Robustness.}
In this experiment, we show that LAWA is a scalable operator whose
runtime only depends on the size of the dataset and not on its other
characteristics (i.e., neither on the value of the overlapping factor
nor on the number of distinct facts captured by the input tuples).

\begin{table}[htb]
\caption{Dataset Characteristics}
\label{tab:datasets}
\center
\scalebox{0.9}{
\begin{tabular}{ | M{3cm} !{\VRule[1.5pt]} M{0.4cm} | M{0.3cm} |
M{0.3cm}| M{0.3cm}| M{0.3cm} | @{}m{0pt}@{}}

\hline      
      Overlapping Factor & 0.03 & 0.1 & 0.4 & 0.6 & 0.8 & \\ [0.1cm]
      \hline
      Max. Interval Length (R) & 100  & 100 & 50 & 3 & 10 & \\ [0.1cm]
      \hline
      Max. Interval Length (S)  & 3  & 10  & 10  & 3  & 10 & \\ [0.1cm]
      \hline
      Max. Time Distance & \multicolumn{5}{c|}{3} & \\ [0.1cm]
      \hline
\end{tabular}}
\end{table}

In Fig.~\ref{fig:metricInd}, the performance of LAWA for set 
intersection is compared with the one of OIP, which has been
the most competitive approach for datasets where all the tuples
include the same fact. This time, the size of the dataset is fixed
to 30M, and the overlapping factor is assigned to four different
values in $[0,1]$. Table~\ref{tab:datasets} depicts the
\emph{overlapping factor} of the datasets as well as their
\emph{maximum interval lengths} (in terms of the number of time points).
The runtime of OIP increases as the overlapping metric
increases. The reason is that the higher the overlapping factor,
the more tuples occur in a partition and the nested loop performed in
each partition is very time consuming. On the other hand, only minor
variations are observed in the runtime of LAWA for the different
values of the overlapping factor, thus demonstrating that the performance
of LAWA is not negatively affected by interval-related
characteristics of the dataset.

\begin{figure}[!htbp]
\center
\begin{subfigure}[b]{\linewidth}
  \centering
    \scalebox{0.57}{
		\begin{tikzpicture}
		\begin{axis}[
		legend columns=-1, 
		label style={font=\Large},
       tick label style={font=\large} ,
				legend style={nodes={scale=0.7, transform shape, font=\large},legend pos=north west },
		xlabel=Overlapping Factor, 
		ylabel=Runtime {[s]},
		scale only axis, 
	height=4cm,
	width=1.2\linewidth,
		bar width = 10pt,
	    xtick={0.03,0.1,0.4,0.6,0.8},
       symbolic x coords={0.03,0.1,0.4,0.6,0.8},
       ybar=2*\pgflinewidth,
       ymax=300,
       enlargelimits=0.2
]
			\addplot[color=red, fill=red, area legend] coordinates  {(0.4, 205.60455) (0.6, 191.86334) (0.8, 213.6621) (0.03,	187.78823) (0.1, 192.92347)}; 
			\addplot[color=blue, fill=blue, area legend] coordinates  {(0.8, 300.725) (0.1, 110.581984) (0.6, 284.39178) (0.03, 86.33033) (0.4, 152.53466)}; 
		\legend{LAWA,OIP}
		\end{axis}
    	\end{tikzpicture}} 
    	\caption{Performance for varying overlapping factors.}
  \label{fig:metricInd}
\end{subfigure}


\begin{subfigure}[b]{\linewidth}
  \centering
    \scalebox{0.57}{
		\begin{tikzpicture}
		\begin{axis}[
		legend columns=-1, 
		label style={font=\Large},
       tick label style={font=\large} ,
		legend style={nodes={scale=0.7, transform shape, font=\large},at={(0.85,0.98)} },
		xlabel=Approaches, 
		ylabel=Runtime {[ms]},
		bar width = 6pt,
		scale only axis, 
	height=4cm,
	width=1.2\linewidth,
       xtick={NORM, LAWA,OIP, TI,TPDB},
       symbolic x coords={NORM, LAWA,OIP, TI,TPDB},
       ybar=2*\pgflinewidth,
       ymax=1*10^6,
       ymode = log,
       enlargelimits=0.2
]
			\addplot +[area legend] coordinates {(NORM, 1330.2607) (LAWA, 281.7925) (OIP, 2377.918) (TI, 25.37) (TPDB, 5057.6665)};
			\addplot +[area legend] coordinates {(NORM, 1611969.8) (LAWA, 293.8) (OIP, 3.224) (TI, 1312.4669) (TPDB, 5057.6665)};
						\addplot +[area legend] coordinates {(NORM, 246397.38) (LAWA, 299.1845) (OIP, 23.836) (TI, 18425.906) (TPDB, 13435)};
			\addplot +[area legend] coordinates {(NORM, 25128.434) (LAWA, 299.19797) (OIP, 42.956) (TI, 39676.266) (TPDB, 87842.336)};
			\addplot  +[area legend] coordinates {(NORM, 3262059.8) (LAWA, 303.87924) (OIP, 60.52) (TI, 118737.31) (TPDB, 815180)};
			
			\legend{30000F, 100F, 10F, 5F, 1F}
		\end{axis}
    	\end{tikzpicture}}
  \caption{Performance for varying numbers of distinct facts.}
  \label{fig:groupInd}
\end{subfigure}
\caption{Robustness Tests}
\label{fig:scaleRob}
\end{figure}
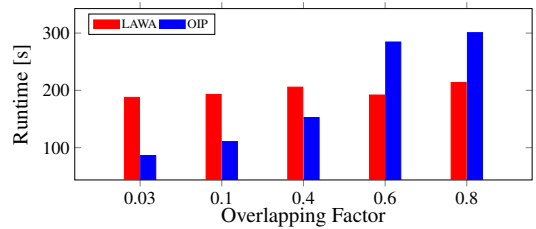
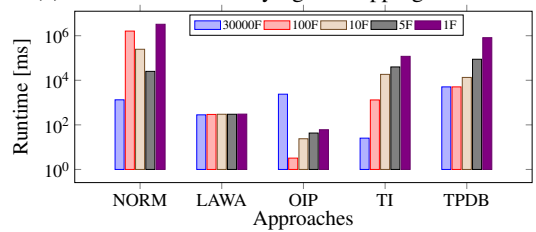

\begin{figure*}[htb] \centering
  \begin{subfigure}[b]{0.3\linewidth}\centering
    \scalebox{0.5}{
      \begin{tikzpicture}
        \begin{axis}[scale only axis,
                     height=5cm,
                     width=1.7\linewidth,
                     xlabel=Number of Input Tuples {[K]},
                     ylabel=Runtime {[ms]},
                     label style={font=\large},
                     legend style={font=\large},
                     ymax=160*10^3,
                     tick label style={font=\large},
                     mark size=4pt,
                     xmin = 20,
                     xmax = 200]
          \addplot[color=red,mark=x,mark size=6pt]
            file [skip first]{graphs/meteo_intersection_small_wal.tsv};
          \addplot[color=blue,mark=+,mark size=6pt]
            file[skip first]{graphs/meteo_intersection_small_oip.tsv};
          \addplot[color=green,mark=*]
            file[skip first] {graphs/meteo_intersection_small_ti.tsv};
          \addplot[color=cyan,mark=triangle*]
            file[skip first] {graphs/meteo_intersection_small_tpdb.tsv};
          \addplot[color=black,mark=o]
            file[skip first] {graphs/meteo_intersection_small_norm.tsv};
          \legend{LAWA,OIP,TI,TPDB,NORM}
	\end{axis}
      \end{tikzpicture}
    }
    \caption{Set Intersection}
    \label{fig:meteoIntersectionSmall}
  \end{subfigure}
  \qquad
  \begin{subfigure}[b]{0.3\linewidth}\centering
    \scalebox{0.5}{
      \begin{tikzpicture}
        \begin{axis}[scale only axis,
                     height=5cm,
                     width=1.7\linewidth,
                     xlabel=Number of Input Tuples {[K]},
                     ylabel=Runtime {[ms]},
                     label style={font=\large},
                     legend style={font=\large},
                     ymax=160*10^3,
                     tick label style={font=\large},
                     mark size=4pt,
                     xmin = 20,
                     xmax = 200]
          \addplot[color=red,mark=x,mark size=6pt]
            file [skip first] {graphs/meteo_difference_small_wal.tsv};
          \addplot[color=black,mark=o]
            file [skip first] {graphs/meteo_difference_small_norm.tsv};
	  \legend{LAWA,NORM}
	\end{axis}
      \end{tikzpicture}
    }
  \caption{Set Difference}
  \label{fig:meteoDifferenceSmall}
  \end{subfigure}%
  \qquad
  \begin{subfigure}[b]{0.3\linewidth} \centering
    \scalebox{0.5}{
      \begin{tikzpicture}
        \begin{axis}[scale only axis,
                     height=5cm,
                     width=1.7\linewidth,
                     xlabel=Number of Input Tuples {[K]},
                     ylabel=Runtime {[ms]},
                     label style={font=\large},
                     legend style={font=\large},
                     ymax=160*10^3,
                     tick label style={font=\large},
                     mark size=4pt,
                     xmin = 20,
                     xmax = 200]
	\addplot[color=red,mark=x,mark size=6pt]
          file[skip first] {graphs/meteo_union_small_wal.tsv};
	\addplot[color=cyan,mark=triangle*]
          file[skip first] {graphs/meteo_union_small_tpdb.tsv};
	\addplot[color=black,mark=o]
          file[skip first] {graphs/meteo_union_small_norm.tsv};
	\legend{LAWA,TPDB,NORM}
	\end{axis}
\end{tikzpicture}}
\caption{Set Union}
\label{fig:meteoUnionSmall}
\end{subfigure}%
 \vspace*{-0.1cm}
\caption{Meteo Swiss Dataset}
\label{fig:meteoAll}
\end{figure*}
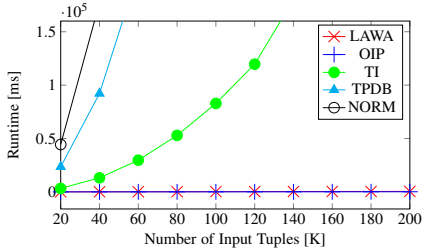
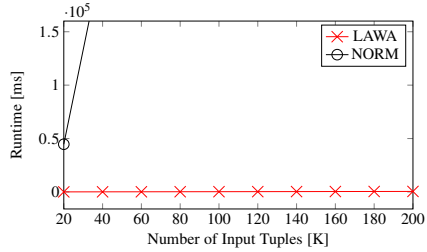
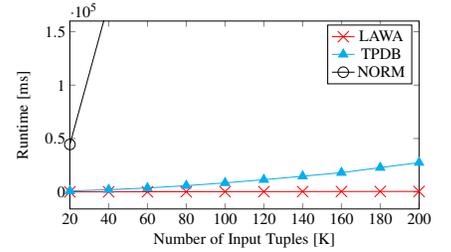

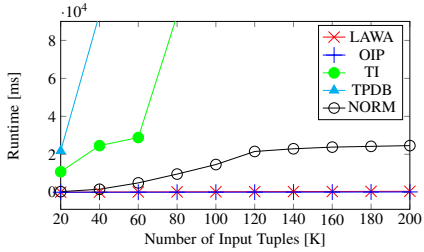
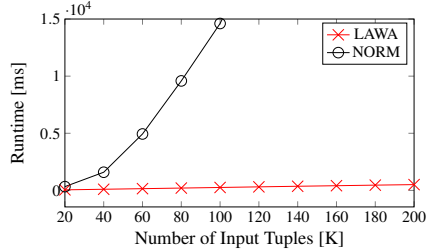
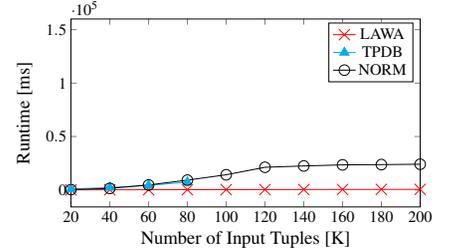
\begin{figure*}[htb]\centering
  \begin{subfigure}[b]{0.3\linewidth}\centering
    \scalebox{0.5}{
      \begin{tikzpicture}
        \begin{axis}[scale only axis,
                     height=5cm,
                     width=1.7\linewidth,
                     xlabel=Number of Input Tuples {[K]},
                     ylabel=Runtime {[ms]},
                     label style={font=\large},
                     legend style={font=\large},
                     ymax=9*10^4,
                     tick label style={font=\large},
                     mark size=4pt,
                     xmin = 20,
                     xmax = 200]
          \addplot[color=red,mark=x,mark size=6pt]
            file [skip first]{graphs/webkit_intersection_small_wal.tsv};
          \addplot[color=blue,mark=+,mark size=6pt]
            file [skip first] {graphs/webkit_intersection_small_oip.tsv};
          \addplot[color=green,mark=*]
            file [skip first] {graphs/webkit_intersection_small_ti.tsv};
          \addplot[color=cyan,mark=triangle*]
            file [skip first] {graphs/webkit_intersection_small_tpdb.tsv};
          \addplot[color=black,mark=o]
            file [skip first] {graphs/webkit_intersection_small_norm.tsv};
	\legend{LAWA,OIP, TI,TPDB,NORM}
      \end{axis}
    \end{tikzpicture}
  }
  \caption{Set Intersection}
  \label{fig:webkitIntersectionSmall}
  \end{subfigure}
\qquad
\begin{subfigure}[b]{0.3\linewidth}
\centering
\scalebox{0.5}{
\begin{tikzpicture}
\begin{axis}[
	scale only axis,
	height=5cm,
	width=1.7\linewidth,
	label style={font=\Large},
	tick label style={font=\large} ,
	xlabel=Number of Input Tuples {[K]},
        mark size=4pt,
	ylabel=Runtime {[ms]},
  	legend style={font=\large},
  	ymax=15*10^3,
  	xmin = 20, 
  	xmax = 200]	
	\addplot[color=red,mark=x,mark size=6pt]
        file [skip first] {graphs/webkit_difference_small_wal.tsv};
	\addplot[color=black,mark=o]
        file [skip first] {graphs/webkit_difference_small_norm.tsv};
	\legend{LAWA,NORM}
\end{axis}
\end{tikzpicture}}
\caption{Set Difference}
\label{fig:webkitDifferenceSmall}
\end{subfigure}%
\qquad
\begin{subfigure}[b]{0.3\linewidth}
\centering
\scalebox{0.5}{
\begin{tikzpicture}
\begin{axis}[
	scale only axis,
	height=5cm,
	width=1.7\linewidth,
	label style={font=\Large},
	tick label style={font=\large} ,
        mark size=4pt,
	xlabel=Number of Input Tuples {[K]},
	ylabel=Runtime {[ms]},
  	legend style={font=\large},
  	ymax=160*10^3,
  	xmin = 20, 
  	xmax = 200]]
 	\addplot[color=red,mark=x,mark size=6pt]
        file [skip first] {graphs/webkit_union_small_wal.tsv};
	\addplot[color=cyan,mark=triangle*]
        file [skip first] {graphs/webkit_union_small_tpdb.tsv};
	\addplot[color=black,mark=o]
        file [skip first] {graphs/webkit_union_small_norm.tsv};
	\legend{LAWA,TPDB,NORM}
	\end{axis}
\end{tikzpicture}}
\caption{Set Union}
\label{fig:webkitUnionSmall}
\end{subfigure}%
\caption{Webkit Dataset}
 \vspace{-0.5cm}
\label{fig:webkitAll}
\end{figure*}


In Fig.~\ref{fig:groupInd}, we show how the number of distinct facts
in the input relations affects the performance of LAWA and all
other approaches during a TP set intersection. The size of the dataset
is set to 60K, so that the runtimes of the approaches are comparable,
and the overlapping metric is set to 0.6. The number of facts is set
to values much less than the size of the dataset, but also to a value
that is equal to half the size of the dataset. The runtime of LAWA 
remains stable as the number of the facts included in the input tuples
decreases, whereas the performance of the other approaches 
deteriorates. OIP is an exception since, if the number of facts
becomes comparable to the number of tuples, it suffers from the
overhead of partitioning the tuples of each fact, performing the
corresponding join and merging the results. Concerning the other
approaches, TI has a better performance than LAWA but only in the case
of 30K facts. This behaviour is expected, since there is a low number
of joined pairs, thus reducing the number of required lookups. NORM's
performance improves as well when the number of facts increases, but
this approach does not scale to datasets with more than 30K tuples.
TPDB, on the other hand, appears to have diminishing improvements. 

\subsection{Real-World Datasets}

In this subsection, we compare the runtimes of TP set operations using
two real-world temporal data\-sets. The main properties of these
datasets are summarized in Table~\ref{tab:realDatasetSummary}. The
Meteo Swiss dataset\footnote{\scriptsize{Federal Office of Meteorology
    and Climatology: \url{http://www.meteoswiss.ch} (2016)}} includes
temperature predictions that have been extracted from the website of
the Swiss Federal Office of Meterology and Climatology.  The
measurements were taken at 80 different meteorological stations in
Switzerland from 2005 to 2015. Measurements are 10 minutes apart and
-- in order to produce intervals -- we merged time points whose
measurements differ by less than 0.1.  The Webkit
dataset\footnote{\scriptsize{The WebKit Open Source
    Project: \url{http://www.webkit.org} (2012)}} \cite{DignosBG14,
  PlatovICDE16, CafagnaB17} records the history of 484K files of the
SVN repository of the Webkit project over a period of 11 years at a
granularity of milliseconds. The valid times indicate the periods when
a file remained unchanged. For both datasets we produced a second
relation by shifting the intervals of the original dataset, without
modifying the lengths of the intervals. The start/end points of the
new relation were randomly chosen, following the distribution of the
original ones.

\begin{table}[htb]
\caption{Real-World Dataset Properties}
\label{tab:realDatasetSummary}
\centering
\scalebox{0.9}{
\begin{tabular}{ |  M{5cm}  !{\VRule[1.5pt]} M{1.2cm} | M{1.2cm}  | @{}m{0cm}@{}}
  \hline
  \textbf{}            &  \textbf{Meteo} &  \textbf{Webkit} & \\  \specialrule{1.2pt}{0pt}{0pt}
  Cardinality     & 10.2M  & 1.5M  &  \\[0.1cm] \hline 
  Time Range      & 347M   & 7M    &   \\ [0.1cm] \hline 
  Min. Duration   & 600    & 0.02  &   \\ [0.1cm]\hline 
  Max. Duration   & 19.3M  & 6M    & \\ [0.1cm]\hline 
  Avg. Duration   & 152M   & 1.7M  &   \\ [0.1cm]\hline 
  Num. of Facts   & 80     & 484K  & \\ [0.1cm]\hline 
  Distinct Points & 545K   & 144K  &   \\ [0.1cm]\hline 
  Max Num. of Tuples (per time point)	& 140 & 369K  &   \\[0.1cm] \hline 
  Avg Num. of Tuples (per time point) & 37   & 21  &   \\ [0.1cm]\hline 
\end{tabular}}
\end{table}

In Fig.~\ref{fig:meteoAll} and Fig.~\ref{fig:webkitAll}, we perform
TP set intersection, difference and union over two equally sized
relations created from random subsets of the initial dataset and
its shifted counterpart, respectively. The runtime of each approach
is based on the number of tuples in the input relations. In all
cases, LAWA has the best performance. All approaches perform similarly
to the synthetic dataset, with the exception of TI and NORM for the
Webkit dataset. In this dataset, the maximum number of tuples starting
or ending at a certain time point is very high, thus negatively
affecting the performance of TI that has to make pairs among all of
the tuples at a time point before it rejects the ones that do not match
the nontemporal condition. Also, the number of facts is much higher
than in the Meteo Swiss Dataset, making NORM significantly faster.


\section{Conclusions}
\label{sec:conclusion}

We proposed a novel data model that---for the first time in
the literature---unifies the two areas of temporal and
probabilistic databases under a sequenced semantics. We
defined and implemented TP set operations, which can be
supported very efficiently for a wide range of queries but
received only very little attention so far. We introduced
the lineage-aware temporal window as a mechanism to
accelerate the computation of TP set operations. Our LAWA
algorithm produces lineage-aware temporal windows that can
be filtered directly by the time of their creation based on
input lineage expressions. Using a generic window-sweeping
technique, LAWA manages to produce all output intervals,
not only for TP set intersection but also for TP set
difference and TP set union, in a scalable and predictable
manner. A thorough experimental evaluation reveals that our
implementation is robust and outperforms comparable approaches
from both temporal and probabilistic databases. As future work,
we intend to investigate both tuple correlations and support
for full relational algebra.

\bibliographystyle{IEEEtran}
\bibliography{IEEEabrv,citation}

\begin{thebibliography}{10}
\providecommand{\url}[1]{#1}
\csname url@samestyle\endcsname
\providecommand{\newblock}{\relax}
\providecommand{\bibinfo}[2]{#2}
\providecommand{\BIBentrySTDinterwordspacing}{\spaceskip=0pt\relax}
\providecommand{\BIBentryALTinterwordstretchfactor}{4}
\providecommand{\BIBentryALTinterwordspacing}{\spaceskip=\fontdimen2\font plus
\BIBentryALTinterwordstretchfactor\fontdimen3\font minus
  \fontdimen4\font\relax}
\providecommand{\BIBforeignlanguage}[2]{{%
\expandafter\ifx\csname l@#1\endcsname\relax
\typeout{** WARNING: IEEEtran.bst: No hyphenation pattern has been}%
\typeout{** loaded for the language `#1'. Using the pattern for}%
\typeout{** the default language instead.}%
\else
\language=\csname l@#1\endcsname
\fi
#2}}
\providecommand{\BIBdecl}{\relax}
\BIBdecl

\bibitem{DyllaMT13}
M.~Dylla, I.~Miliaraki, and M.~Theobald, ``A temporal-probabilistic database
  model for information extraction,'' \emph{PVLDB}, vol.~6, no.~14, pp.
  1810--1821, 2013.

\bibitem{DignosTODS16}
A.~Dign{\"o}s, M.~H. B{\"o}hlen, J.~Gamper, and C.~S. Jensen, ``{Extending the
  Kernel of a Relational DBMS with Comprehensive Support for Sequenced Temporal
  Queries},'' \emph{TODS}, vol.~41, no.~4, pp. 26:1--26:46, 2016.

\bibitem{DignosBG12}
A.~Dign{\"o}s, M.~H. B{\"o}hlen, and J.~Gamper, ``Temporal alignment,'' in
  \emph{SIGMOD}, 2012, pp. 433--444.

\bibitem{FinkOR11}
R.~Fink, D.~Olteanu, and S.~Rath, ``Providing support for full relational
  algebra in probabilistic databases,'' in \emph{{ICDE}}, 2011, pp. 315--326.

\bibitem{BohJen2009}
M.~H. B{\"o}hlen and C.~Jensen, ``{Sequenced Semantics},'' in
  \emph{{Encyclopedia of Database Systems}}.\hskip 1em plus 0.5em minus
  0.4em\relax Springer Berlin, Heidelberg, Germany, 2009, pp. 2619--2621.

\bibitem{Suciu2009}
D.~Suciu, ``{Probabilistic Databases},'' in \emph{{Encyclopedia of Database
  Systems}}.\hskip 1em plus 0.5em minus 0.4em\relax Springer Berlin,
  Heidelberg, Germany, 2009, pp. 2150--2155.

\bibitem{SuciuKochOlteanuReBook}
D.~Suciu, D.~Olteanu, R.~Christopher, and C.~Koch, \emph{Probabilistic
  Databases}, 1st~ed.\hskip 1em plus 0.5em minus 0.4em\relax Morgan \& Claypool
  Publishers, 2011.

\bibitem{papaioannou2018}
K.~Papaioannou, M.~Theobald, and M.~B{\"o}hlen, ``Set operations in
  temporal-probabilistic databases,'' in \emph{ICDE}, 2018.

\bibitem{AlKatebGCBCP13}
M.~Al{-}Kateb, A.~Ghazal, A.~Crolotte, R.~Bhashyam, J.~Chimanchode, and S.~P.
  Pakala, ``Temporal query processing in teradata,'' in \emph{EDBT/ICDT}, 2013,
  pp. 573--578.

\bibitem{lorentzos1997sql}
N.~A. Lorentzos and Y.~G. Mitsopoulos, ``{SQL extension for interval data},''
  \emph{TKDE}, vol.~9, no.~3, pp. 480--499, 1997.

\bibitem{Viqueira2007}
J.~R.~R. Viqueira and N.~A. Lorentzos, ``{SQL Extension for Spatio-temporal
  Data},'' \emph{VLDB-J}, vol.~16, no.~2, pp. 179--200, 2007.

\bibitem{toman1998point}
D.~Toman, ``{Point-based temporal extensions of SQL and their efficient
  implementation},'' in \emph{Temporal databases: research and practice}.\hskip
  1em plus 0.5em minus 0.4em\relax Springer, 1998, pp. 211--237.

\bibitem{KaufmannTI}
M.~Kaufmann, A.~A. Manjili, P.~Vagenas, P.~M. Fischer, D.~Kossmann,
  F.~F{\"{a}}rber, and N.~May, ``Timeline index: a unified data structure for
  processing queries on temporal data in {SAP} {HANA},'' in \emph{SIGMOD},
  2013, pp. 1173--1184.

\bibitem{DignosBG14}
A.~Dign\"{o}s, M.~H. B\"{o}hlen, and J.~Gamper, ``Overlap interval partition
  join,'' in \emph{SIGMOD}, 2014, pp. 1459--1470.

\bibitem{PlatovICDE16}
D.~Piatov, S.~Helmer, and A.~Dign{\"{o}}s, ``An interval join optimized for
  modern hardware,'' in \emph{ICDE}, 2016, pp. 1098--1109.

\bibitem{CafagnaB17}
F.~Cafagna and M.~H. B{\"{o}}hlen, ``Disjoint interval partitioning,''
  \emph{{VLDB} J.}, vol.~26, no.~3, pp. 447--466, 2017.

\bibitem{KaufmannVFKF13}
M.~Kaufmann, P.~Vagenas, P.~M. Fischer, D.~Kossmann, and F.~F{\"{a}}rber,
  ``Comprehensive and interactive temporal query processing with {SAP}
  {HANA},'' \emph{PVLDB}, vol.~6, no.~12, pp. 1210--1213, 2013.

\bibitem{Arge1998}
L.~Arge, O.~Procopiuc, S.~Ramaswamy, T.~Suel, and J.~S. Vitter, ``Scalable
  sweeping-based spatial join,'' in \emph{VLDB}, 1998, pp. 570--581.

\bibitem{SarmaTW08}
A.~D. Sarma, M.~Theobald, and J.~Widom, ``Exploiting lineage for confidence
  computation in uncertain and probabilistic databases,'' in \emph{ICDE}, 2008,
  pp. 1023--1032.

\bibitem{DBLP:journals/vldb/BenjellounSHTW08}
O.~Benjelloun, A.~D. Sarma, A.~Y. Halevy, M.~Theobald, and J.~Widom,
  ``Databases with uncertainty and lineage,'' \emph{{VLDB} J.}, vol.~17, pp.
  243--264, 2008.

\bibitem{FinkO16}
R.~Fink and D.~Olteanu, ``Dichotomies for queries with negation in
  probabilistic databases,'' \emph{{ACM} Trans. Database Syst.}, vol.~41, pp.
  4:1--4:47, 2016.

\bibitem{Dekhtyar2001}
A.~Dekhtyar, R.~Ross, and V.~S. Subrahmanian, ``Probabilistic temporal
  databases, i: Algebra,'' \emph{ACM Trans. Database Syst.}, vol.~26, no.~1,
  pp. 41--95, 2001.

\bibitem{Imielinski1984}
T.~Imieli\'{n}ski and W.~Lipski, Jr., ``Incomplete information in relational
  databases,'' \emph{J. ACM}, vol.~31, pp. 761--791, 1984.

\bibitem{OlteanuHK09}
D.~Olteanu, J.~Huang, and C.~Koch, ``{Sprout: Lazy vs. eager query plans for
  tuple-independent probabilistic databases},'' in \emph{ICDE}, 2009, pp.
  640--651.

\bibitem{DalviS07}
N.~N. Dalvi and D.~Suciu, ``Efficient query evaluation on probabilistic
  databases,'' \emph{{VLDB} J.}, vol.~16, no.~4, pp. 523--544, 2007.

\bibitem{DalviS12}
------, ``The dichotomy of probabilistic inference for unions of conjunctive
  queries,'' \emph{J. {ACM}}, vol.~59, no.~6, pp. 30:1--30:87, 2012.

\bibitem{olteanu2008using}
D.~Olteanu and J.~Huang, ``{Using OBDDs for efficient query evaluation on
  probabilistic databases},'' in \emph{SUM}, 2008, pp. 326--340.

\bibitem{FinkHO13}
R.~Fink, J.~Huang, and D.~Olteanu, ``Anytime approximation in probabilistic
  databases,'' \emph{{VLDB} J.}, vol.~22, no.~6, pp. 823--848, 2013.

\bibitem{FinkO11}
R.~Fink and D.~Olteanu, ``On the optimal approximation of queries using
  tractable propositional languages,'' in \emph{{ICDT}}, 2011, pp. 174--185.

\bibitem{gatterbauer2014oblivious}
W.~Gatterbauer and D.~Suciu, ``Oblivious bounds on the probability of boolean
  functions,'' \emph{TODS}, vol.~39, no.~1, p.~5, 2014.

\bibitem{GatterbauerS15}
------, ``Approximate lifted inference with probabilistic databases,''
  \emph{PVLDB}, vol.~8, no.~5, pp. 629--640, 2015.

\bibitem{OlteanuHK10}
D.~Olteanu, J.~Huang, and C.~Koch, ``Approximate confidence computation in
  probabilistic databases,'' in \emph{{ICDE}}, 2010, pp. 145--156.

\bibitem{DBLP:journals/pvldb/KhannaRT11}
S.~Khanna, S.~Roy, and V.~Tannen, ``Queries with difference on probabilistic
  databases,'' \emph{{PVLDB}}, vol.~4, no.~11, pp. 1051--1062, 2011.

\bibitem{Khayyat2015}
Z.~Khayyat, W.~Lucia, M.~Singh, M.~Ouzzani, P.~Papotti, J.-A. Quian{\'e}-Ruiz,
  N.~Tang, and P.~Kalnis, ``Lightning fast and space efficient inequality
  joins,'' \emph{PVLDB}, vol.~8, no.~13, pp. 2074--2085, 2015.

\bibitem{Allen1983}
J.~F. Allen, ``{Maintaining Knowledge About Temporal Intervals},''
  \emph{Commun. ACM}, vol.~26, no.~11, pp. 832--843, 1983.

\end{thebibliography}

\end{document}